\documentclass[journal]{IEEEtran}

\usepackage{graphicx}
\usepackage{subfigure}
\usepackage{caption2}
\usepackage{amssymb}
\usepackage{amsthm}
\usepackage{amsmath}
\usepackage{cite}
\begin{document}
%
\title{Achievable DoF Regions of MIMO Networks with Imperfect CSIT}

\author{Chenxi Hao, Borzoo Rassouli and Bruno Clerckx 
\thanks{Chenxi Hao, Borzoo Rassouli and Bruno Clerckx are with the Communication and Signal Processing group of Department of Electrical and Electronic Engineering, Imperial College London. Bruno Clerckx is also with the School of Electrical Engineering, Korea University. This work was partially supported by the Seventh Framework Programme for Research of the European Commission under grant number HARP-318489 and the EPSRC of UK, under grant EP/N015312/1.}}


\maketitle

\begin{abstract}
We focus on a two-receiver Multiple-Input-Multiple-Output (MIMO) Broadcast Channel (BC) and Interference Channel (IC) with an arbitrary number of antennas at each node. We assume an imperfect knowledge of local Channel State Information at the Transmitters, whose error decays with the Signal-to-Noise-Ratio. With such configuration, we characterize the achievable Degrees-of-Freedom (DoF) regions in both BC and IC, by proposing a Rate-Splitting (RS) approach, which divides each receiver's message into a common part and a private part. Compared to the RS scheme designed for the symmetric MIMO case, the novelties of the proposed block lie in 1) delivering additional non-ZF-precoded private symbols to the receiver with the greater number of antennas, and 2) a Space-Time implementation. These features provide more flexibilities in balancing the common-message-decodabilities at the two receivers, and fully exploit asymmetric antenna arrays. Besides, in IC, we modify the power allocation designed for the asymmetric BC based on the signal space where the two transmitted signals interfere with each other. We also derive an outer-bound for the DoF regions and show that the proposed achievable DoF regions are optimal under some antenna configurations and CSIT qualities.
\end{abstract}
\newtheorem{myprop}{Proposition}
\newtheorem{myremark}{Remark}
\newtheorem{mydef}{Definition}
\newtheorem{myassump}{Assumption}
\section{Introduction}\label{sec:Intro}
The capacity region of Multiple-Input-Multiple-Output (MIMO) Broadcast Channel (BC) and the Degrees-of-Freedom (DoF) region of a two-receiver MIMO interference channel (IC) with perfect channel state information at the transmitter side (CSIT) were fully characterized in \cite{Wein06capcityBC} and \cite{JarfarMIMOIA2pair}, respectively. However, in current wireless communication frameworks, in order to perform multiuser transmission, CSIT is a necessary condition. But guaranteeing highly-accurate CSIT is challenging due to the channel estimation error, latency and/or finite rate in the feedback/backhaul link. Hence, a more realistic and meaningful scenario is the case with imperfect CSIT, whose optimal DoF region remains unknown.

For MIMO BC and IC, imperfect CSIT knowledge results in distorted interference-nulling and causes residual interference at each receiver. This fact draws a strong similarity to the deterministic IC. When the interference is strong, conventional multi-user transmission strategies developed by treating the interference as noise yields a significant DoF loss compared to the case with perfect CSIT. However, Han-Kobayashi (HK) scheme \cite{HanKobayashi} developed for deterministic IC provides a different idea. It suggests that the DoF performance can be enhanced by decoding part or whole of the interference. This motivates the pioneering work \cite{Ges12}\footnote{Note that \cite{Ges12} aims at characterizing the DoF region of two-receiver MISO BC with a mixture of imperfect current CSIT and perfect delayed CSIT. However, one corner point of the DoF region can be achieved by RS with only imperfect current CSIT.}, where a rate-splitting (RS) scheme is designed for a two-receiver multiple-input-single-output (MISO) BC with imperfect CSIT. Particularly, each user's message is split into a common part and a private part. The private parts are transmitted via Zero-Forcing BeamForming (ZFBF) with a fraction of the total power, while the common parts of the two users are encoded into a super common message, which is multicast using the remaining power. At the receiver side, each user decodes the super common message and the desired private message. Considering that the CSIT error decays with the Signal-to-Noise-Ratio (SNR) as ${\rm SNR}^{-\alpha}$, the resultant sum DoF is $1{+}\alpha$, which outperforms $2\alpha$ achieved by treating interference as noise. The optimality of this result was shown in \cite{Davoodi14}. Moreover, subsequent works \cite{Tandon12,icc13freq,pimrc2013,Elia13,Jinyuan_evolving_misobc} studied the DoF region of two-receiver MISO BC with time-varying CSIT qualities. The sum rate analysis in the presence of quantized CSIT and a robust design of the RS scheme are investigated in \cite{RateAnalysis,HamdiRS}, respectively.

For symmetric MIMO IC where each transmitter has the same number of antennas and each receiver has the same number of antennas, the optimality of the DoF region achieved by HK scheme was found in \cite{Akuiyibo} when there is no CSIT and the number of transmit antennas is no greater than the number of receive antennas. When there is imperfect CSIT, the extension of the RS scheme designed for MISO BC to the MISO IC was reported in \cite{xinping_MISOIC}; the generalizations to the symmetric MIMO BC and IC were reported in \cite{JinyuanMIMO,Tasos}, respectively. In the context of two-receiver MIMO BC and IC with arbitrary antenna configuration, the DoF region with no CSIT was fully characterized in \cite{HJSV09} for BC and in \cite{zhunoCSIT} for IC, while the DoF region with a mixture of perfect delayed CSIT and imperfect current CSIT were found in \cite{xinping_mimo,VVICdelay} for both BC and IC. However, the characterization of the DoF region with only imperfect current CSIT remains an open problem.

Toward this, in this paper, we first design a novel RS scheme for a two-receiver MIMO BC where the transmitter and two receivers have arbitrary number of antennas. The key ingredients of the scheme lie in 1) transmitting additional private messages (apart from the ZF-precoded private messages) to the receiver with a greater number of antennas, and 2) performing a space-time transmission. These features fully exploit the spatial dimensions at the two receivers and provide more flexibilities in balancing the common-message-decoding capabilities at the two receivers. We found the achievable DoF region by calculating the power allocation that maximizes the sum DoF. The resultant achievable DoF region with imperfect CSIT smoothly connect the achievable DoF region with no CSIT and the achievable DoF region with perfect CSIT.

Second, we consider a two-receiver MIMO IC, where each node has an arbitrary number of antennas. The proposed RS transmission block inherits the key features of the RS scheme designed for the MIMO BC, but with some modifications. The modifications are motivated by a row transformation to the channel matrices. Such an operation allows us to identify the signal space where the transmitted signals interfere with each other, so as to derive a proper power allocation policy. The achievable DoF region is characterized by finding the optimal power levels that maximize the DoF of Rx2 (sum DoF of the common and private messages intended for Rx2) for a given DoF of Rx1.

Third, we also derive an outer-bound for the DoF region of MIMO BC and IC using the aligned image set proposed in \cite{Davoodi14} and the sliding window lemma proposed in \cite{Borzoo_Kuser}. Using this outer-bound and the optimal DoF region found in \cite{xinping_mimo} when there is a mixture of perfect delayed CSIT and imperfect current CSIT, we show that the optimality of the proposed achievable DoF region holds for some antenna configurations and CSIT qualities.

The rest of the paper is organized as follows. System models are defined in Section \ref{sec:SM}. In Section \ref{sec:pa}, we revisit the related works and point out the difficulties in designing RS scheme for the asymmetric MIMO case. In Section \ref{sec:MR}, we highlight our main contributions on the RS transmission block design and summarize the main results of the achievable DoF regions. Then, Section \ref{sec:BC} elaborates on the proposed schemes designed for asymmetric MIMO BC, while Section \ref{sec:IC} presents the detail of the proposed RS scheme for asymmetric MIMO IC. Section \ref{sec:conclusion} concludes the paper.

\emph{Notations}: Bold upper and lower letters denote matrices and vectors respectively. A symbol not in bold font denotes a scalar. $({\cdot})^H$, $({\cdot})^T$ and $({\cdot})^\bot$ respectively denote the Hermitian, transpose and the null space of a matrix or vector. The term $\mathbf{I}_M$ refers to an identity matrix of size $M$, while $\mathbf{0}_{M{\times}N}$ and $\mathbf{0}_M$ stand for all-zero matrices of size $M{\times}N$ and $M{\times}M$, respectively. $\mathbb{E}\left[{\cdot}\right]$ refers to the statistical expectation. $(a)^+$ stands for $\max(a{,}0)$. $f\left(P\right){\sim}{P^{B}}$ corresponds to ${\lim_{P{\to}{\infty}}}\frac{{\log_2}f\left(P\right)}{{\log_2}P}{=}B$. $\det(\mathbf{A})$ refers to the determinant of a square matrix $\mathbf{A}$. $\lfloor a\rfloor$ denotes the greatest integer that is smaller than or equal to $a$.

\section{System Model}\label{sec:SM}
In this section, we introduce the signal model, and the definitions of CSIT quality and Rate-Splitting, which are considered throughout the paper. For a two-receiver $(M{,}N_1{,}N_2)$ MIMO BC and a two-receiver $(M_1{,}M_2{,}N_1{,}N_2)$ MIMO IC, the signals received by Rx$k$, for $k{=}1{,}2$, write as
\begin{IEEEeqnarray}{rcl}
\text{\rm BC:}\quad\mathbf{y}_k&{=}&\mathbf{H}_k^H\mathbf{s}{+}\mathbf{n}_{k},\IEEEyesnumber\IEEEyessubnumber\\
\text{\rm IC:}\quad\mathbf{y}_k&{=}&\mathbf{H}_{k1}^H\mathbf{s}_1{+}\mathbf{H}_{k2}^H\mathbf{s}_2{+}\mathbf{n}_{k},\IEEEyessubnumber
\end{IEEEeqnarray}
where the transmitted signal $\mathbf{s}$ (resp. $\mathbf{s}_j$, $j{=}1{,}2$) is subject to the power constraint $\mathbb{E}[{\parallel}\mathbf{s}{\parallel}^2]{\leq}P$ (resp. $\mathbb{E}[{\parallel}\mathbf{s}_j{\parallel}^2]{\leq}P$, $j{=}1{,}2$); $\mathbf{H}_k{\in}\mathbb{C}^{M{\times}N_k}$ (resp. $\mathbf{H}_{kj}{\in}\mathbb{C}^{M_j{\times}N_k}$) denotes the channel matrix between the Tx (resp. Tx$j$, $j{=}1{,}2$) and Rx$k$, and it is drawn from a continuous distribution; $\mathbf{n}_k\stackrel{d}{\sim}\mathcal{CN}(0{,}\mathbf{I}_{N_k})$ refers to the additive white Gaussian noise vector at Rx$k$ and is independent of the channel matrices. Note that we do not restrict the channel matrices to be correlated or uncorrelated across channel uses (e.g., slot/subband), as the proposed schemes are applicable to both cases.

We consider a general setup where there is \emph{imperfect local CSIT} due to the estimation error, latency and/or the finite rate in the feedback/backhaul link. Let $\hat{\mathbf{H}}_k$ and $\hat{\mathbf{H}}_{kj}$ denote the imperfect CSIT in BC, i.e., between the Tx and Rx$k$, and the imperfect CSIT in IC, i.e., between Tx$j$ and Rx$k$, respectively, for $k{,}j{=}1{,}2$. Then, to be specific, in BC, the Tx knows $\hat{\mathbf{H}}_1$ and $\hat{\mathbf{H}}_2$, while in IC, Tx$k$ knows $\hat{\mathbf{H}}_{1k}$ and $\hat{\mathbf{H}}_{2k}$. Besides, we consider that there is \emph{perfect local CSIR}, namely Rx$k$ perfectly knows the effective channels, i.e., the multiplication of the precoders and the channel matrices $\mathbf{H}_k$ in BC and $\mathbf{H}_{k1}$, $\mathbf{H}_{k2}$ in IC, so as to decode the desired signal.

We assume that the probability density function of the channel $\mathbf{H}_k$ (resp. $\mathbf{H}_{kj}$) conditioned on the imperfect CSIT $\hat{\mathbf{H}}_k$ (resp. $\hat{\mathbf{H}}_{kj}$) exists and is bounded. According to \cite{Davoodi14}, this assumption allows us to preclude the compound setting case and it is also consistent with the assumption made in \cite{Lapidoth}, where the differential entropy of the channel matrices conditioned on the imperfect CSIT is bounded away from ${-}\infty$. In addition, similar to \cite{Davoodi14}, we require that the probability that a subset of channel coefficients takes values in any measurable set, conditioned on the available CSIT, is no more than $f_{\max}$ times the Lebesgue measure of that set. In this paper, we consider $f_{\max{,}1}{=}O(P^{\alpha_1})$ for Rx$1$ and $f_{\max{,}2}{=}O(P^{\alpha_2})$ for Rx$2$ to scale with the SNR, and term $\alpha_k$ as the CSIT quality of Rx$k$ throughout the paper. This definition is useful in deriving outer-bounds of the DoF region.

Moreover, as mentioned in \cite{Davoodi14}, this definition of the channel uncertainty can link to the cases where the CSIT error is due to channel quantization \cite{Jin06} and/or Doppler effect \cite{Ges12,Gou12,Jinyuan_evolving_misobc,xinping_mimo,Elia13,JinyuanMIMO}. In these cases, one has
\begin{IEEEeqnarray}{rcl}
\text{\rm BC:}\quad&\mathbb{E}\left[|\mathbf{h}_{k{,}i}^H\mathbf{w}_k|^2\right]&{\sim} P^{-\alpha_k},
\IEEEyesnumber\IEEEyessubnumber\\
\text{\rm IC:}\quad&\mathbb{E}\left[|\mathbf{h}_{kj{,}i}^H\mathbf{w}_{kj}|^2\right]&{\sim} P^{-\alpha_k},j{=}1{,}2{,}
\IEEEyessubnumber
\end{IEEEeqnarray}
where $\mathbf{h}_{k{,}i}$ (resp. $\mathbf{h}_{kj{,}i}$) is the $i$th column of $\mathbf{H}_k$ (resp. $\mathbf{H}_{kj}$), while $\mathbf{w}_k{\in}\mathbb{C}^{M{\times}1}$ (resp. $\mathbf{w}_{kj}{\in}\mathbb{C}^{M_j{\times}1}$) is a unit norm vector in the null space of $\hat{\mathbf{H}}_k$ (resp. $\hat{\mathbf{H}}_{kj}$), i.e., a ZF precoder. Then, if the transmitted signal $\mathbf{s}$ (resp. $\mathbf{s}_j$, $j{=}1{,}2$) contains a ZF-precoded message, the quantity $|\mathbf{h}_{k{,}i}^H\mathbf{w}_{k}|^2$ (resp. $|\mathbf{h}_{kj{,}i}^H\mathbf{w}_{kj}|^2$) represents the strength of the residual interference received at the unintended receiver. Note that this quantity is important as it is frequently used in the achievability proof in Section \ref{sec:BC} and \ref{sec:IC}. 


The CSIT qualities $\alpha_1$ and $\alpha_2$ are non-negative values. As supported by the findings in \cite{Jin06}, $\alpha_k{\geq}1$ is equivalent to perfect CSIT because the interference will be nulled within noise variance via ZFBF and the full CSIT DoF region can be achieved. $\alpha_k{=}0$ is equivalent to no CSIT because the interference terms are overheard with the same power level as the desired signal at high SNR, such that the imperfect CSIT cannot benefit the DoF when doing ZFBF. Therefore, we focus on the case $\alpha_k{\in}[0{,}1]$ henceforth.

We consider that the message intended for Rx$k$ is split into two parts, namely $m_{ck}$ and $m_{pk}$, $k{=}1{,}2$, where $m_{ck}$ is the common part that is drawn from a codebook shared by the two receivers, such that $m_{ck}$ is decodable by both receivers, while $m_{pk}$ is the private part and is to be decoded by Rx$k$ only. Note that data-sharing is not considered in IC so that $m_{ck}$ and $m_{pk}$ are transmitted only by the corresponding Tx$k$. Specifically, the encoding function for each transmitter can be expressed as
\begin{IEEEeqnarray}{rcl}
\text{\rm BC:}\,\mathbf{s}&{=}&f(m_{c1}{,}m_{c2}{,}m_{p1}{,}m_{p2}{,}\hat{\mathbf{H}}_1{,}\hat{\mathbf{H}}_2),\IEEEyesnumber\IEEEyessubnumber\\
\text{\rm IC:}\,\mathbf{s}_k&{=}&f(m_{ck}{,}m_{pk}{,}\hat{\mathbf{H}}_{1k}{,}\hat{\mathbf{H}}_{2k}),k{=}1{,}2.\IEEEyessubnumber
\end{IEEEeqnarray}

Let $R_{pk}$ denote the rate of the private message and $R_{ck}$ denote the rate of the common message, for $k{=}1{,}2$. A rate tuple $(R_{p1}{,}R_{p2}{,}R_{c1}{,}R_{c2})$ is said to be achievable if each receiver decodes the common messages $m_{c1}$, $m_{c2}$ and the desired private message with arbitrary small error probability. Then, the achievable DoF tuple is defined as $d_{ck}{\triangleq}\lim\limits_{P{\to}\infty}\frac{R_{ck}}{{\log}_2P}$ and $d_{pk}{\triangleq}\lim\limits_{P{\to}\infty}\frac{R_{pk}}{{\log}_2P}$, for $k{=}1{,}2$. The achievable DoF pair $(d_1{,}d_2)$ writes as $(d_{c1}{+}d_{p1}{,}d_{c2}{+}d_{p2})$.

Moreover, in BC, since the transmitter has the common messages of both receivers, i.e., $m_{c1}$ and $m_{c2}$, we introduce $m_c{\triangleq}(m_{c1}{,}m_{c2})$ to represent a \emph{general} common message that is jointly formed by $m_{c1}$ and $m_{c2}$, and is drawn from the message set $\left[1{:}2^{nR_c}\right]$ with $R_c{=}R_{c1}{+}R_{c2}$. Hence, there are three types of messages in BC, i.e., $m_c$, $m_{p1}$ and $m_{p2}$. The encoding function rewrites as
\begin{IEEEeqnarray}{rcl}
\text{\rm BC:}\,\mathbf{s}&{=}&f(m_c{,}m_{p1}{,}m_{p2}{,}\hat{\mathbf{H}}_1{,}\hat{\mathbf{H}}_2).
\end{IEEEeqnarray}
Then, if both receivers are able to successfully recover $m_c$ with the rate $R_c$ (resp. DoF $d_c{\triangleq}\lim\limits_{P{\to}\infty}\frac{R_{c}}{{\log}_2P}$), we can see that any rate pair $(R_{c1}{,}R_{c2})$ (resp. DoF pair $(d_{c1}{,}d_{c2})$) such that $R_{c1}{+}R_{c2}{=}R_c$ (resp. $d_{c1}{+}d_{c2}{=}d_c$) is achievable. 

\section{Prior Art}\label{sec:pa}
The RS approach gives a fundamental idea of how to enhance the DoF performance in a two-receiver MISO BC with imperfect CSIT \cite{Ges12,Gou12,pimrc2013}. Each user's message is split into a common part and a private part. The common parts are encoded into a super common message, and then the super common message is superposed on top of the ZF-precoded private messages. Specifically, for a $(2{,}1{,}1)$ MISO BC, the transmitted signal writes as
\begin{IEEEeqnarray}{rcl}
\mathbf{s}&{=}&\underbrace{\mathbf{c}}_{P{-}P^\alpha}{+}\underbrace{\mathbf{v}_1u_1}_{P^\alpha/2}{+}
\underbrace{\mathbf{v}_2u_2}_{P^\alpha/2},\label{eq:sMISO}
\end{IEEEeqnarray}
where $\mathbf{v}_k{\in}\mathbb{C}^{M{\times}1}$ is in the subspace of $\hat{\mathbf{h}}_j^\bot{,}k{,}j{=}1{,}2{,}k{\neq}j$. $u_k$ refers to the private message intended for Rx$k$ and is sent via ZFBF using a fraction of the total power, i.e., $P^\alpha$ with $\alpha{\triangleq}\min\{\alpha_1{,}\alpha_2\}$. The common message $\mathbf{c}$ is made up of the common messages intended for Rx1 and Rx2. It is transmitted using the remaining power $P{-}P^\alpha{\sim}P$. The received signal writes as
\begin{IEEEeqnarray}{rcl}
\!\!\!\!\!\!\!y_k&{=}&\underbrace{\mathbf{h}_{k}^H\mathbf{c}}_{P}{+}\underbrace{\mathbf{h}_k^H\mathbf{v}_ku_k}_{P^\alpha}{+}
\underbrace{\mathbf{h}_k^H\mathbf{v}_ju_j}_{P^{\alpha{-}\alpha_k}}{+}\underbrace{n_k}_{P^0},k{,}j{=}1{,}2{,}k{\neq}j{.}\label{eq:ykMISO}
\end{IEEEeqnarray}
We see that, due to ZFBF with imperfect CSIT, $u_j$ is received by Rx$k$ with a power smaller than the noise as $\alpha{\leq}\alpha_k$. Then, each user decodes the super common message $c$ and the desired private message sequentially using SIC. This yields $d_{p1}{=}d_{p2}{=}\alpha$ and $d_c{=}1{-}\alpha$. The DoF pairs $(1{,}\alpha)$ and $(\alpha{,}1)$ are obtained if $\mathbf{c}$ only carries information intended for Rx1 and Rx2, respectively. The achievable DoF region is specified by $d_1{+}d_2{\leq}1{+}\alpha$.

When the two receivers have different number of antennas, the transmission block design encounters following challenges.
\begin{enumerate}
  \item[i)] The capability of decoding common message at each receiver is determined by the number of receive antennas and the power allocated to the private messages. If the private messages are transmitted with equal power, the achievable DoF of the common message is limited by the receiver with a smaller number of antennas, i.e., Rx1. This results in the fact that the achievable DoF of Rx2 (contributed by its private message and common message) is always smaller than $N_2$.
  \item[ii)] Concern i) can be solved to some extent by employing unequal power allocation to the private messages. However, if the transmission block is designed with ZF-precoded private messages plus common message multicasting, the spatial dimension at Rx2 cannot be fully exploited under some circumstances. Let us consider a $(4{,}2{,}3)$ MIMO BC where $1$ ZF-precoded private symbol is transmitted to Rx1 using power $P^{A_1}$, $2$ private symbols are transmitted to Rx2 using power $P^{A_2}$, while the remaining power is used to multicast the common message. Given this transmission block, when $\alpha_1{=}\alpha_2{=}0$, the maximum achievable sum DoF is $2$ by choosing $A_1{=}1$  and $A_2{=}0$. However, in this case, one can achieve sum DoF $3$ by transmitting three private symbols to Rx2 using full power without the need of ZFBF. Hence, the RS scheme designed for the asymmetric antenna setting should align with the case where transmitting non-ZF-precoded private symbols is beneficial to the DoF performance.
\end{enumerate}
These two concerns apply for both asymmetric MIMO BC and IC.

As studied in \cite{xinping_mimo}, the above concerns resulted by asymmetric antenna setting can be solved when there is a mixture of the imperfect current CSIT and perfect delayed CSIT. The achievable scheme contains non-ZF-precoded private symbols, thus causing some level of overheard interference at each receiver. Then, with a Block-Markov implementation and backward decoding, each receiver is able to to 1) cancel the interference that is overheard in the previous slot, and 2) have the side information of its desired private messages received by the other receiver. By doing so, each receiver obtains an $(N_1{+}N_2)$-dimensional observation of the transmitted signals, which balances the decoding capabilities at the two receivers. Nonetheless, when there is only imperfect current CSIT, it is unable to exchange the side information so that each receiver has to perform the decoding process using its own received signal. This leads to the emergence of designing a novel RS transmission block with proper power allocation policy that solve the above concerns. 

\section{Main Contributions and Results}\label{sec:MR}
\subsection{Key ingredients of the proposed RS scheme}
In this part, we highlight the key ingredients that constitute the novel RS transmission blocks designed for asymmetric MIMO BC and IC.
\subsubsection{Additional non-ZF-precoded private symbols}
To address the issues mentioned in the previous section, we propose an RS transmission block by allocating unequal power to the private messages, and by transmitting additional non-ZF-precoded private symbols to Rx2. These features provide more flexibility in balancing the capability of decoding common messages at the two receivers, and exploit the larger antenna array at Rx2. Specifically, for a $(4{,}2{,}3)$ MIMO BC, the RS scheme consists of $1$ ZF-precoded private symbols to Rx1 allocated with power $P^{A_1}$, $2$ ZF-precoded private symbols to Rx2 allocated with power $P^{A_2}$, $1$ non-ZF-precoded private symbols to Rx2 allocated with power $P^{(A_2{-}\alpha_1)^+}$, while the common messages are multicast with the remaining power. As we will see later on, when the CSIT quality of Rx1 is not sufficiently good, choosing $A_2{>}\alpha_1$ is beneficial to the sum DoF, though it results in some level of interference at Rx$1$. This is because the private message spans $3$ dimensions at Rx2, while the interference at Rx1 spans only $2$ dimensions. In the extreme case of $\alpha_1{=}\alpha_2{=}0$, by choosing the power exponents $(A_1{,}A_2){=}(0{,}1)$, the transmitted signal consists of three private symbols intended for Rx2. This yields the sum DoF $3$, which is consistent with the maximum sum DoF with no CSIT.

In contrast, the RS scheme designed for the symmetric case \cite{Ges12} has equal power allocation and no additional private symbols is transmitted. This is because unequal power allocation and delivering non-ZF-precoded private symbols are useless in enhancing the sum DoF when the two receivers have the same number of antennas. In the scheme proposed in \cite{xinping_mimo}, the transmitter delivers non-ZF-precoded private symbols to both receivers. This feature is useful because the overheard interference at each receiver can be exploited as side information when there is perfect delayed CSIT.

\subsubsection{Space-time transmission}
When the CSIT quality of Rx1 is not sufficiently good, we perform a space-time transmission using the proposed transmission block. Specifically, we employ power exponents $(A_1{,}A_2){=}(\alpha_2{,}\alpha_1)$ for a fraction of the total time slots, while employ the power exponents $(A_1{,}A_2){=}(\alpha_2{,}1)$ for the rest of the time. Since choosing $A_2{>}\alpha_1$ is beneficial to the sum DoF when CSIT quality of Rx1 is not sufficiently good, the proposed space-time transmission is carried out to fully exploit the spatial dimension at Rx2.

In contrast, when the CSIT qualities are fixed across the time line, the RS scheme designed for the symmetric case \cite{Ges12} does not employ space-time transmission. This is because choosing power exponents greater than the CSIT quality does not provide sum DoF gain. Besides, the Block-Markov implementation proposed in \cite{xinping_mimo} also spans the time-domain, but it requires perfect delayed CSIT to perform a sequential backward decoding. However, in our space-time implementation, a joint decoding is performed focusing on the aggregate received signals, and only current imperfect CSIT is used.

\subsubsection{Interference space identification}\label{sec:int_space}
We characterize the asymmetric MIMO IC into two cases. Case I has the antenna configuration $M_1{\geq}N_2$ (As a reinder, we consider $M_2{\geq}N_1$ and $N_2{\geq}N_1$). This setting yields a similar scenario to BC because the subspace spanned by the desired signal is completely overlapped with the subspace spanned by the interference signal. Accordingly, we propose an RS scheme by inheriting the key features, i.e., transmitting additional non-ZF-precoded private symbols and space-time implementation, of the RS scheme designed for the asymmetric MIMO BC.

Case II has the antenna configuration $M_1{\leq}N_2$. In this case, no ZF-precoded private symbols is delivered to Rx1. Moreover, by performing a row transformation to the channel matrices, we learn that at Rx2, the subspace spanned by the desired signal is \emph{partially} overlapped with the signal sent by Tx1. Then, since the private symbols lying in the non-overlapping part do not impact the common-message-decodability at Rx2, we modify the RS scheme designed for the Case I by allocating different power exponents to the private symbols that are overlapped with the signal sent by Tx1 and the private symbols that are not overlapped with the signal sent by Tx1.

\subsection{Main Results on Achievable DoF Regions}
We state the achievable DoF regions as follows.
\begin{myprop}\label{prop:BC}
For a $(M{,}N_1{,}N_2)$ MIMO BC, supposing $N_1{\leq}N_2$, an achievable DoF region with imperfect CSIT is characterized by \eqref{eq:BCRegion} at the top of the next page, where $\alpha_{0{,}BC}$ and $\Phi_{BC}$ are defined in \eqref{eq:BCalpha0} and \eqref{eq:BCphi}, respectively.
\begin{figure*}
{\small\begin{IEEEeqnarray}{rrl}\label{eq:BCRegion}
L_0:&d_1{\leq}&\min\{M{,}N_1\},\IEEEyesnumber\IEEEyessubnumber\\
L_0^\prime: &d_2{\leq}&\min\{M{,}N_2\},\IEEEyessubnumber\\
L_1:&d_1{+}d_2{\leq}&\min\{M{,}N_2\}{+}\left[\min\{M{,}N_1{+}N_2\}{-}\min\{M{,}N_2\}\right]\alpha_{0{,}BC},\IEEEyessubnumber\label{eq:BCsum}\\
L_2:&\frac{d_1}{\min\{M{,}N_1\}}{+}\frac{d_2}{\min\{M{,}N_2\}}{\leq}&1{+}\frac{\min\{M{,}N_1{+}N_2\}{-}\min\{M{,}N_1\}}{\min\{M{,}N_2\}}\alpha_1,
\IEEEyessubnumber\label{eq:BCwsum}
\end{IEEEeqnarray}}
{\small\begin{IEEEeqnarray}{rcl}
\alpha_{0{,}BC}&{=}&\left\{\begin{array}{ll}\alpha_2&\text{\rm if }\Phi_{BC}{\leq}0\\
\alpha_2{-}\frac{\Phi_{BC}}{\min\{M{,}N_1{+}N_2\}{-}\min\{M{,}N_1\}}&\text{\rm Else if }\alpha_1{\geq}1{-}\alpha_2;\\
\frac{\alpha_1\alpha_2[\min\{M{,}N_1{+}N_2\}{-}\min\{M{,}N_2\}]}{\left[\min\{M{,}N_2\}{-}\min\{M{,}N_1\}\right](1{-}\alpha_1){+}
\left[\min\{M{,}N_1{+}N_2\}{-}\min\{M{,}N_2\}\right]\alpha_2}&\text{\rm Else if }\alpha_1{\leq}1{-}\alpha_2.\end{array}\right.\label{eq:BCalpha0}\\
\Phi_{BC}&{\triangleq}&\min\{M{,}N_2\}{-}\min\{M{,}N_1\}{+}[\min\{M{,}N_1{+}N_2\}{-}\min\{M{,}N_2\}]\alpha_2{-}\nonumber\\&&
[\min\{M{,}N_1{+}N_2\}{-}\min\{M{,}N_1\}]\alpha_1.\label{eq:BCphi}
\end{IEEEeqnarray}}
\hrulefill
\end{figure*}
\end{myprop}

\begin{figure}[!t]
\renewcommand{\captionfont}{\small}
\captionstyle{center}
\centering
\subfigure[$\Phi_{BC}{\leq}0$]{
                \centering
                \includegraphics[width=0.24\textwidth,height=3cm]{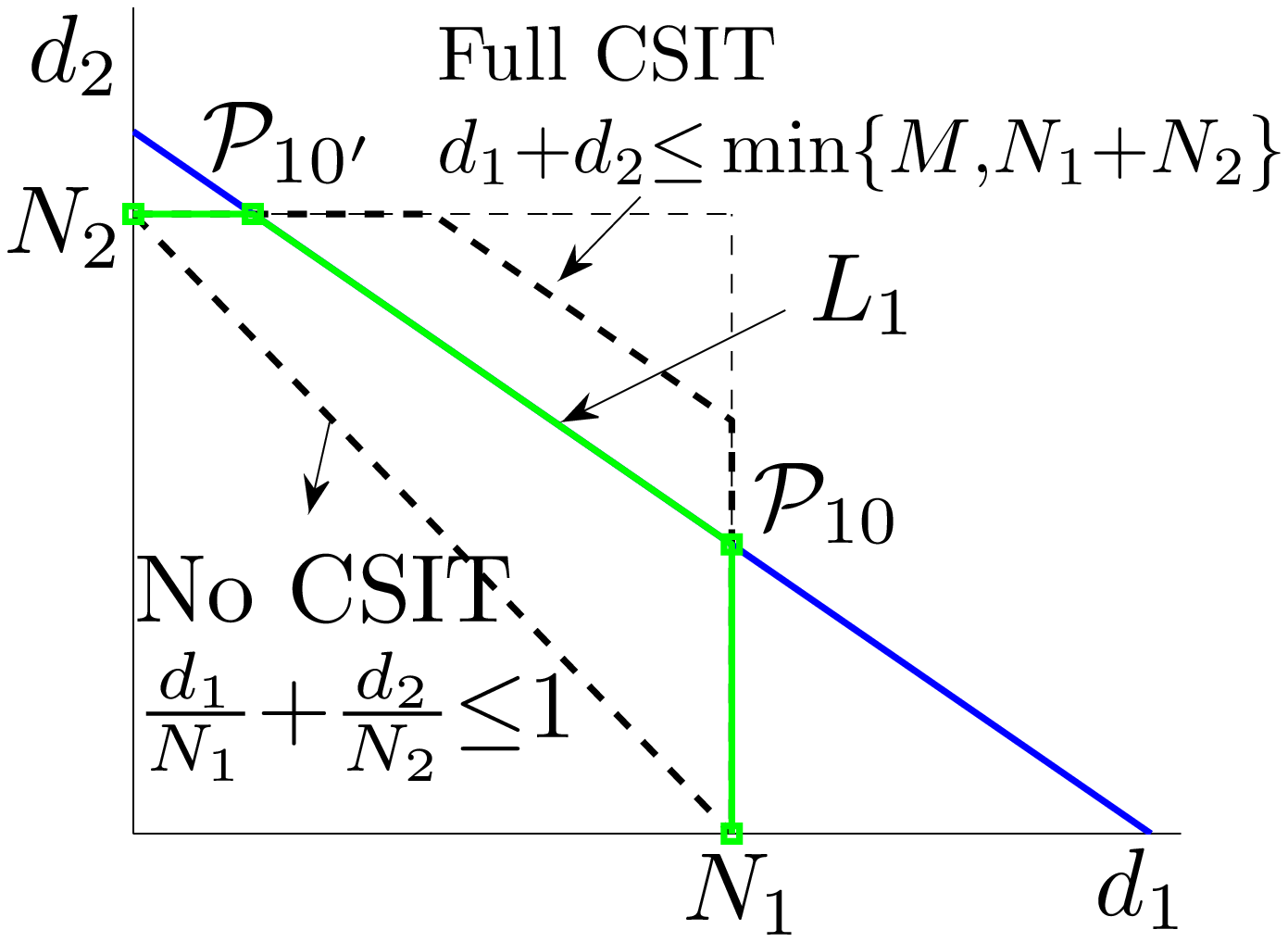}
                \label{fig:BCcase1}
        }
\subfigure[$\Phi_{BC}{\geq}0$]{
                \centering
                \includegraphics[width=0.24\textwidth,height=3cm]{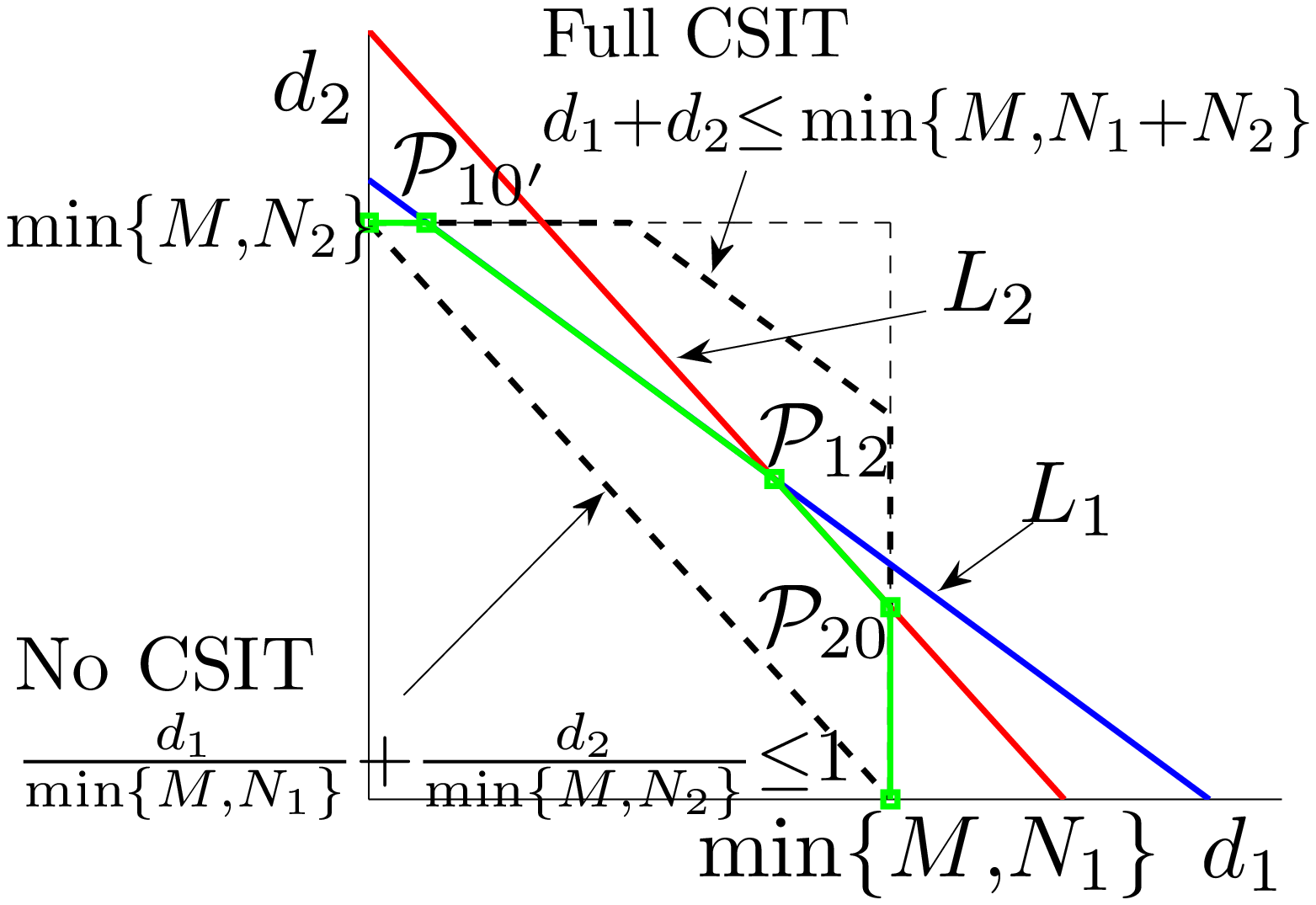}
                \label{fig:BCcase2}
        }
\caption{Achievable DoF region of $(M{,}N_1{,}N_2)$ MIMO BC.}\label{fig:BCregion}
\end{figure}

Figure \ref{fig:BCregion} illustrates the DoF region stated in Proposition \ref{prop:BC}, where $\mathcal{P}_{ij}$ denotes the intersection of line $L_i$ and $L_j$. When $\alpha_1$ is large enough such that $\Phi_{BC}{\leq}0$, the weighted-sum constraint, i.e.,\eqref{eq:BCwsum}, becomes inactive and the DoF region is formed by $\mathcal{P}_{10}$ and $\mathcal{P}_{10^\prime}$. Moreover, the DoF region with perfect CSIT and no CSIT can be reached with $\alpha_1{=}\alpha_2{=}1$ and $\alpha_1{=}0$ (${\forall}\alpha_2{\in}[0{,}1]$), respectively. When $N_1{=}N_2$, $L_1$ and $L_2$ in Proposition \ref{prop:BC} boil down to the the sum DoF constraint in the symmetric antenna case \cite{JinyuanMIMO}.

For a general $(M_1{,}M_2{,}N_1{,}N_2)$ MIMO IC, as explained in Section \ref{sec:int_space}, we categorize the antenna configurations as Case I with $M_1{\geq}N_2$ and Case II with $M_1{\leq}N_2$. The antenna configuration in Case I yields a similar scenario as BC, while the antenna configuration in Case II implies a different scenario where Tx1 is not able to perform ZFBF, and in the received signals, some messages of Rx2 do not align with the messages intended for Rx1. Due to these facts, the transmission schemes are designed differently in these two cases and lead to different achievable DoF regions.

\begin{myprop}\label{prop:IC1}
For a $(M_1{,}M_2{,}N_1{,}N_2)$ MIMO IC of Case I, an achievable DoF region with imperfect CSIT is characterized by \eqref{eq:ICRegion1} at the top of the next page, where $\alpha_{0{,}IC}$ and $\Phi_{IC}$ are defined in \eqref{eq:ICalpha0} and \eqref{eq:ICphi}, respectively.
\begin{figure*}
{\small\begin{IEEEeqnarray}{rrl}\label{eq:ICRegion1}
L_0:&d_1{\leq}&N_1,\IEEEyesnumber\IEEEyessubnumber\\
L_0^\prime:&d_2{\leq}&\min\{M_2{,}N_2\},\IEEEyessubnumber\\
L_1:&d_1{+}d_2{\leq}&\min\{M_2{,}N_2\}{+}\left[\min\{M_1{,}N_1{+}N_2\}{-}N_2\right]\alpha_{0{,}IC},\IEEEyessubnumber\label{eq:L1}\\
L_2:&\frac{d_1}{N_1}{+}\frac{d_2}{\min\{M_2{,}N_2\}}{\leq}&
1{+}\frac{\left[\min\{M_2{,}N_1{+}N_2\}{-}N_1\right]\alpha_1}{\min\{M_2{,}N_2\}},
\IEEEyessubnumber\label{eq:ICL2}
\end{IEEEeqnarray}}
{\small \begin{IEEEeqnarray}{rcl}
\alpha_{0{,}IC}&{=}&\left\{\begin{array}{ll}0&\text{\rm If }M_2{\leq}N_2\\
\alpha_2&\text{\rm Else if }\Phi_{IC}{\leq}0\\
\frac{\min\{M_2{,}N_1{+}N_2\}{-}N_2}{\min\{M_1{,}N_1{+}N_2\}{-}N_2}\alpha_1&\text{\rm Else if }\frac{\min\{M_1{,}N_1{+}N_2\}{-}\min\{M_2{,}N_1{+}N_2\}}{\min\{M_1{,}N_1{+}N_2\}{-}N_2}\alpha_1{\geq}1{-}\alpha_2\\
\alpha_2{-}\frac{\Phi_{IC}}{\min\{M_1{,}N_1{+}N_2\}{-}N_1}&\text{\rm Else if }\alpha_1{\geq}1{-}\alpha_2\\
\frac{\alpha_1\alpha_2\left[\min\{M_2{,}N_1{+}N_2\}{-}N_2\right]}{(N_2{-}N_1)(1{-}\alpha_1){+}(\min\{M_1{,}N_1{+}N_2\}{-}N_2)\alpha_2}&\text{\rm Else if }\alpha_1{\leq}1{-}\alpha_2
\end{array}\right.\label{eq:ICalpha0}\\
\Phi_{IC}&{=}&N_2{-}N_1{+}[\min\{M_1{,}N_1{+}N_2\}{-}N_2]\alpha_2{-}[\min\{M_2{,}N_1{+}N_2\}{-}N_1]\alpha_1.\label{eq:ICphi}
\end{IEEEeqnarray}}
\hrulefill
\end{figure*}
\end{myprop}

\begin{myprop}\label{prop:IC2}
For a $(M_1{,}M_2{,}N_1{,}N_2)$ MIMO IC of Case II, an achievable DoF region with imperfect CSIT is characterized by \eqref{eq:ICRegion2} at the top of the next page, where where $\mu_2{\triangleq}\min\{M_1{,}\min\{M_2{,}N_1{+}N_2\}{-}N_2{+}M_1{-}N_1^\prime\}$, $N_1^{\prime\prime}{\triangleq}\max\{M_1{,}N_1\}$ and $N_k^\prime{\triangleq}\min\{M_k{,}N_k\}{,}k{=}1{,}2$.
\begin{figure*}
{\small \begin{IEEEeqnarray}{rrl}\label{eq:ICRegion2}
L_0:&d_1{\leq}&N_1^\prime,\IEEEyesnumber\IEEEyessubnumber\\
L_0^\prime:&d_2{\leq}&N_2^\prime,\IEEEyessubnumber\\
L_1:&d_1{+}d_2{\leq}&N_2^\prime,\IEEEyessubnumber\label{eq:L1_2}\\
L_2:&\frac{d_1}{N_1^\prime}{+}\frac{d_2}{N_2^\prime{-}N_1{+}N_1^\prime}{\leq}&
\frac{N_2^\prime{+}\left[\min\{M_2{,}N_1{+}N_2\}{-}N_1\right]\alpha_1}{N_2^\prime{-}N_1{+}N_1^\prime},
\IEEEyessubnumber\label{eq:L2_2}\\
&(d_1{,}d_2)\,\text{\rm subject to }L_3,&\quad\text{\rm If }M_2{\geq}N_2{,}N_1{+}M_1{\leq}N_2\nonumber\\
&(d_1{,}d_2)\,\text{\rm subject to }L_4{,}L_5,&\quad\text{\rm If }M_2{\geq}N_2{,}N_1{+}M_1{\geq}N_2\nonumber
\end{IEEEeqnarray}
\begin{IEEEeqnarray}{rrl}
L_3:&\frac{d_1}{N_1^\prime}{+}\frac{d_2}{N_2{-}N_1^{\prime\prime}{+}N_1^\prime}{\leq}&
\frac{N_2{+}\left[N_2{-}N_1^{\prime\prime}\right]\alpha_1}{N_2{-}N_1^{\prime\prime}{+}N_1^\prime},
\IEEEyessubnumber\label{eq:L3}\\
L_4:&\frac{d_1}{M_1}{+}\frac{d_2}{N_2{-}N_1{+}M_1}{\leq}&\frac{N_2}{N_2{-}N_1{+}M_1}{+}
\left[\frac{N_2{-}N_1^{\prime\prime}}{N_2{-}N_1{+}M_1}{+}\frac{\mu_2(M_1{+}N_1{-}N_2)}{M_1(N_2{-}N_1{+}M_1)}\right]\alpha_1,
\IEEEyessubnumber\label{eq:L4}\\
L_5:&d_1{+}\frac{d_2}{2}{\leq}&\frac{1}{2}(M_1{+}N_1{+}(N_2{-}N_1^{\prime\prime})\alpha_1),
\IEEEyessubnumber\label{eq:L5}
\end{IEEEeqnarray}}
\hrulefill
\end{figure*}
\end{myprop}
\begin{table*}[t]
\renewcommand{\arraystretch}{1.3}
\vspace{.6em}
\centering
\begin{tabular}{|l|c|c|c|}
\hline
Conditions & Active Constraints & Corner Points & Optimality\\
\hline
Case I.1: $M_1{\geq}N_2${,}$M_2{\leq}N_2$ & $L_1$, $L_2$ & $\mathcal{P}_{20}$, $\mathcal{P}_{12}$, $\mathcal{P}_{10^\prime}$ & Yes\\
\hline
Case I.2: $M_1{\geq}N_2${,}$M_2{\geq}N_2$ & & &\\
$\quad$ If $\Phi_{IC}{\leq}0$: & $L_1$ & $\mathcal{P}_{10}$, $\mathcal{P}_{10^\prime}$ & Yes\\
$\quad$ If $\Phi_{IC}{\geq}0$: & $L_1$, $L_2$ & $\mathcal{P}_{20}$, $\mathcal{P}_{12}$, $\mathcal{P}_{10^\prime}$ & Unknown\\
\hline
Case II.1: $M_1{\leq}N_2${,}$M_2{\leq}N_2$ & $L_1$, $L_2$ & $\mathcal{P}_{20}$, $\mathcal{P}_{12}$, $\mathcal{P}_{10^\prime}$ & Yes, if $N_1{\leq}M_1{\leq}N_2$\\
\hline
Case II.2.a: $M_1{\leq}N_2${,}$M_2{\geq}N_2$ and $M_1{+}N_1{\leq}N_2$ & & & Unknown\\
$\quad$ If $\alpha_1{\leq}\frac{M_1{-}N_1^\prime}{\mu_2}$ & $L_1$, $L_2$, $L_3$ & $\mathcal{P}_{20}$, $\mathcal{P}_{23}$, $\mathcal{P}_{13}$, $\mathcal{P}_{10^\prime}$ & \\
$\quad$ If $\alpha_1{\geq}\frac{M_1{-}N_1^\prime}{\mu_2}$ & $L_1$, $L_3$ & $\mathcal{P}_{30}$, $\mathcal{P}_{13}$, $\mathcal{P}_{10^\prime}$ & \\
\hline
Case II.2.b: $M_1{\leq}N_2${,}$M_2{\geq}N_2$ and $M_1{+}N_1{\geq}N_2$ & & & Unknown\\
$\quad$ If $\alpha_1{\leq}\frac{M_1{-}N_1^\prime}{\mu_2}$ & $L_1$, $L_2$, $L_4$ & $\mathcal{P}_{20}$, $\mathcal{P}_{24}$, $\mathcal{P}_{14}$, $\mathcal{P}_{10^\prime}$ &\\
$\quad$ If $\frac{M_1{-}N_1^\prime}{\mu_2}{\leq}\alpha_1{\leq}\frac{N_2{-}N_1}{\mu_2{+}N_2{-}N_1^{\prime\prime}}$ & $L_1$, $L_4$, $L_5$ & $\mathcal{P}_{50}$, $\mathcal{P}_{54}$, $\mathcal{P}_{14}$, $\mathcal{P}_{10^\prime}$ &\\
$\quad$ If $\alpha_1{\geq}\frac{N_2{-}N_1}{\mu_2{+}N_2{-}N_1^{\prime\prime}}$ & $L_1$, $L_5$ & $\mathcal{P}_{50}$, $\mathcal{P}_{15}$, $\mathcal{P}_{10^\prime}$ &\\
\hline
\end{tabular}
\caption{The DoF regions of $(M_1{,}M_2{,}N_1{,}N_2)$ MIMO IC: active constraints, corner points and optimality.}\label{tab:ICregions}
\vspace{-0mm}
\end{table*}
\begin{figure}[!t]
\renewcommand{\captionfont}{\small}
\captionstyle{center}
\centering
        \subfigure[Case I.1 and II.1: $M_2{\leq}N_2$]{
                \centering
                \includegraphics[width=0.2\textwidth,height=2.5cm]{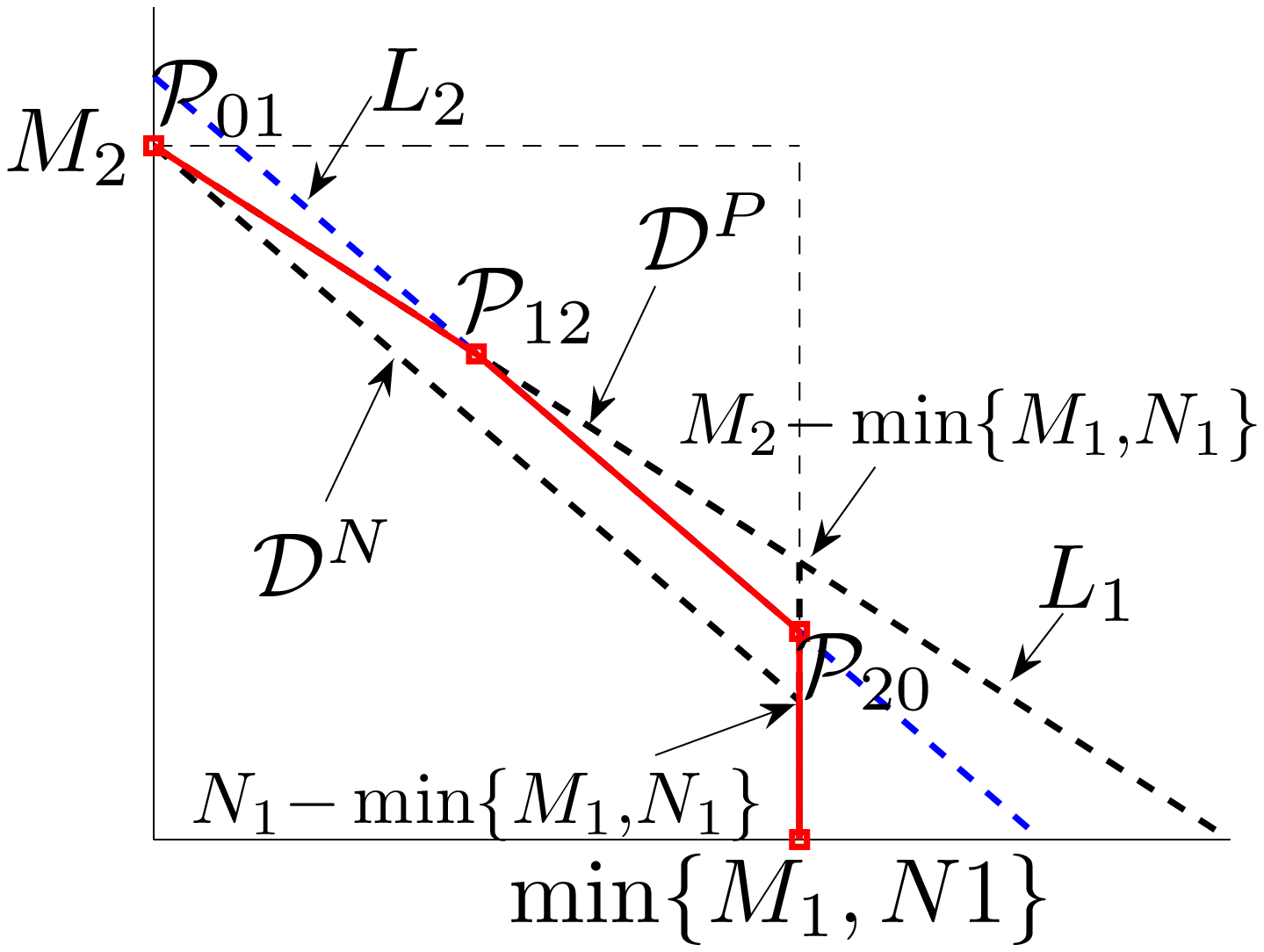}
                \label{fig:M2leqN2}
        }
        \subfigure[Case I.2 and $\Phi_{IC}{\leq}0$, where $L_2$ inactive]{
                \centering
                \includegraphics[width=0.2\textwidth,height=2.5cm]{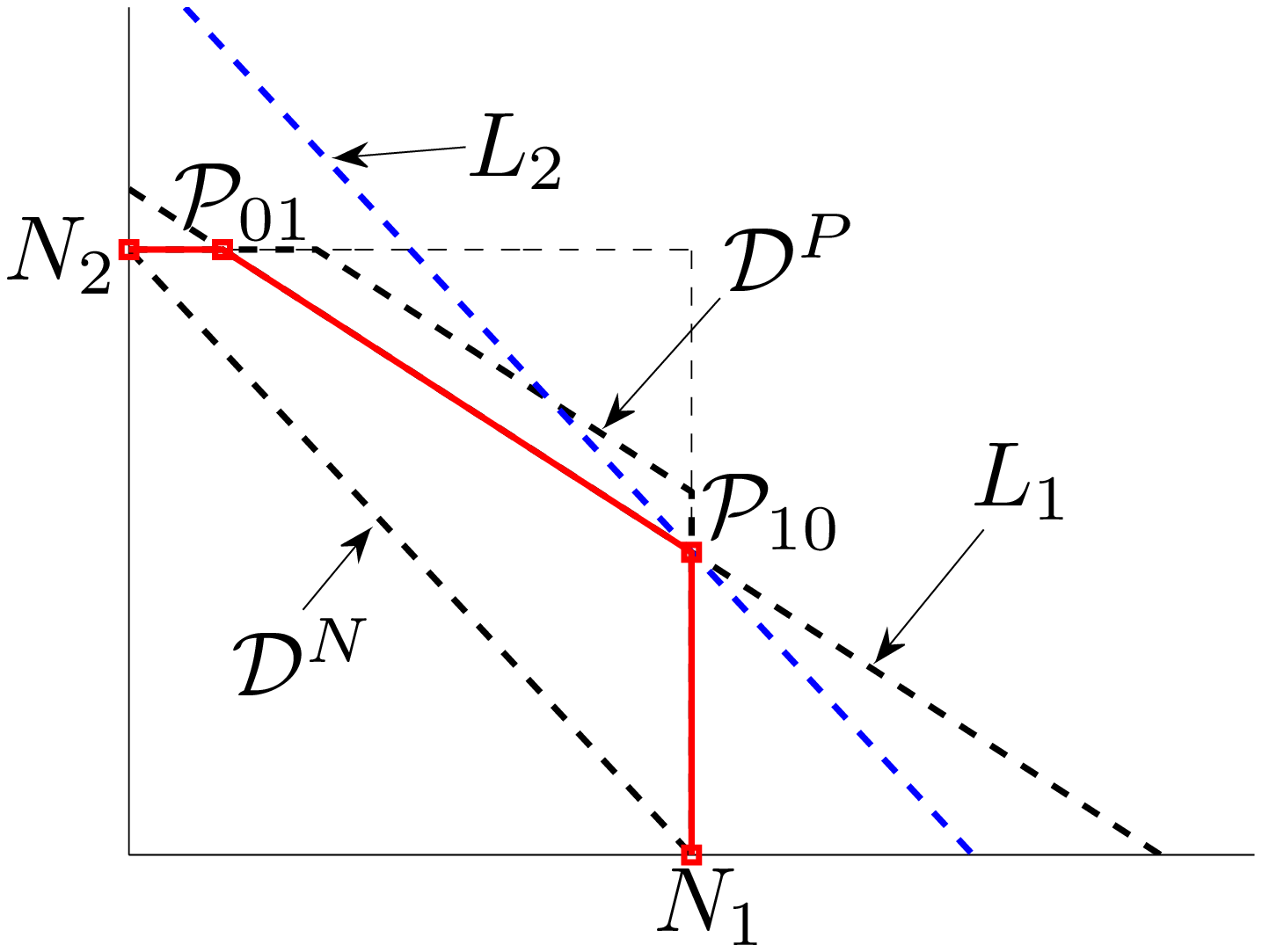}
                \label{fig:I2b}
        }
        \subfigure[Case I.2 and $\Phi_{IC}{\geq}0$]{
                \centering
                \includegraphics[width=0.2\textwidth,height=2.5cm]{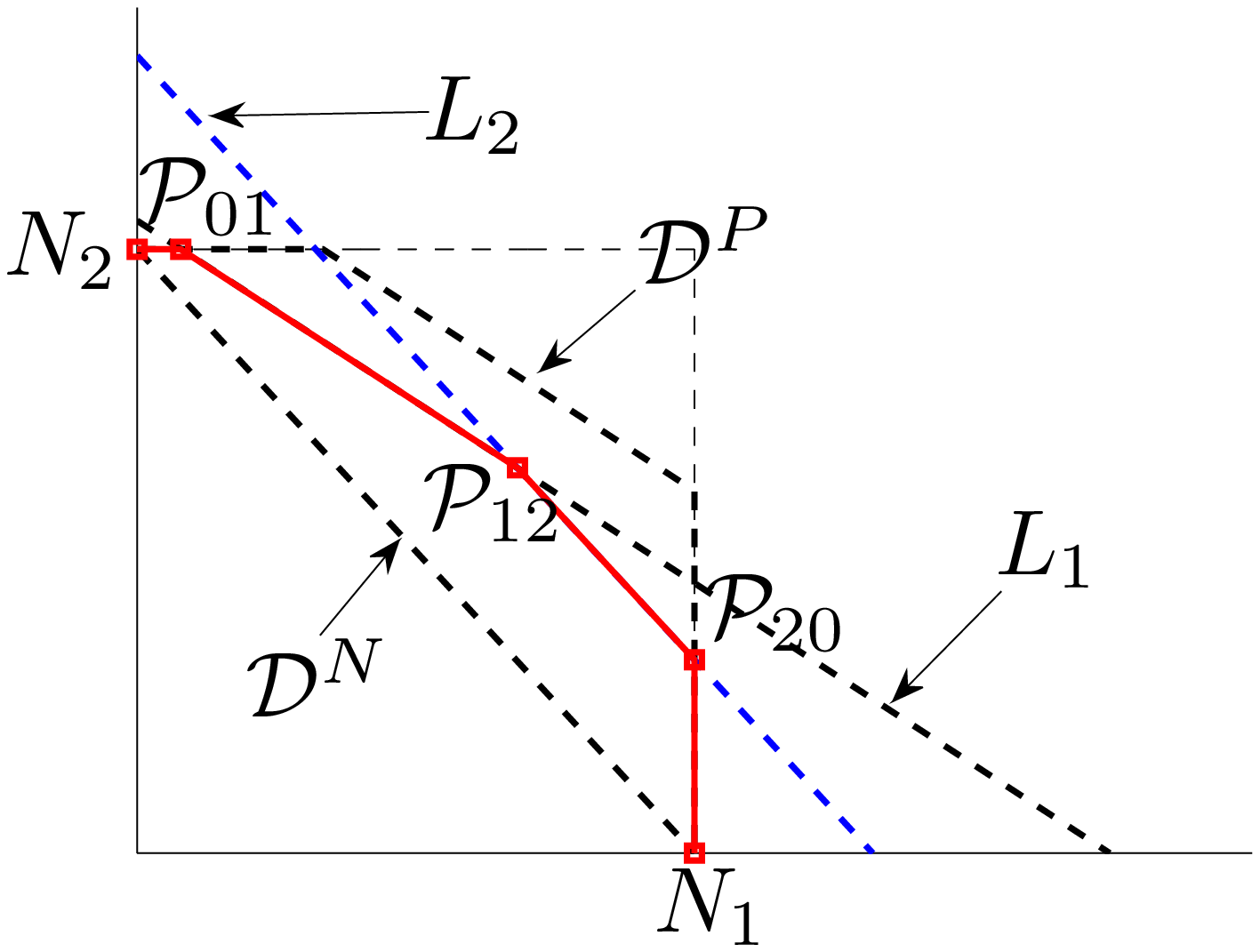}
                \label{fig:I2a}
        }
        \subfigure[Case II.2.a: If $\alpha_1{\leq}\frac{M_1{-}N_1^\prime}{\mu_2}$ (only valid when $N_1{\leq}M_1$)]{
                \centering
                \includegraphics[width=0.2\textwidth,height=2.5cm]{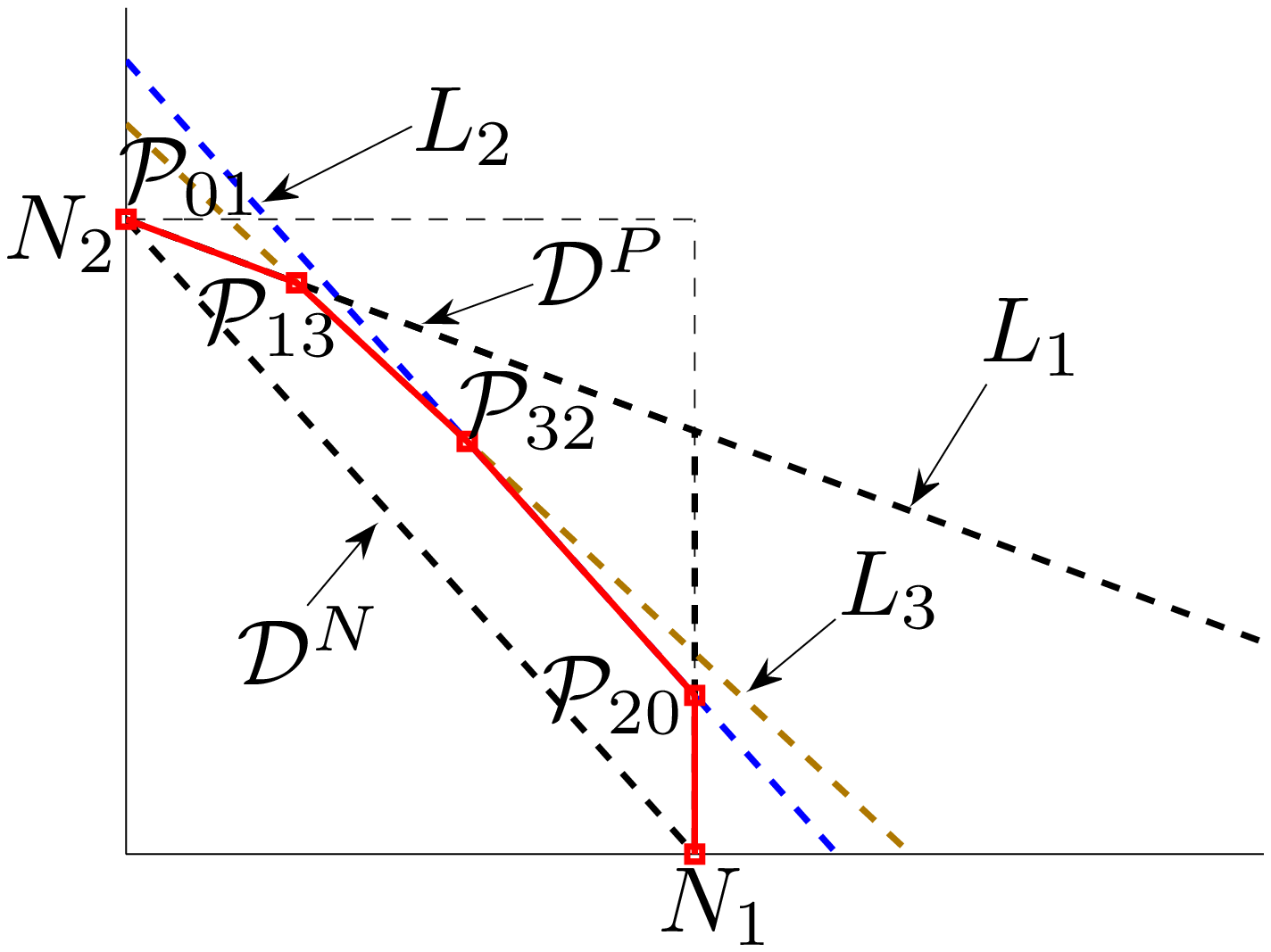}
                \label{fig:II2a1}
        }\\
        \subfigure[Case II.2.a: If $\alpha_1{\geq}\frac{M_1{-}N_1^\prime}{\mu_2}$, $L_2$ inactive]{
                \centering
                \includegraphics[width=0.2\textwidth,height=2.5cm]{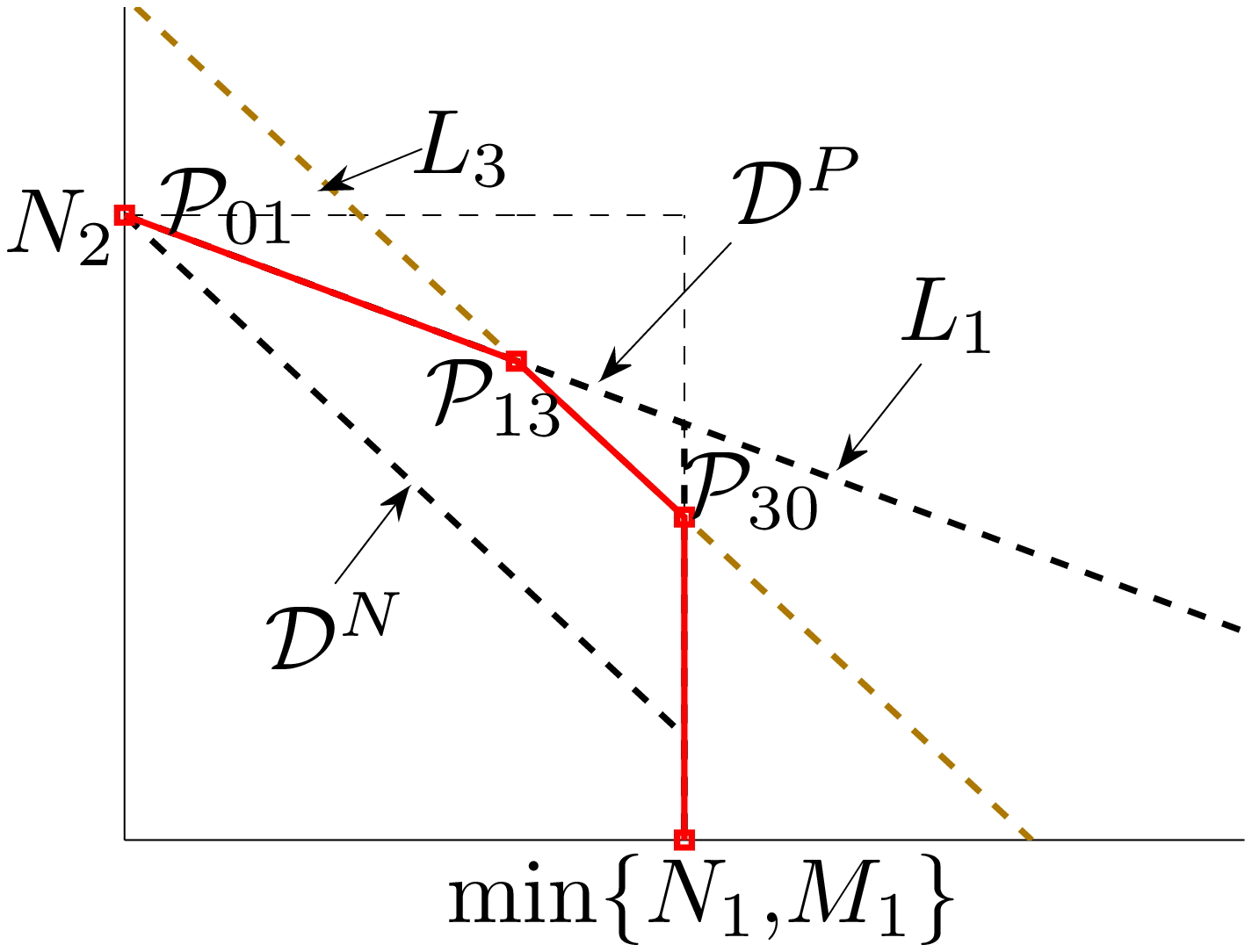}
                \label{fig:II2a2}
        }
        \subfigure[Case II.2.b: If $\alpha_1{\leq}\frac{M_1{-}N_1^\prime}{\mu_2}$ (only valid when $N_1{\leq}M_1$), $L_5$ inactive]{
                \centering
                \includegraphics[width=0.2\textwidth,height=2.5cm]{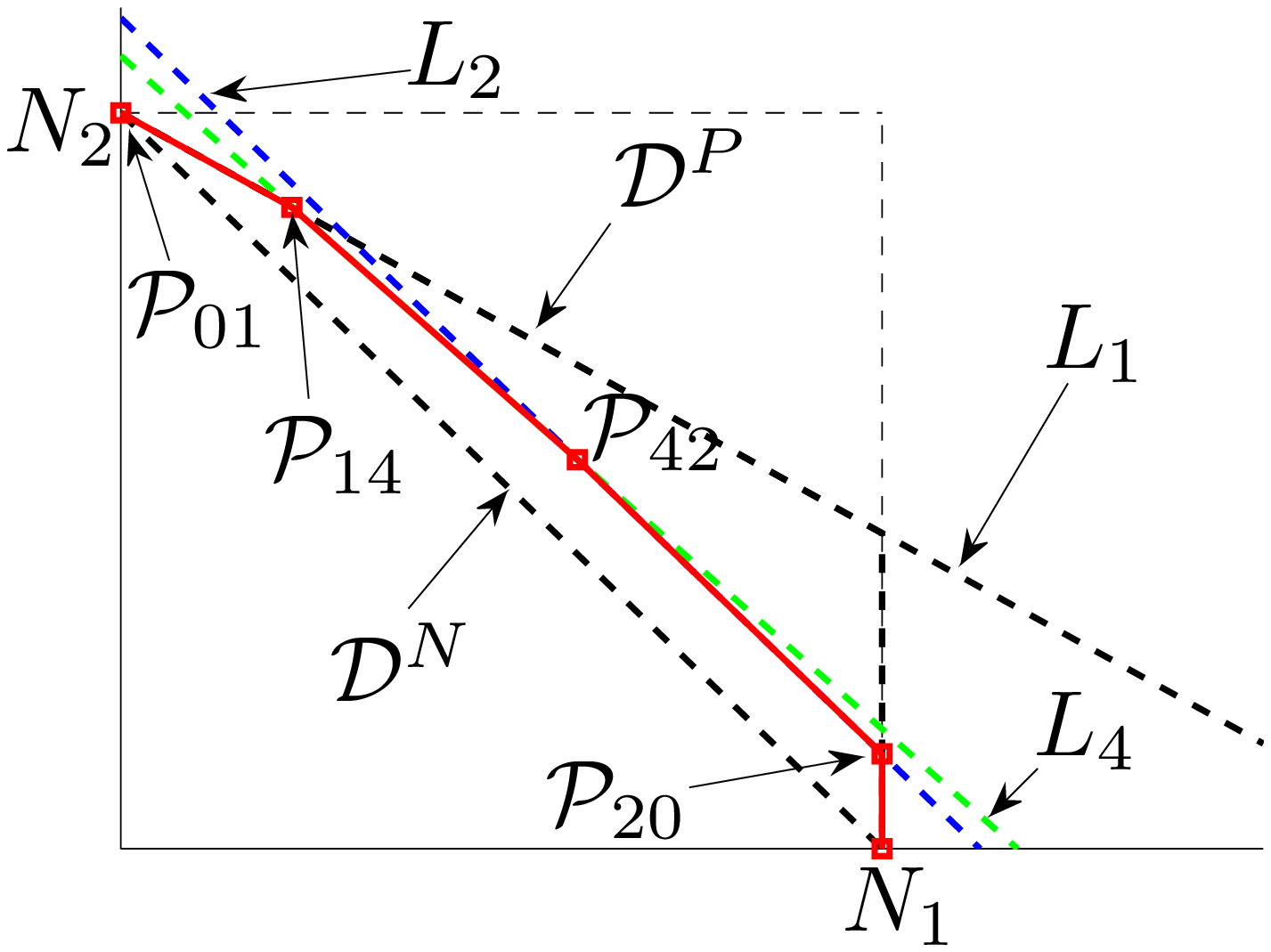}
                \label{fig:II2b1}
        }
        \subfigure[Case II.2.b: If $\frac{M_1{-}N_1^\prime}{\mu_2}{\leq}\alpha_1{\leq}\frac{N_2{-}N_1}{\mu_2{+}N_2{-}N_1^{\prime\prime}}$, $L_2$ inactive]{
                \centering
                \includegraphics[width=0.2\textwidth,height=2.5cm]{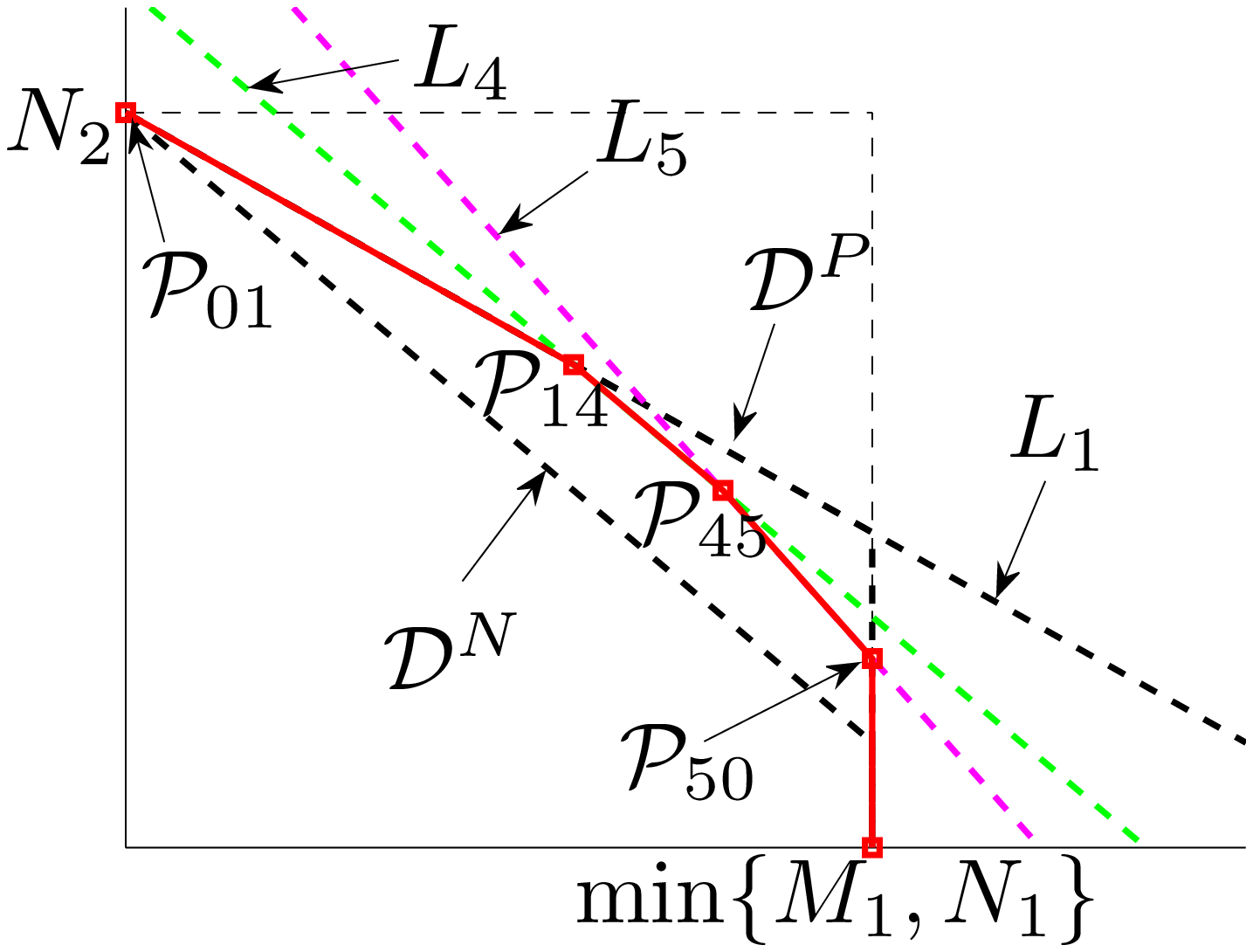}
                \label{fig:II2b2}
        }
        \subfigure[Case II.2.b: If $\alpha_1{\geq}\frac{N_2{-}N_1}{\mu_2{+}N_2{-}N_1^{\prime\prime}}$, $L_2$ and $L_4$ inactive]{
                \centering
                \includegraphics[width=0.2\textwidth,height=2.5cm]{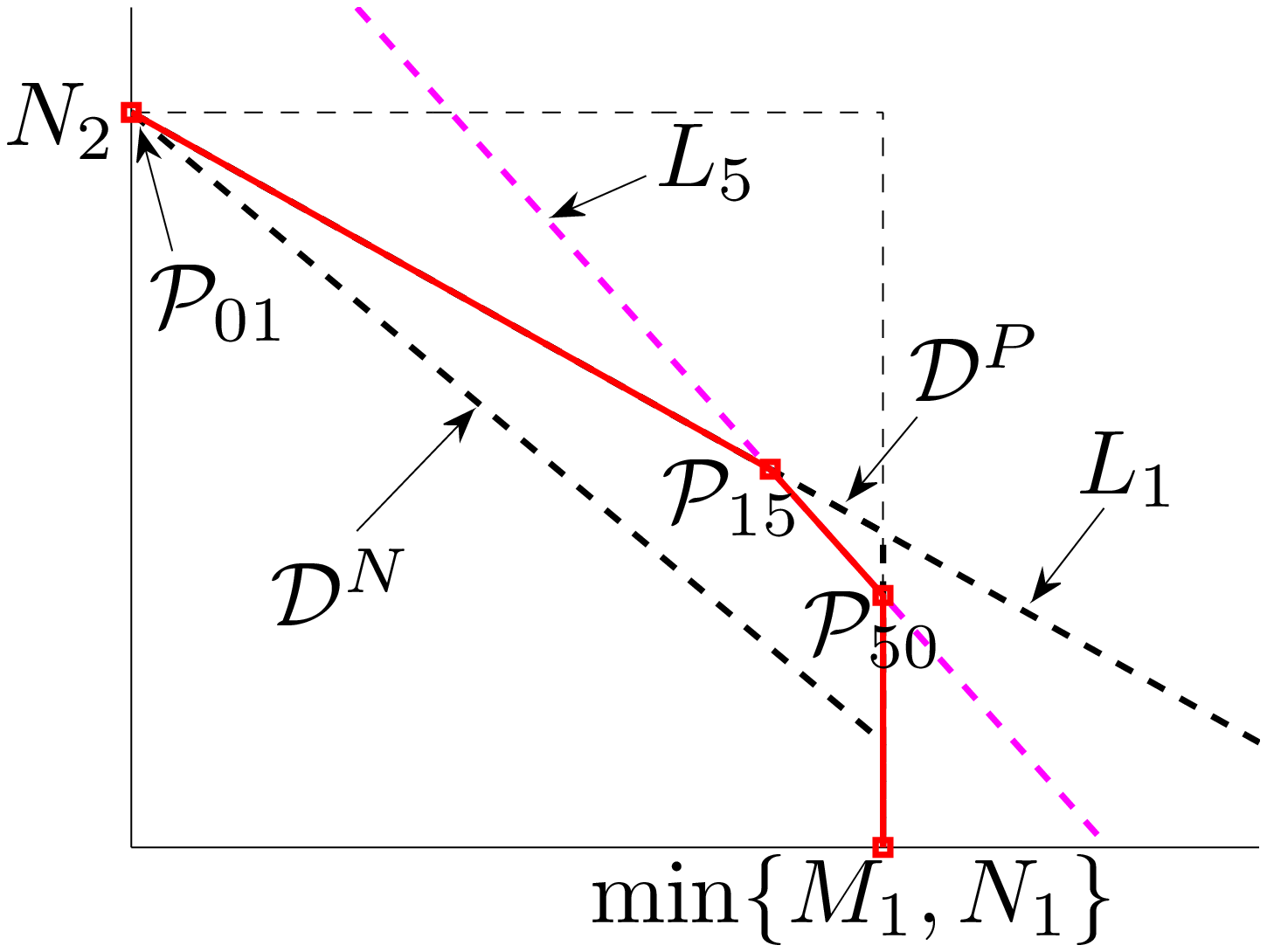}
                \label{fig:II2b3}
        }
        \caption{Achievable DoF regions in $(M_1{,}M_2{,}N_1{,}N_2)$ MIMO IC.}\label{fig:ICregions}
\end{figure}

For clarity, we summarize the active constraints and the resulted corner points in Table \ref{tab:ICregions} for different antenna configurations. Figure \ref{fig:ICregions} illustrates the DoF regions in Proposition \ref{prop:IC1} and \ref{prop:IC2}, where $\mathcal{D}^P$ and $\mathcal{D}^N$ stand for the optimal DoF region when there is perfect CSIT \cite{JarfarMIMOIA2pair} and no CSIT \cite{zhunoCSIT}, respectively.

When $M_1{\geq}N_2$, the DoF region is a function of $\alpha_1$ and $\alpha_2$ according to Proposition \ref{prop:IC1}. When $\alpha_1{=}0$ (${\forall}\alpha_2{\in}[0{,}1]$), the DoF regions become the DoF region with no CSIT. The values of $\alpha_1$ and $\alpha_2$ that lead to the DoF region with perfect CSIT are different according to the antenna configurations, namely $\alpha_1{=}\alpha_2{=}1$ if $N_2{\leq}M_1{\leq}M_2$; $\alpha_1{=}1$, $\alpha_2{\geq}\frac{\min\{M_2{,}N_1{+}N_2\}{-}N_2}{\min\{M_1{,}N_1{+}N_2\}{-}N_2}$ if $N_2{\leq}M_2{\leq}M_1$; and $\alpha_1{=}1$, ${\forall}\alpha_2{\in}[0{,}1]$ if $M_2{\leq}N_2$.

When $M_1{\leq}N_2$, we can see that the DoF region is only a function of $\alpha_1$ according to Proposition \ref{prop:IC2}, the DoF region with perfect and no CSIT are reached when $\alpha_1{=}1$ and $\alpha_2{=}0$ (${\forall}\alpha_2{\in}[0{,}1]$), respectively.

Moreover, if $N_1{=}N_2{=}N$, $M_1{=}M_2{=}M$ and $M{\geq}2N$, $L_1$ and $L_2$ in Proposition \ref{prop:IC1} boil down to the the sum DoF constraint in the symmetric case \cite{Tasos}.


\subsection{Discussion on outer-bound}
\subsubsection{MIMO BC}
An outer-bound of the DoF region of MIMO BC with imperfect CSIT is stated in the following proposition.
\begin{myprop}\label{prop:BC_outer}
For a $(M{,}N_1{,}N_2)$ MIMO BC, supposing $N_1{\leq}N_2$, the DoF region with imperfect CSIT lies in \eqref{eq:BC_outer} at the top of the next page.
\begin{figure*}
{\small\begin{IEEEeqnarray}{rrl}\label{eq:BC_outer}
d_1&{\leq}&\min\{M{,}N_1\},\IEEEyesnumber\IEEEyessubnumber\\
d_2&{\leq}&\min\{M{,}N_2\},\IEEEyessubnumber\\
d_1{+}d_2&{\leq}&\min\{M{,}N_2\}{+}\left[\min\{M{,}N_1{+}N_2\}{-}\min\{M{,}N_2\}\right]\alpha_2,\IEEEyessubnumber\label{eq:outer_sum}\\
\frac{d_1}{\min\{M{,}N_1\}}{+}\frac{d_2}{\min\{M{,}N_2\}}&{\leq}&1{+}\frac{\min\{M{,}N_1{+}N_2\}{-}\min\{M{,}N_1\}}{\min\{M{,}N_2\}}\alpha_1,
\IEEEyessubnumber\label{eq:outer_wsum}
\end{IEEEeqnarray}}
\hrulefill
\end{figure*}
\end{myprop}
\begin{proof}
  See Appendix C.
\end{proof}

The achievable DoF region stated in Proposition \ref{prop:BC} and the outer-bound stated in Proposition \ref{prop:BC_outer} only differ by the sum DoF inequality. It can be verified that the optimality of the achievable DoF region stated in Proposition \ref{prop:BC} holds in two cases, i.e., $\Phi_{BC}{\leq}0$ and $M{\leq}N_2$. In the first case, the optimal sum DoF is\footnote{Note that when $\Phi_{BC}{\leq}0$, one has $M{\geq}N_2$} $N_2{+}\min\{N_1{,}M{-}N_2\}\alpha_2$. In the second case, the optimal sum DoF is $N_2$. Moreover, when $M{\leq}N_2$, the optimal DoF region with imperfect CSIT coincides with the optimal DoF region with a mixture of perfect delayed CSIT and imperfect current CSIT \cite{xinping_mimo}, which implies the uselessness of the delayed CSIT under the antenna configuration $M{\leq}N_2$.

\subsubsection{MIMO IC}

Allowing transmitters to cooperate produces a $(M_1{+}M_2{,}N_1{,}N_2)$ MIMO BC. Then, by replacing $M$ with $M_1{+}M_2$ into \eqref{eq:BC_outer}, we obtain an outer-bound of the DoF region of $(M_1{,}M_2{,}N_1{,}N_2)$ MIMO IC with imperfect CSIT. We discuss the tightness of the this outer-bound following the cases presented in Table \ref{tab:ICregions}.
\begin{itemize}
  \item Case I.1, $M_1{\geq}N_2$ and $M_2{\leq}N_2$: In this case, the obtained outer-bound is loose. However, the optimality of the achievable DoF region stated in Proposition \ref{prop:IC1} holds as it is consistent with the optimal DoF region of a mixture of perfect delayed CSIT and imperfect current CSIT \cite{xinping_mimo}.
  \item Case I.2, $M_1{\geq}N_2$ and $M_2{\geq}N_2$: In this case, when $\Phi_{IC}{\leq}0$, the obtained outer-bound is tight and the achievable DoF region stated in Proposition \ref{prop:IC1} is optimal; otherwise, the outer-bound is loose and the optimal DoF region is unknown.
  \item Case II.1, $M_1{\leq}N_2$ and $M_2{\leq}N_2$: In this case, the obtained outer-bound is loose. However, when $N_1{\leq}M_1{\leq}N_2$, the optimality of the achievable DoF region stated in Proposition \ref{prop:IC2} holds as it is consistent with the optimal DoF region of a mixture of perfect delayed CSIT and imperfect current CSIT \cite{xinping_mimo}.
  \item Case II.2, $M_1{\leq}N_2$ and $M_2{\geq}N_2$: In this case, the obtained outer-bound is loose, and the optimal DoF region is unknown.
\end{itemize}

Next, we will show the achievability proof of Proposition \ref{prop:BC}, \ref{prop:IC1} and \ref{prop:IC2} in Section \ref{sec:BC}, \ref{sec:icI} and \ref{sec:icII}, respectively, by proposing suitable RS schemes with proper power allocation. 

\section{Achievability Proof: Broadcast Channel}\label{sec:BC}
In this section, we firstly design an RS scheme focusing on a a $(4{,}2{,}3)$ BC example, and secondly propose the unified framework for the general asymmetric MIMO BC, which achieves the DoF region stated in Proposition \ref{prop:BC}.

\subsection{RS scheme for the asymmetric case: a $(4{,}2{,}3)$ BC example}\label{sec:BC423}
We constitute the RS transmission block for the $(4{,}2{,}3)$ BC as follows.
\begin{itemize}
  \item $1$ private symbol, denoted by $u_1$, is sent to Rx1 along a ZF-precoder $\mathbf{v}_1{=}\mathbf{H}_2^{\bot}{\in}\mathbb{C}^{4{\times}1}$ with power exponent $A_1$;
  \item $2$ private symbols, denoted by $\mathbf{u}_2^{(1)}{\in}\mathbb{C}^{2{\times}1}$, are sent to Rx2 along a ZF-precoder $\mathbf{V}_2^{(1)}{=}\mathbf{H}_1^{\bot}{\in}\mathbb{C}^{4{\times}2}$ with power exponent $A_2$;
  \item $1$ private symbol, denoted by $u_2^{(2)}$, is sent to Rx2 along a precoder $\mathbf{v}_2^{(2)}{\in}\mathbb{C}^{4{\times}1}$ in the subspace spanned by $\hat{\mathbf{H}}_2$. Its power exponent is $(A_2{-}\alpha_1)^+$.
  \item A common message, denoted by $\mathbf{c}{\in}\mathbb{C}^{4{\times}1}$, is multicast using the remaining power.
\end{itemize}
Moreover, the power exponents $A_1$ and $A_2$ are defined as $A_1{\in}[0{,}\alpha_2]$ and $A_2{\in}[0{,}1]$. Mathematically, the transmitted and received signals write as
\begin{IEEEeqnarray}{rcl}\label{eq:RS423}
\!\!\!\!\!\!\mathbf{s}&{=}&\underbrace{\mathbf{c}}_{P}{+}\underbrace{\mathbf{v}_1u_1}_{P^{A_1}}{+}
\underbrace{\mathbf{V}_2^{(1)}\mathbf{u}_2^{(1)}}_{P^{A_2}}{+}\underbrace{\mathbf{v}_2^{(2)}u_2^{(2)}}_{P^{(A_2{-}\alpha_1)^+}}
\IEEEyesnumber\IEEEyessubnumber\label{eq:s423}\\
\!\!\!\!\!\!\mathbf{y}_1&{=}&\underbrace{\mathbf{H}_1^H\mathbf{c}}_{P}{+}
\underbrace{\mathbf{H}_1^H\mathbf{v}_1u_1}_{P^{A_1}}{+}
\underbrace{\mathbf{H}_1^H\left(\mathbf{V}_2^{(1)}\mathbf{u}_2^{(1)}{+}\mathbf{v}_2^{(2)}u_2^{(2)}\right)}_{P^{(A_2{-}\alpha_1)^+}}
,\IEEEyessubnumber\label{eq:y1_423}\\
\!\!\!\!\!\!\mathbf{y}_2&{=}&\underbrace{\mathbf{H}_2^H\mathbf{c}}_{P}{+}
\underbrace{\mathbf{H}_2^H\mathbf{v}_1u_1}_{P^{A_1{-}\alpha_2}}{+}
\underbrace{\mathbf{H}_2^H\mathbf{V}_2^{(1)}\mathbf{u}_2^{(1)}}_{P^{A_2}}{+}
\underbrace{\mathbf{H}_2^H\mathbf{v}_2^{(2)}u_2^{(2)}}_{P^{(A_2{-}\alpha_1)^+}}.\IEEEyessubnumber\label{eq:y2_423}
\end{IEEEeqnarray}

As we can see from the received signal, if $A_2{\leq}\alpha_1$, the undesired private symbols are drowned into the noise. If $A_2{>}\alpha_1$, the power allocation policy ensures that all the three private symbols intended for Rx2 are received by Rx1 with the same power level. Considering that each receiver decodes the common message and the desired private symbols successively, the following DoF tuple is achievable
\begin{IEEEeqnarray}{rcl}\label{eq:dof423}
\!\!\!\!\!\!\text{\rm At Rx1:}\,\, d_c&{\leq}&d_c^{(1)}{\triangleq}2{-}\max\{A_1{,}A_2{-}\alpha_1\}{-}(A_2{-}\alpha_1)^+\!\!,\IEEEyesnumber\IEEEyessubnumber\label{eq:dc1_423}\\
\!\!\!\!\!\!d_{p1}&{=}&(A_1{-}(A_2{-}\alpha_1)^+)^+,\IEEEyessubnumber\label{eq:dp1_423}\\
\!\!\!\!\!\!\text{\rm At Rx2:}\,\, d_c&{\leq}&d_c^{(2)}{\triangleq}3{-}2A_2{-}(A_2{-}\alpha_1)^+,\IEEEyessubnumber\label{eq:dc2_423}\\
\!\!\!\!\!\!d_{p2}&{=}&2A_2{+}(A_2{-}\alpha_1)^+.\IEEEyessubnumber\label{eq:dp2_423}
\end{IEEEeqnarray}

With the above achievable DoF tuple, we can see that when $\alpha_1{=}\alpha_2{=}0$, the sum DoF $3$ is achieved with $d_{p2}{=}3$, $d_c{=}0$ and $d_{p1}{=}0$. This result is consistent with the optimal sum DoF when there is no CSIT. Besides, when $\alpha_1{=}\alpha_2{=}1$, the sum DoF $4$ is achieved with $d_{p1}{=}2$, $d_{p2}{=}1$ and $d_c{=}1$. This result is consistent with the optimal sum DoF of the perfect CSIT case.

Next, we characterize the achievable DoF region of the $(4{,}2{,}3)$ MIMO BC by finding the maximum achievable sum DoF. We will firstly show the achievability of corner points $\mathcal{P}_{10}$ and $\mathcal{P}_{10^\prime}$ in the case $\Phi_{BC}{\leq}0$, and secondly show the achievability of corner points $\mathcal{P}_{12}$, $\mathcal{P}_{10^\prime}$ and $\mathcal{P}_{20}$ in the case $\Phi_{BC}{\geq}0$ by performing a Space-Time transmission.

\subsubsection{When $\alpha_1{\geq}\frac{1{+}\alpha_2}{2}$, i.e., $\Phi_{BC}{\leq}0$}
\begin{figure}[t]
\renewcommand{\captionfont}{\small}
\captionstyle{center}
\centering
\subfigure[$\alpha_1{\geq}\frac{1{+}\alpha_2}{2}$]{
                \centering
                \includegraphics[width=0.24\textwidth,height=3cm]{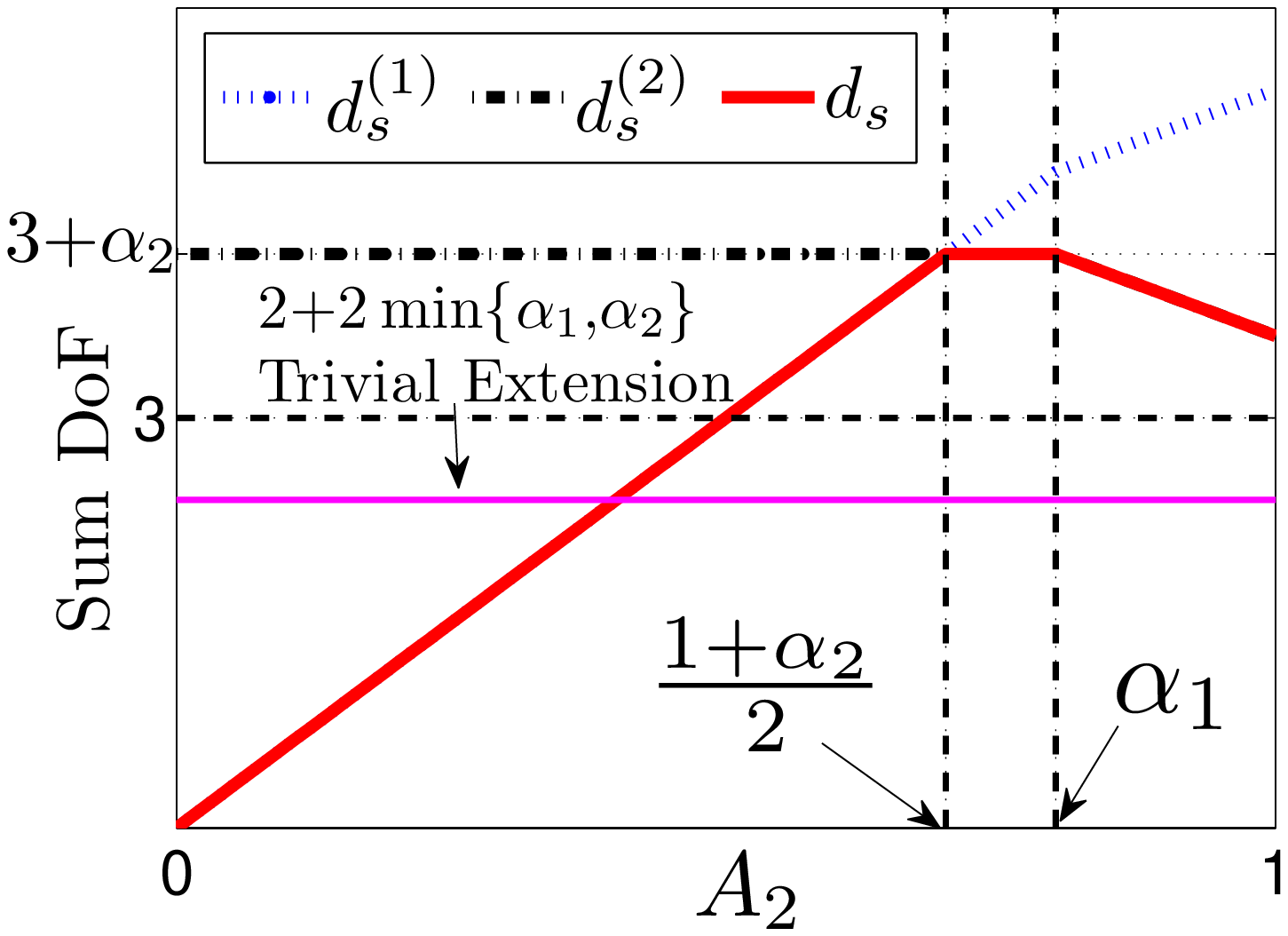}
                \label{fig:BC423caseI}
        }
\subfigure[$1{-}\alpha_2{\leq}\alpha_1{<}\frac{1{+}\alpha_2}{2}$]{
                \centering
                \includegraphics[width=0.24\textwidth,height=3cm]{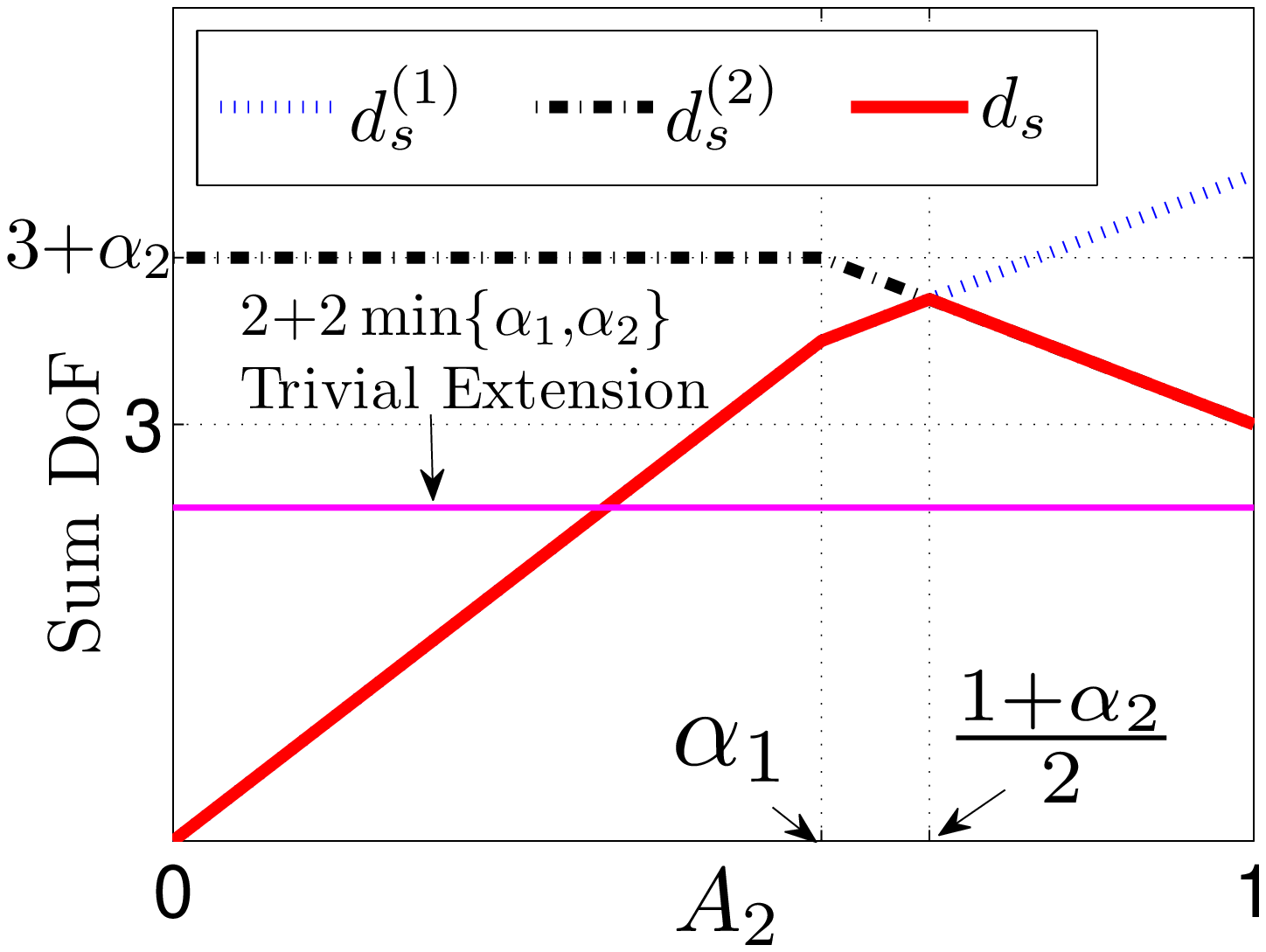}
                \label{fig:BC423caseIIa}
        }
        \\
        \subfigure[$\frac{1{-}\alpha_2}{2}{\leq}\alpha_1{\leq}\min\{1{-}\alpha_2{,}\frac{1{+}\alpha_2}{2}\}$]{
                \centering
                \includegraphics[width=0.24\textwidth,height=3cm]{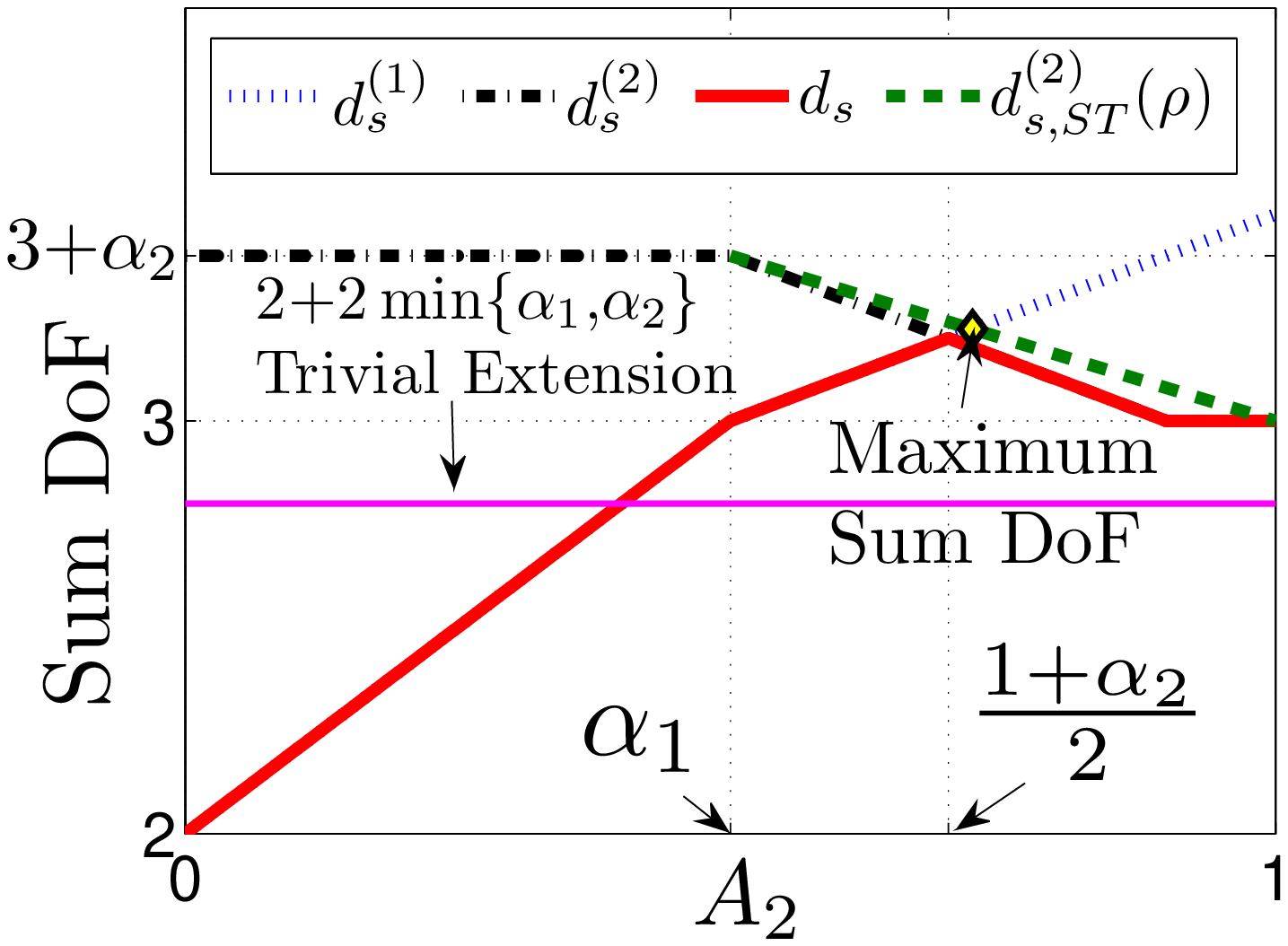}
                \label{fig:BC423caseIIb}
        }
        \subfigure[$\alpha_1{\leq}\frac{1{-}\alpha_2}{2}$]{
                \centering
                \includegraphics[width=0.24\textwidth,height=3cm]{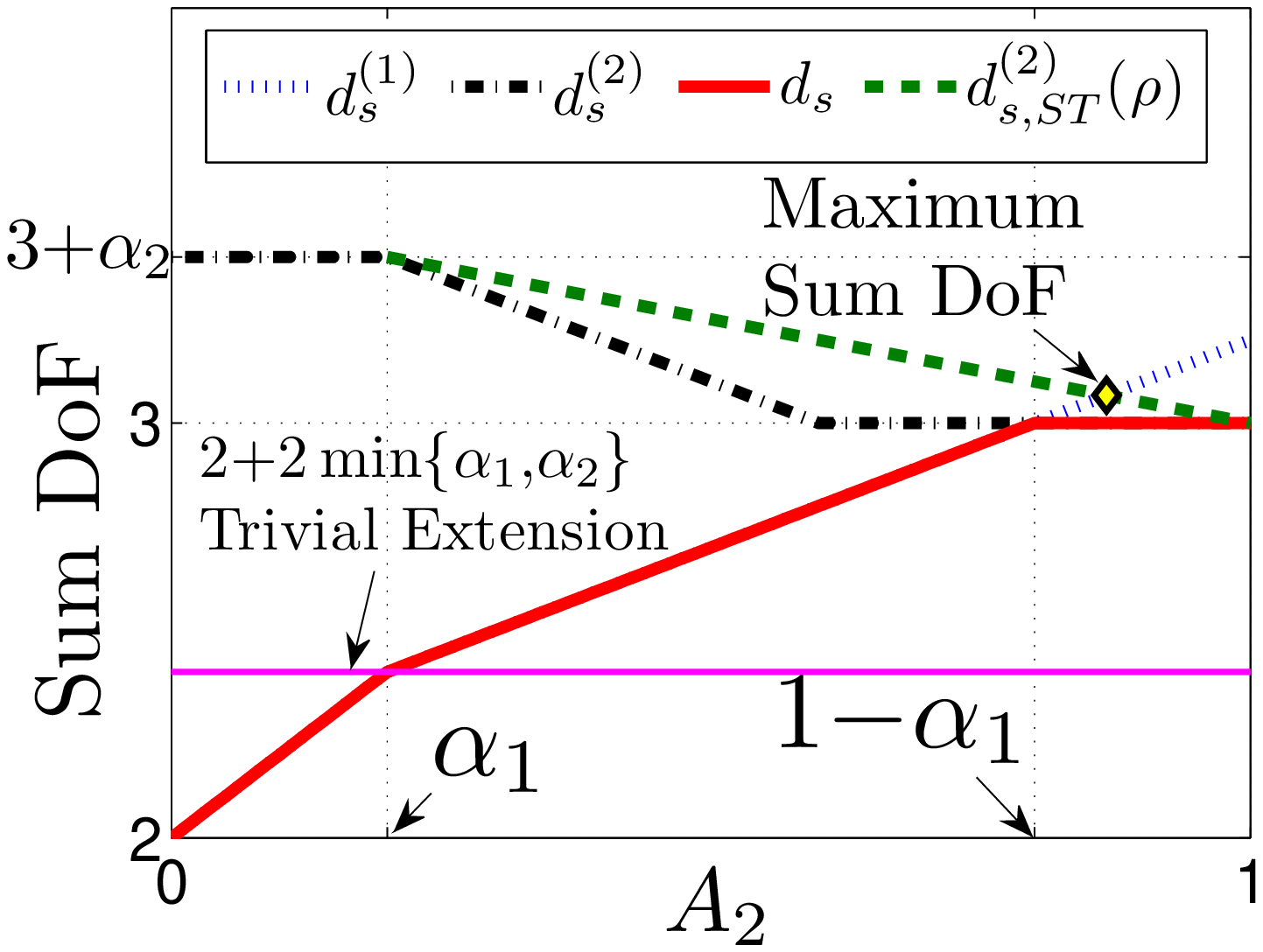}
                \label{fig:BC423caseIIc}
        }
\caption{Sum DoF of a $(4{,}2{,}3)$ MIMO BC}\label{fig:ds423}
\end{figure}
Let us define the achievable sum DoF as a function of the power levels, i.e., $d_s(A_1{,}A_2){\triangleq}\min\{d_s^{(1)}(A_2){,}d_s^{(2)}(A_1{,}A_2)\}$, where
\begin{IEEEeqnarray}{rcl}
d_s^{(1)}(A_2)&{=}&2{+}2A_2{-}(A_2{-}\alpha_1)^+,\IEEEyesnumber\IEEEyessubnumber\label{eq:ds1_423}\\
d_s^{(2)}(A_1{,}A_2)&{=}&3{+}(A_1{-}(A_2{-}\alpha_1)^+)^+,\IEEEyessubnumber\label{eq:ds2_423}
\end{IEEEeqnarray}
are obtained by summing \eqref{eq:dc1_423}, \eqref{eq:dp1_423}, \eqref{eq:dp2_423} and \eqref{eq:dp1_423}, \eqref{eq:dc2_423}, \eqref{eq:dp2_423}, respectively. Then, it can be shown that the power levels $(A_1^*{,}A_2^*){\triangleq}{\arg\max}d_s(A_1{,}A_2)$ that maximize the sum DoF are given by
\begin{equation}
A_1^*{=}\alpha_2{,}\quad A_2^*{=}\max\left\{\frac{1{+}\alpha_2}{2}{,}1{-}\alpha_1\right\},\label{eq:A1A2_423}
\end{equation}
because $d_s(A_1{,}A_2)$ increases with $A_1$, while $A_2$ is chosen such that the common-message-decodabilities at the two users are equalized, i.e., $d_s^{(1)}{=}d_s^{(2)}$ (or $d_c^{(1)}{=}d_c^{(2)}$). Figure \ref{fig:ds423} illustrates the maximum sum DoF for different values of $\alpha_1$ and $\alpha_2$ (the highest point of the red solid curve).

Here, as we are considering $\alpha_1{\geq}\frac{1{+}\alpha_2}{2}$, i.e., $\Phi_{BC}{\leq}0$, it can be verified that the sum DoF is maximized with $A_2^*{=}\frac{1{+}\alpha_2}{2}$ (as shown in Figure \ref{fig:BC423caseI}), which is smaller than $\alpha_1$. Plugging $A_2^*{=}\frac{1{+}\alpha_2}{2}$ and $A_1^*{=}\alpha_2$ into \eqref{eq:dp1_423}, \eqref{eq:dp2_423} and \eqref{eq:dc1_423} yields $d_{p1}{=}\alpha_2$, $d_{p2}{=}1{+}\alpha_2$ and $d_{c}{=}2{-}\alpha_2$. If the common message only carries information intended for user $1$ (resp. user $2$), we obtain the corner point $\mathcal{P}_{10}{=}(2{,}1{+}\alpha_2)$ (resp. $\mathcal{P}_{10^\prime}{=}(\alpha_2{,}3)$). Note that in this case, $L_{2{,}BC}$ in \eqref{eq:BCwsum} is inactive and the DoF region is formed by corner points $\mathcal{P}_{10}$ and $\mathcal{P}_{10^\prime}$.

\subsubsection{When $\alpha_1{\leq}\frac{1{+}\alpha_2}{2}$, i.e., $\Phi_{BC}{\geq}0$}
In this case, as shown by the highest point of the red solid curves in Figure \ref{fig:BC423caseIIa}, \ref{fig:BC423caseIIb} and \ref{fig:BC423caseIIc}, the optimal $A_2^*$ in \eqref{eq:A1A2_423} is greater than or equal to $\alpha_1$. Notably, this fact contrasts the power allocation in the MISO case where choosing the power level $A_1{=}A_2{=}\min\{\alpha_1{,}\alpha_2\}$ suffices to achieve the maximal sum DoF. The reason responsible for this observation is that with $A_2^*{\geq}\alpha_1$, the transmitter exploits the larger spatial dimension at user $2$ by delivering $3$ private messages to user $2$, while the interference overheard by user $1$ spans only $2$ dimensions.

However, the sum DoF can be further improved by a Space-Time transmission when $\Phi_{BC}{\geq}0$, which leads to the corner point $\mathcal{P}_{10^\prime}$ and $\mathcal{P}_{12}$. The transmission lasts for $T$ time slots. Letting $A_{k{,}l}$ denote the power level chosen for user $k$ in slot $l$, we choose $(A_{1{,}l}{,}A_{2{,}l}){=}(\alpha_2{,}1)$ for $l{=}1{,}\cdots{,}{\rho}T$ and $(A_{1{,}l}{,}A_{2{,}l}){=}(\alpha_2{,}\alpha_1)$ for $l{=}{\rho}T{+}1{,}\cdots{,}T$, where $0{\leq}\rho{\leq}1$. Note that we consider that $T$ is a sufficiently large integer such that $\rho T$ is an integer as well. The decoding is performed focusing on the aggregate received signals, namely $\left[\mathbf{y}_{k}(1){,}\cdots{,}\mathbf{y}_{k}(T)\right]^T$. Then, by plugging these power levels into \eqref{eq:ds1_423} and \eqref{eq:ds2_423} and computing the average sum DoF over the total $T$ channel uses, we have $d_{s{,}ST}(\rho){\triangleq}\min\{d_{s{,}ST}^{(1)}(\rho){,}d_{s{,}ST}^{(2)}(\rho)\}$, where
\begin{IEEEeqnarray}{rcl}
d_{s{,}ST}^{(1)}(\rho)&{=}&\rho d_s^{(1)}(1){+}(1{-}\rho)d_s^{(1)}(\alpha_1),\IEEEyesnumber\IEEEyessubnumber\label{eq:ds1_423rho}\\
d_{s{,}ST}^{(2)}(\rho)&{=}&\rho d_s^{(2)}(\alpha_2{,}1){+}(1{-}\rho)d_s^{(2)}(\alpha_2{,}\alpha_1).\IEEEyessubnumber\label{eq:ds2_423rho}
\end{IEEEeqnarray}

In Figure \ref{fig:BC423caseIIb} and \ref{fig:BC423caseIIc}, $d_{s{,}ST}^{(2)}(\rho)$ is illustrated by the green dotted line with $d_{s{,}ST}^{(2)}(0){=}d_s^{(2)}(\alpha_2{,}\alpha_1)$ and $d_{s{,}ST}^{(2)}(1){=}d_s^{(2)}(\alpha_2{,}1)$. However, in Figure \ref{fig:BC423caseIIa}, $d_{s{,}ST}^{(2)}(\rho)$ coincides with $d_s^{(2)}(\alpha_2{,}A_2)$ because $d_s^{(2)}(\alpha_2{,}A_2)$ is linear within the range $A_2{\in}[\alpha_1{,}1]$. Besides, $d_{s{,}ST}^{(1)}(\rho)$ coincides with $d_s^{(1)}(A_2)$ in Figure \ref{fig:BC423caseIIa}, \ref{fig:BC423caseIIb} and \ref{fig:BC423caseIIc}. In all the three figures, the maximum sum DoF achieved with Space-Time transmission is obtained with $\rho^*$ such that $d_{s{,}ST}^{(1)}(\rho^*){=}d_{s{,}ST}^{(2)}(\rho^*)$ holds (see the diamond points). Compared to the sum DoF achieved without Space-Time transmission (i.e., the highest point on the red solid curve), we can read from Figure \ref{fig:BC423caseIIb} and \ref{fig:BC423caseIIc} that $d_{s{,}ST}(\rho^*){>}d_s(A_1^*{,}A_2^*)$. However, in Figure \ref{fig:BC423caseIIa}, we have $d_{s{,}ST}(\rho^*){=}d_s(A_1^*{,}A_2^*)$. Through some calculation, we present the choices of $\rho^*$ and the the sum DoF achieved with and without Space-Time transmission in Table \ref{tab:ST}.
\begin{table*}[t]
\renewcommand{\arraystretch}{1.3}
\vspace{.6em}
\centering
\begin{tabular}{|c|c|c|}
\hline
Conditions & Without Space-Time Transmission & With Space-Time Transmission\\
\hline
a) & $A_2^*{=}\frac{1{+}\alpha_2}{2}$, $d_s{=}3{+}\alpha_2$ & N/A\\
\hline
b) & $A_2^*{=}\frac{1{+}\alpha_2}{2}$, $d_s{=}\frac{5{+}\alpha_2{+}2\alpha_1}{2}$ & $\rho^*{=}\frac{1{-}2\alpha_1{+}\alpha_2}{2{-}2\alpha_1}$ $d_s{=}\frac{5{+}\alpha_2{+}2\alpha_1}{2}$\\
\hline
c) & $A_2^*{=}\frac{1{+}\alpha_2}{2}$, $d_s{=}\frac{5{+}\alpha_2{+}2\alpha_1}{2}$ & $\rho^*{=}\frac{1{-}2\alpha_1{+}\alpha_2}{1{-}\alpha_1{+}\alpha_2}$ $d_s{=}3{+}\frac{\alpha_1\alpha_2}{1{-}\alpha_1{+}\alpha_2}$\\
\hline
d) & $A_2^*{=}1{-}\alpha_1$, $d_s{=}3$ & $\rho^*{=}\frac{1{-}2\alpha_1{+}\alpha_2}{1{-}\alpha_1{+}\alpha_2}$ $d_s{=}3{+}\frac{\alpha_1\alpha_2}{1{-}\alpha_1{+}\alpha_2}$\\
\hline
\end{tabular}
\caption{Sum DoF achieved with different schemes in a $(4{,}2{,}3)$ MIMO BC, where conditions a), b), c) and d) are such that in Figure \ref{fig:BC423caseI}, \ref{fig:BC423caseIIa}, \ref{fig:BC423caseIIb} and \ref{fig:BC423caseIIc}, respectively.}\label{tab:ST}
\vspace{-0mm}
\end{table*}

With the power allocation across the $T$ slots, we can compute
\begin{IEEEeqnarray}{rcl}
  d_{p1}&{=} & \rho^*(\alpha_1{+}\alpha_2{-}1)^+{+}(1{-}\rho^*)\alpha_2{,}\IEEEyesnumber\IEEEyessubnumber\\
  d_{p2}&{=} & \rho^*(3{-}\alpha_1){+}(1{-}\rho^*)\cdot2\alpha_1{,}\IEEEyessubnumber\\
  d_c&{=} & \rho^*\alpha_1{+}(1{-}\rho^*)(3{-}\alpha_1),\IEEEyessubnumber
\end{IEEEeqnarray}
where $\rho^*$ is given in Table \ref{tab:ST}. Considering that the common message only carries information intended for Rx$1$ and Rx$2$, we obtain the corner points $\mathcal{P}_{12}$ and $\mathcal{P}_{10^\prime}$ in Figure \ref{fig:BCcase2}, respectively.

To be complete, the corner point $\mathcal{P}_{20}{=}(2{,}2\alpha_1)$ is achievable by substituting $A_1{=}\alpha_2$ and $A_2{=}\alpha_1$ into \eqref{eq:dc1_423} through to \eqref{eq:dp2_423}, and assuming that the common message only carries information for Rx$1$.

\subsection{RS scheme for the asymmetric case: Unified Framework}\label{sec:BCuni}
In this part, we consider the asymmetric MIMO case with $M{\leq}N_1{+}N_2$ and $M{\geq}N_2{\geq}N_1$, as the achievability for other cases can be shown by switching off the redundant transmit/receive antennas. Motivated by the $(4{,}2{,}3)$ MIMO BC example in the last subsection, the transmission block is constructed as follows.
\begin{itemize}
  \item $M{-}N_2$ private symbols, denoted by $\mathbf{u}_1{\in}\mathbb{C}^{(M{-}N_2){\times}1}$, are sent to Rx1 with power exponent $A_1$ along a ZF-precoder $\mathbf{V}_1{=}\mathbf{H}_2^{\bot}{\in}\mathbb{C}^{M{\times}(M{-}N_2)}$ ;
  \item $M{-}N_1$ private symbols, denoted by $\mathbf{u}_2^{(1)}{\in}\mathbb{C}^{(M{-}N_1){\times}1}$, are sent to Rx2 with power exponent $A_2$ along a ZF-precoder $\mathbf{V}_2^{(1)}{=}\mathbf{H}_1^{\bot}{\in}\mathbb{C}^{M{\times}(M{-}N_1)}$;
  \item $N_1{+}N_2{-}M$ private symbols, denoted by $\mathbf{u}_2^{(2)}$, are sent to Rx2 along a precoder $\mathbf{V}_2^{(2)}{\in}\mathbb{C}^{4{\times}1}$ in the subspace spanned by $\hat{\mathbf{H}}_2$. Its power exponent is $(A_2{-}\alpha_1)^+$.
  \item A common message, denoted by $\mathbf{c}{\in}\mathbb{C}^{M{\times}1}$, is multicast using the remaining power.
\end{itemize}
The power exponents are defined as $0{\leq}A_1{\leq}\alpha_2$ and $0{\leq}A_2{\leq}1$. Mathematically, the transmitted and received signals write as
\begin{IEEEeqnarray}{rcl}
\!\!\!\!\!\!\!\!\mathbf{s}&{=}&\underbrace{\mathbf{c}}_{P}{+}
\underbrace{\mathbf{V}_1\mathbf{u}_1}_{P^{A_1}}{+}
\underbrace{\mathbf{V}_2^{(1)}\mathbf{u}_2^{(1)}}_{P^{A_2}}{+}
\underbrace{\mathbf{V}_2^{(2)}\mathbf{u}_2^{(2)}}_{P^{(A_2{-}\alpha_1)^+}}{,}\IEEEyesnumber\IEEEyessubnumber\label{eq:JMBBC}\\
\!\!\!\!\!\!\!\!\mathbf{y}_1&{=}&\underbrace{\mathbf{H}_1^H\mathbf{c}}_{P}{+}
\underbrace{\mathbf{H}_1^H\mathbf{V}_1\mathbf{u}_1}_{P^{A_1}}{+}
\underbrace{\mathbf{H}_1^H\left(\mathbf{V}_2^{(1)}\mathbf{u}_2^{(1)}{+}
\mathbf{V}_2^{(2)}\mathbf{u}_2^{(2)}\right)}_{P^{(A_2{-}\alpha_1)^+}}{,}\IEEEyessubnumber\label{eq:y1BC}\\
\!\!\!\!\!\!\!\!\mathbf{y}_2&{=}&\underbrace{\mathbf{H}_2^H\mathbf{c}}_{P}{+}
\underbrace{\mathbf{H}_2^H\mathbf{V}_1\mathbf{u}_1}_{P^{A_1{-}\alpha_2}}{+}
\underbrace{\mathbf{H}_2^H\mathbf{V}_2^{(1)}\mathbf{u}_2^{(1)}}_{P^{A_2}}{+}
\underbrace{\mathbf{H}_2^H\mathbf{V}_2^{(2)}\mathbf{u}_2^{(2)}}_{P^{(A_2{-}\alpha_1)^+}}{.}\IEEEyessubnumber\label{eq:y2BC}
\end{IEEEeqnarray}
For the MACs given in \eqref{eq:y1BC} and \eqref{eq:y2BC}, using the proof presented in Appendix A, the common message and private messages are successfully decoded if the DoF tuple lies in
\begin{IEEEeqnarray}{rcl}\label{eq:dofnew}
\!\!\!\!\text{\rm At Rx1:}\,\, d_{p1}&{=}&(M{-}N_2)(A_1{-}(A_2{-}\alpha_1)^+)^+,\IEEEyesnumber\IEEEyessubnumber\label{eq:dp1BC}\\
\!\!\!\!d_c&{\leq}&d_c^{(1)}{\triangleq}N_1{-}(M{-}N_2)\max\{A_1{,}A_2{-}\alpha_1\}{-}\nonumber\\
\!\!\!\!&&(N_1{+}N_2{-}M)(A_2{-}\alpha_1)^+,\IEEEyessubnumber\label{eq:dc1BC}\\
\!\!\!\!\text{\rm At Rx2:}\,\, d_{p2}&{=}&(M{-}N_1)A_2{+}(N_1{+}N_2{-}M)(A_2{-}\alpha_1)^+.\IEEEyessubnumber\label{eq:dp2BC}\\
\!\!\!\!d_c&{\leq}&d_c^{(2)}{\triangleq}N_2{-}(M{-}N_1)A_2{-}\nonumber\\
\!\!\!\!&&(N_1{+}N_2{-}M)(A_2{-}\alpha_1)^+.\IEEEyessubnumber\label{eq:dc2BC}
\end{IEEEeqnarray}

Following the footsteps in the $(4{,}2{,}3)$ example, we find that the sum DoF without Space-Time transmission is maximized with the power exponents
\begin{equation}
A_2^*{=}\max\left\{\frac{N_2{-}N_1{+}(M{-}N_2)\alpha_2}{M{-}N_1}{,}1{-}\frac{M{-}N_2}{N_2{-}N_1}\alpha_1\right\}.\label{eq:A2tmp}
\end{equation}

\subsubsection{When $\alpha_1{\geq}\frac{N_2{-}N_1{+}(M{-}N_2)\alpha_2}{M{-}N_1}$, i.e., $\Phi_{BC}{\leq}0$} In this case, choosing $A_2^*{=}\alpha_1^\prime{=}\frac{N_2{-}N_1{+}(M{-}N_2)\alpha_2}{M{-}N_1}{\leq}\alpha_1$ and $A_1^*{=}\alpha_2$ allows us to achieve the maximum sum DoF $N_2{+}(M{-}N_2)\alpha_2$. If $\mathbf{c}$ only carries information intended for Rx1 (resp. Rx2), the corner points $\mathcal{P}_{10}{=}(N_1{,}(M{-}N_1)\alpha_1^\prime)$ (resp. $\mathcal{P}_{10^\prime}{=}((M{-}N_2)\alpha_2{,}N_2)$) in Figure \ref{fig:BCcase1} is achieved.

\subsubsection{When $\alpha_1{\leq}\frac{N_2{-}N_1{+}(M{-}N_2)\alpha_2}{M{-}N_1}$, i.e., $\Phi_{BC}{\geq}0$} In this case, similar to the $(4{,}2{,}3)$ example, we further enhance the sum DoF by performing a Space-Time transmission, where the power exponents are $(A_1{,}A_2){=}(\alpha_2{,}\alpha_1)$ for a fraction $\rho$ of the total time, while the power exponents are $(A_1{,}A_2){=}(\alpha_2{,}1)$ for the rest of the time. The sum DoF is maximized by choosing the optimal $\rho{=}\rho_{BC}^*$ such that the common message decodabilities at the two receivers are balanced (focusing on the aggregate received signals). We present the value of $\rho_{BC}^*$ as
\begin{equation}
\!\!\!\rho_{BC}^*{=} \frac{(M{-}N_1)(1{-}\alpha_1){-}(M{-}N_2)(1{-}\alpha_2)}{(N_2{-}N_1)(1{-}\alpha_1){+}(M{-}N_2)(\alpha_2{-}(\alpha_2{+}\alpha_1{-}1)^+)},
\end{equation}
while the derivation is omitted as it follows the same footsteps as the $(4{,}2{,}3)$ example. Then, the achievable DoF tuple writes as
\begin{IEEEeqnarray}{rcl}
d_{p1{,}ST}(\rho_{BC}^*)&{=}&(M{-}N_2)\left[\rho_{BC}^*(\alpha_1{+}\alpha_2{-}1)^+{+}\right.\nonumber\\
&&\left.(1{-}\rho_{BC}^*)\alpha_2\right], \IEEEyesnumber\IEEEyessubnumber\label{eq:dp1TBC}\\
d_{p2{,}ST}(\rho_{BC}^*)&{=}&\rho_{BC}^*\left(N_2{-}(N_1{+}N_2{-}M)\alpha_1\right){+}\nonumber\\
&&(1{-}\rho_{BC}^*)(M{-}N_1)\alpha_1. \IEEEyessubnumber\label{eq:dp2TBC}\\
d_{c{,}ST}^{(2)}(\rho_{BC}^*)&{=}&\rho_{BC}^*(N_2{+}N_1{-}M)\alpha_1{+}\nonumber\\
&&(1{-}\rho_{BC}^*)\left(N_2{-}(M{-}N_1)\alpha_1\right), \IEEEyessubnumber\label{eq:dc2TBC}
\end{IEEEeqnarray}

If $\mathbf{c}$ only carries information intended for Rx1 and Rx2, the corner point $\mathcal{P}_{12}$ and $\mathcal{P}_{10^\prime}$ in Figure \ref{fig:BCcase2} are obtained, respectively.

To be complete, it remains to achieve the corner point $\mathcal{P}_{20}$ in Figure \ref{fig:BCcase2}. Using the new RS scheme, taking $A_k{=}\alpha_j$ into \eqref{eq:dofnew} yields $d_c{=}N_1{-}(M{-}N_2)\alpha_2$ and $d_{pk}{=}(M{-}N_j)\alpha_j$ for $k{,}j{=}1{,}2{,}k{\neq}j$. Then, corner point $\mathcal{P}_{20}{=}(N_1{,}(M{-}N_1)\alpha_1)$ is immediate if $\mathbf{c}$ is intended for Rx1. Linking $\mathcal{P}_{20}$ and $\mathcal{P}_{12}$ yields $L_2$ in Proposition \ref{prop:BC}. 

\section{Achievability Proof: Interference Channel}\label{sec:IC}
In this section, we move on to discuss the achievabile DoF region in the interference channel. The discussion for Case I with $M_1{\geq}N_2$ and Case II with $M_1{\leq}N_2$ are presented in Section \ref{sec:icI} and \ref{sec:icII}, respectively. For each case, we propose the RS transmission block and perform the DoF calculation in two subcases, i.e., $M_2{\leq}N_2$ and $M_2{\geq}N_2$. Without loss of generality, we consider $M_k{\leq}N_1{+}N_2$ and $N_k{\leq}M_1{+}M_2$, $k{=}1{,}2$, throughout the section, as the achievability in all the other configurations can be shown similarly by switching off the redundant transmit/receive antennas.

\subsection{Case I: $M_1{\geq}N_2$}\label{sec:icI}
In this part, for convenience, we employ the notation $N_2^\prime{\triangleq}\min\{M_2{,}N_2\}$. Since $M_2{\geq}N_1$ and $M_1{\geq}N_2$, we point out that this antenna configuration yields a scenario similar to BC based on the following facts: 1) the desired signal of each receiver is completely mixed with the interference signal, and 2) both receivers are able to deliver ZF-precoded private messages in the null space of the cross-link. Accordingly, we build the RS scheme similar to that in the asymmetric MIMO BC but in a distributed manner. Specifically, the transmitted signals write as
\begin{IEEEeqnarray}{rcl}\label{eq:ICs}
\mathbf{s}_1&{=}&\underbrace{\mathbf{c}_1}_{P}{+}\underbrace{\mathbf{V}_1\mathbf{u}_1}_{P^{A_1}}{,} \IEEEyesnumber\IEEEyessubnumber\label{eq:s1caseI}\\
\mathbf{s}_2&{=}&\underbrace{\mathbf{c}_2}_{P}{+}\underbrace{\mathbf{V}_2^{(1)}\mathbf{u}_2^{(1)}}_{P^{A_2}}{+}
\underbrace{\mathbf{V}_2^{(2)}\mathbf{u}_2^{(2)}}_{P^{(A_2{-}\alpha_1)^+}}{,}\IEEEyessubnumber\label{eq:s2caseI}
\end{IEEEeqnarray}
where $\mathbf{V}_1{=}\hat{\mathbf{H}}_{21}^\bot$ and $\mathbf{V}_2^{(1)}{=}\hat{\mathbf{H}}_{12}^\bot$ are ZF-precoders, $\mathbf{u}_1{\in}\mathbb{C}^{(M_1{-}N_2){\times}1}$ and $\mathbf{u}_2^{(1)}{\in}\mathbb{C}^{(M_2{-}N_1){\times}1}$ are the ZF-precoded private symbols intended for Rx1 and Rx2, respectively, while $\mathbf{u}_2^{(2)}{\in}\mathbb{C}^{(N_1{+}N_2^\prime{-}M_2){\times}1}$ is precoded with the full rank matrix $\mathbf{V}_2^{(2)}{\in}\mathbb{C}^{M_2{\times}(N_1{+}N_2^\prime{-}M_2)}$ in the subspace of $\hat{\mathbf{H}}_{22}$. The power exponents are defined as $A_1{\in}[0{,}\alpha_2]$ and $A_2{\in}[0{,}1]$. Unlike the BC case where the common messages are generally denoted by $\mathbf{c}$, we introduce $\mathbf{c}_k$ to denote the common message carries information intended for Rx$k$, $k{=}1{,}2$, as $\mathbf{c}_1$ and $\mathbf{c}_2$ are transmitted from different transmitters. The resultant received signals are expressed as
\begin{IEEEeqnarray}{rcl}\label{eq:ICy}
\!\!\!\!\!\!\mathbf{y}_1&{=}&\underbrace{\mathbf{H}_{11}^H\mathbf{c}_1}_{P}{+}\underbrace{\mathbf{H}_{12}^H\mathbf{c}_2}_{P}{+}
\underbrace{\mathbf{H}_{11}^H\mathbf{V}_1\mathbf{u}_1}_{P^{A_1}}{+}
\underbrace{\boldsymbol\eta_1}_{P^{(A_2{-}\alpha_1)^+}}
,\IEEEyesnumber\IEEEyessubnumber\label{eq:y1caseI}\\
\!\!\!\!\!\!\mathbf{y}_2&{=}&\underbrace{\mathbf{H}_{21}^H\mathbf{c}_1}_{P}\!\!{+}\underbrace{\mathbf{H}_{22}^H\mathbf{c}_2}_{P}{+}\!\!\!\!\!\!
\underbrace{\boldsymbol\eta_2}_{P^{A_1{-}\alpha_2}}\!\!\!\!\!\!{+}
\underbrace{\mathbf{H}_{22}^H\mathbf{V}_2^{(1)}\mathbf{u}_2^{(1)}}_{P^{A_2}}\!\!\!{+}
\underbrace{\mathbf{H}_{22}^H\mathbf{V}_2^{(2)}\mathbf{u}_2^{(2)}}_{P^{(A_2{-}\alpha_1)^+}}.\IEEEyessubnumber\label{eq:y2caseI}
\end{IEEEeqnarray}
where $\boldsymbol\eta_1{\triangleq}\mathbf{H}_{12}^H\left(\mathbf{V}_2^{(1)}\mathbf{u}_2^{(1)}{+}\mathbf{V}_2^{(2)}\mathbf{u}_2^{(2)}\right)$ and $\boldsymbol\eta_2{\triangleq}\mathbf{H}_{21}^H\mathbf{V}_1\mathbf{u}_1$.

Following the derivations in Appendix A, the MACs in \eqref{eq:y1caseI} and \eqref{eq:y2caseI} yield the following achievable DoF tuple
\begin{IEEEeqnarray}{rcl}\label{eq:dofcaseI}
\text{\rm At Rx1:} \quad d_{c1}&{\leq}&N_1{-}(M_1{-}N_2)\max\{A_1{,}A_2{-}\alpha_1\}{-}\nonumber\\
&&(N_1{+}N_2{-}M_1)(A_2{-}\alpha_1)^+,\IEEEyesnumber\IEEEyessubnumber\label{eq:dc1y1caseI}\\
d_{c2}&{\leq}&\text{\rm r.h.s. of \eqref{eq:dc1y1caseI}},\IEEEyessubnumber\label{eq:dc2y1caseI}\\
d_{c1}{+}d_{c2}&{\leq}&\text{\rm r.h.s. of \eqref{eq:dc1y1caseI}},\IEEEyessubnumber\label{eq:dcsy1caseI}\\
d_{p1}&{=}&(M_1{-}N_2)(A_1{-}(A_2{-}\alpha_1)^+)^+.\IEEEyessubnumber\label{eq:dp1caseI}\\
\text{\rm At Rx2:}\quad d_{c1}&{\leq}&N_2{-}(M_2{-}N_1)A_2{-}\nonumber\\
&&(N_1{+}N_2^\prime{-}M_2)(A_2{-}\alpha_1)^+\!\!,\IEEEyessubnumber\label{eq:dc1y2caseI}\\
d_{c2}&{\leq}&N_2^\prime{-}(M_2{-}N_1)A_2{-}\nonumber\\
&&(N_1{+}N_2^\prime{-}M_2)(A_2{-}\alpha_1)^+\!\!,
\IEEEyessubnumber\label{eq:dc2y2caseI}\\
d_{c1}{+}d_{c2}&{\leq}&\text{\rm r.h.s. of \eqref{eq:dc1y2caseI}},\IEEEyessubnumber\label{eq:dcsy2caseI}\\
d_{p2}&{=}&(M_2{-}N_1)A_2{+}\nonumber\\
&&(N_1{+}N_2^\prime{-}M_2)(A_2{-}\alpha_1)^+.\IEEEyessubnumber\label{eq:dp2caseI}
\end{IEEEeqnarray}

Next, let us proceed to discuss the achievability of the corner points in Figure \ref{fig:M2leqN2} when $M_2{\leq}N_2$ and the corner points in Figure \ref{fig:I2b} and \ref{fig:I2a} when $M_2{>}N_2$, because some of the constraints in \eqref{eq:dofcaseI} become inactive in each particular case, which improves the tractability of the analysis.
\subsubsection{Case I.1: $M_1{\geq}N_2$ and $M_2{\leq}N_2$}
In this case, we have $N_2^\prime{=}\min\{M_2{,}N_2\}{=}M_2$. It can be shown that the r.h.s. of \eqref{eq:dc2y2caseI} is greater than or equal to the r.h.s. of \eqref{eq:dc2y1caseI} for any values of $A_1$ and $A_2$. Therefore, \eqref{eq:dc1y2caseI}, \eqref{eq:dc2y2caseI} and \eqref{eq:dcsy2caseI} become inactive. In this way, from \eqref{eq:dc1y1caseI}, \eqref{eq:dc2y1caseI} and \eqref{eq:dcsy1caseI}, we can see that if only $\mathbf{c}_1$ is transmitted (i.e.,  $d_{c2}{=}0$), we achieve
\begin{multline}\label{eq:pair1caseI1}
  (d_{c1}{+}d_{p1}{,}d_{p2}){=}\left(N_1(1{-}(A_2{-}\alpha_1)^+){,}\right. \\
  \left.M_2{-}N_1)A_2{+}N_1(A_2{-}\alpha_1)^+\right).
\end{multline}
If only $\mathbf{c}_2$ is transmitted (i.e.,  $d_{c1}{=}0$), we achieve
\begin{multline}\label{eq:pair2caseI1}
  (d_{p1}{,}d_{c2}{+}d_{p2}){=}\left((M_1{-}N_2)(A_1{-}(A_2{-}\alpha_1)^+)^+{,}\right. \\
  \left.N_1{+}(M_2{-}N_1)A_2{-}(M_1{-}N_2)(A_1{-}(A_2{-}\alpha_1)^+)^+\right).
\end{multline}

Clearly, the DoF pairs in \eqref{eq:pair1caseI1} and \eqref{eq:pair2caseI1} yield the sum DoF $N_1{+}(M_2{-}N_1)A_2$. Choosing $A_2{=}1$ yields the maximum sum DoF $d_1{+}d_2{=}M_2$. With the power levels $A_2{=}1$ and $A_1{\leq}1{-}\alpha_1$, the corner points $\mathcal{P}_{12}{=}(N_1\alpha_1{,}M_2{-}N_1\alpha_1)$ and $\mathcal{P}_{10^\prime}{=}(0{,}M_2)$ in Figure \ref{fig:M2leqN2} are obtained using \eqref{eq:pair1caseI1} and \eqref{eq:pair2caseI1}, respectively. Besides, substituting $A_2{=}\alpha_1$ into \eqref{eq:pair1caseI1} yields the corner point $\mathcal{P}_{20}{=}(N_1{,}(M_2{-}N_1)\alpha_2)$ illustrated in Figure \ref{fig:M2leqN2}. Linking $\mathcal{P}_{20}$ with $\mathcal{P}_{12}$ yields $L_{2{,}IC_1}$ in Proposition \ref{prop:IC1} (see Figure \ref{fig:M2leqN2}).

\subsubsection{Case I.2: $M_1{\geq}N_2$ and $M_2{\geq}N_2$}
In this case, we have $N_2^\prime{=}\min\{M_2{,}N_2\}{=}N_2$ and the r.h.s. of \eqref{eq:dc2y2caseI} becomes equal to the r.h.s. of \eqref{eq:dc1y2caseI} and the r.h.s. of \eqref{eq:dcsy2caseI}. Therefore, it is similar to the BC case (see \eqref{eq:dofnew}), namely that the DoF of the common messages, i.e., $d_{c1}$ and $d_{c2}$, are subject to the sum DoF constraints \eqref{eq:dcsy1caseI} and \eqref{eq:dcsy2caseI}. Then, we derive the achievable DoF region following the footsteps in the BC case.

\begin{enumerate}
  \item When $\Phi_{IC}{\leq}0$, we find that the maximum sum DoF without space-time transmission is achieved with $A_1{=}\alpha_2$ and $A_1{=}\frac{N_2{-}N_1{+}(M_1{-}N_2)\alpha_2}{M_2{-}N_1}{\leq}\alpha_1$. Plugging $A_1^*{=}\alpha_2$ and $A_2^*{=}\alpha_1^\prime$ into \eqref{eq:dofcaseI}, we can see if only $\mathbf{c}_1$ is transmitted (i.e., setting $d_{c2}{=}0$), $\mathcal{P}_{10}{=}(N_1{,}(M_2{-}N_1)\alpha_1^\prime)$ in Figure \ref{fig:I2b} is achievable; if only $\mathbf{c}_2$ is transmitted (i.e., setting $d_{c1}{=}0$), $\mathcal{P}_{10^\prime}{=}((M_1{-}N_2)\alpha_2{,}N_2)$ in Figure \ref{fig:I2b} is achievable.
  \item When $\Phi_{IC}{>}0$ and $\frac{M_1{-}M_2}{M_1{-}N_2}\alpha_1{\leq}1{-}\alpha_2$, we perform a Space-Time transmission, where the power exponents are $(A_1{,}A_2){=}(\alpha_2{,}\alpha_1)$ for a fraction $\rho$ of the total time, while the power exponents are $(A_1{,}A_2){=}(\alpha_2{,}1)$ for the rest of the time. The sum DoF is maximized by choosing the optimal $\rho{=}\rho_{IC}^*$ such that the common message decodabilities at the two receivers are balanced (focusing on the aggregate received signals). We present the value of $\rho_{IC}^*$ as
      {\small
      \begin{equation}
      \rho_{IC}^*{=}\frac{(M_2{-}N_1)(1{-}\alpha_1){-}(M_1{-}N_2)(1{-}\alpha_2){+}M_1{-}M_2}
      {(N_2{-}N_1)(1{-}\alpha_1){+}(M_1{-}N_2)(\alpha_2{-}(\alpha_2{+}\alpha_1{-}1)^+)}.\label{eq:TIC}
      \end{equation}}
      Then, the achievable DoF tuple can be obtained by
      \begin{IEEEeqnarray}{rcl}
      d_{p1{,}ST}&{=}&(M_1{-}N_2)\left[\rho_{IC}^*(\alpha_1{+}\alpha_2{-}1)^+{+}\right.\nonumber\\
      &&\left.(1{-}\rho_{IC}^*)\alpha_2\right],
      \IEEEyesnumber\IEEEyessubnumber\label{eq:dp1_T}\\        d_{p2{,}ST}&{=}&\rho_{IC}^*\left(N_2{-}(N_1{+}N_1{-}M_2)\alpha_1\right){+}\nonumber\\
      &&(1{-}\rho_{IC}^*)(M_2{-}N_1)\alpha_1,
      \IEEEyessubnumber\label{eq:dp2_T}\\
      d_{c1{,}ST}{+}d_{c2{,}ST}&{=}&\rho_{IC}^*(N_2{+}N_1{-}M_2)\alpha_1{+}\nonumber\\
      &&(1{-}\rho_{IC}^*)\left(N_2{-}(M_2{-}N_1)\alpha_1\right).\IEEEyessubnumber\label{eq:dcs2_T}
      \end{IEEEeqnarray}

      To be complete, the corner point $\mathcal{P}_{20}{=}(N_1{,}(M_2{-}N_1)\alpha_1)$ in Figure \ref{fig:I2a} is achieved by plugging $A_k{=}\alpha_j$ into \eqref{eq:dc1y1caseI}, \eqref{eq:dc1y2caseI}, \eqref{eq:dp1caseI} and \eqref{eq:dp2caseI}, and considering that only $\mathbf{c}_1$ is transmitted. Linking $\mathcal{P}_{20}$ and $\mathcal{P}_{12}$ yields $L_2$ in Proposition \ref{prop:IC1}.

  \item When $\Phi_{IC}{\geq}0$ and $\frac{M_1{-}M_2}{M_1{-}N_2}\alpha_1{>}1{-}\alpha_2$, Rx2 has a greater common-message-decodability than Rx1 with both of the power exponents $(A_1{,}A_2){=}(\alpha_2{,}\alpha_1)$ and $(A_1{,}A_2){=}(\alpha_2{,}1)$. This fact prevents the Space-Time transmission from benefiting the sum DoF performance. In this case, using \eqref{eq:dofcaseI}, we learn that the maximum sum DoF can be achieved choosing $A_1^*{=}1{-}\frac{M_1{-}M_2}{M_1{-}N_2}\alpha_1{<}\alpha_2$ and $A_2^*{=}1$. With that power allocation policy, $\alpha_{0{,}IC}$ (the third line in \eqref{eq:ICalpha0}) is immediate and the corner point $\mathcal{P}_{12}$ (resp. $\mathcal{P}_{10^\prime}$) is obtained if only $\mathbf{c}_1$ (resp. $\mathbf{c}_2$) is transmitted. To be complete, the achievability of the corner point $\mathcal{P}_{20}$ follows that in the case when $\Phi_{IC}{\geq}0$.
\end{enumerate}

\subsection{Case II: $M_1{\leq}N_2$}\label{sec:icII}
\begin{figure}[t]
\renewcommand{\captionfont}{\small}
\captionstyle{center}
\centering
\includegraphics[width=7cm,height=3cm]{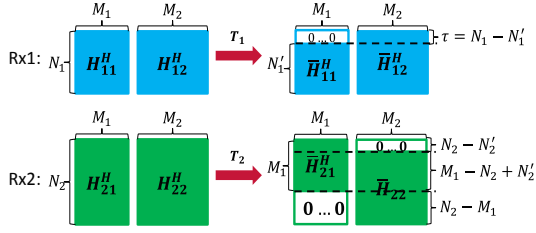}
\caption{Visualization of the linear transform of the channel matrices}\label{fig:LT}
\end{figure}
In this part, we firstly modify the RS scheme proposed in the above subsection based on the dimension of the column space of the channel matrices, and secondly show the achievability of the corner points on the DoF region stated in Proposition \ref{prop:IC2}. For notation convenience, we introduce $N_k^\prime{\triangleq}\min\{M_k{,}N_k\}$ and $N_k^{\prime\prime}{\triangleq}\max\{M_k{,}N_k\}{,}k{=}1{,}2$.

The antenna configuration $M_1{\leq}N_2$ makes a difference from the BC case that, in the received signals, the messages intended for Rx2 only partially overlap with the messages intended for Rx1. To be more specific, as inspired by \cite[Appendix A and Figure 2]{xinping_mimo}, let us identify this fact by performing a row transformation to the channel matrices. Since $\mathbf{H}_{kj}$ and $\mathbf{H}_{kk}$, $k{,}j{=}1{,}2{,}k{\neq}j$, are mutually independent, there exists an invertible row transformation $\mathbf{T}_k{\in}\mathbb{C}^{N_k{\times}N_k}$ that converts the $N_k{\times}(M_1{+}M_2)$ matrix $\mathbf{H}_k^H{\triangleq}\left[\mathbf{H}_{k1}^H{,}\mathbf{H}_{k2}^H\right]$ to
\begin{IEEEeqnarray}{rcl}
\!\!\!\!\!\!\!\!\!\!\mathbf{T}_1\mathbf{H}_{11}^H{=}\left[\!\!\begin{array}{c}\mathbf{0}_{(N_1{-}N_1^\prime){\times}M_1}\\ \bar{\mathbf{H}}_{11}^H\end{array}\!\!\right]\!\!,&\,&\mathbf{T}_1\mathbf{H}_{12}^H{=}\bar{\mathbf{H}}_{12}^H; \IEEEyesnumber\IEEEyessubnumber\label{eq:T1}\\
\!\!\!\!\!\!\!\!\!\!\mathbf{T}_2\mathbf{H}_{21}^H{=}\left[\!\!\begin{array}{c}\bar{\mathbf{H}}_{21}^H\\ \mathbf{0}_{(N_2{-}M_1){\times}M_1}\end{array}\!\!\right]\!\!,&\,&
\mathbf{T}_2\mathbf{H}_{22}^H{=}\left[\!\!\begin{array}{c}\mathbf{0}_{(N_2{-}N_2^\prime){\times}M_2}\\ \bar{\mathbf{H}}_{22}^H\end{array}\!\!\right]\!\!,\IEEEyessubnumber\label{eq:T2}
\end{IEEEeqnarray}
where $\bar{\mathbf{H}}_{11}{\in}\mathbb{C}^{N_1^\prime{\times}M_1}$, $\bar{\mathbf{H}}_{12}{\in}\mathbb{C}^{N_1{\times}M_2}$, $\bar{\mathbf{H}}_{21}{\in}\mathbb{C}^{M_1{\times}M_1}$ and $\bar{\mathbf{H}}_{22}{\in}\mathbb{C}^{N_2^\prime{\times}M_2}$ are full rank almost surely. Therefore, at Rx2, the dimension of the overlapping part between $\bar{\mathbf{H}}_{21}^H$ and $\bar{\mathbf{H}}_{22}^H$ is $M_1{+}N_2^\prime{-}N_2$, while the dimension of the subspace of $\bar{\mathbf{H}}_{22}^H$ that does not overlap with $\bar{\mathbf{H}}_{21}^H$ is $N_2{-}M_1$. Note that the row transformation $\mathbf{T}_2$ is designed such that the dimension of the overlapping part of $\bar{\mathbf{H}}_{22}^H$ and $\bar{\mathbf{H}}_{21}^H$ is minimized. At Rx1, the dimension of the overlapping part between $\bar{\mathbf{H}}_{11}^H$ and $\bar{\mathbf{H}}_{12}^H$ is $N_1^\prime$. Figure \ref{fig:LT} provides an illustrative view of this linear transformation.

Motivated by this, we modify the RS scheme proposed in Section \ref{sec:icI} by choosing different power levels for the private messages of Rx2 interfering or not interfering with the signal from Tx1. Specifically,
\begin{itemize}
\item $\tau{\triangleq}N_1{-}N_1^\prime$ private messages are delivered to Rx2 in the subspace of $\hat{\mathbf{H}}_{22}^H$ using full power without impacting the signal sent from Tx1, as they are received by Rx1 via the part of $\bar{\mathbf{H}}_{12}^H$ that does not overlap with $\bar{\mathbf{H}}_{11}^H$, and received by Rx2 via the part of $\bar{\mathbf{H}}_{22}^H$ that does not overlap with $\bar{\mathbf{H}}_{21}^H$;
\item $\min\{N_2^\prime{-}\tau{,}M_2{-}N_1\}$ private messages are transmitted to Rx2 via ZFBF. They are divided into two parts: 1) $\mu_1{\triangleq}\min\{N_2{-}M_1{-}\tau{,}M_2{-}N_1\}$ of them are delivered using power level $A_2^\prime$, and are received by Rx2 via the part of $\bar{\mathbf{H}}_{22}$ that does not overlap with $\bar{\mathbf{H}}_{21}$, and 2) the remaining $\mu_2{\triangleq}\min\{N_2^\prime{-}\tau{,}M_2{-}N_1\}{-}\mu_1$ ZF-precoded private messages are delivered using power level $A_2$;
\item The remaining $N_2^\prime{-}\tau{-}\mu_1{-}\mu_2$ private messages are delivered to Rx2 in the subspace of $\hat{\mathbf{H}}_{22}^H$. Similar to the ZF-precoded private messages, they are also divided into two parts: 1) $\delta_1{\triangleq}N_2{-}M_1{-}\tau{-}\mu_1$ of them are delivered using power level $A_2^\prime{-}\alpha_1$, and are received by Rx2 via the part of $\bar{\mathbf{H}}_{22}$ that does not overlap with $\bar{\mathbf{H}}_{21}$, and 2) the remaining $\delta_2{\triangleq}N_2^\prime{-}N_2{+}M_1{-}\mu_2$ private messages are delivered using power level $(A_2{-}\alpha_1)^+$;
\end{itemize}
The power levels are defined to be $A_2{\in}[0{,}A_2^\prime]$ and $A_2^\prime{\in}[\alpha_1{,}1]$. Moreover, there is no ZF-precoded private messages delivered to Rx1 as $M_1{\leq}N_2$. 

Notably, when $M_1{=}N_2$, the above private messages categorization becomes the transmission block designed for the case $M_1{\geq}N_2$. Specifically, when $M_1{=}N_2$, since $\mathbf{H}_{11}^H$ and $\mathbf{H}_{21}^H$ have full row rank, there is no all-zero rows in $\mathbf{T}_1\mathbf{H}_{11}^H$ and $\mathbf{T}_2\mathbf{H}_{21}^H$, which leads to $\tau{=}\mu_1{=}\delta_1{=}0$. Moreover, we have $\mu_2{=}M_2{-}N_1$ and $\delta_2{=}N_2^\prime{+}N_1{-}M_2$, corresponding to the number of messages in $\mathbf{u}_2^{(1)}$ and $\mathbf{u}_2^{(2)}$ in \eqref{eq:ICs}, respectively. In the received signals \eqref{eq:ICy}, all these $N_2^\prime$ private messages align with the signal from Tx1.

Consequently, let us write the transmitted signals as
\begin{IEEEeqnarray}{rcl}
\mathbf{s}_1&{=}&\underbrace{\mathbf{c}_1}_{P},\IEEEyesnumber\IEEEyessubnumber\label{eq:s1caseII}\\
\mathbf{s}_2&{=}&\underbrace{\mathbf{c}_2}_{P}{+}\underbrace{\mathbf{V}_2^{(1)}\mathbf{u}_2^{(1)}}_{P}{+}
\underbrace{\mathbf{V}_2^{(2)}\mathbf{u}_2^{(2)}}_{P^{A_2^\prime}}{+}\underbrace{\mathbf{V}_2^{(3)}\mathbf{u}_2^{(3)}}_{P^{A_2}}{+}\nonumber\\
&&\underbrace{\mathbf{V}_2^{(4)}\mathbf{u}_2^{(4)}}_{P^{A_2^\prime{-}\alpha_1}}{+} \underbrace{\mathbf{V}_2^{(5)}\mathbf{u}_2^{(5)}}_{P^{(A_2{-}\alpha_1)^+}},
\IEEEyessubnumber\label{eq:s2caseII}
\end{IEEEeqnarray}
where $\mathbf{u}_2^{(1)}{\in}\mathbb{C}^{\tau{\times}1}$, $\mathbf{u}_2^{(2)}{\in}\mathbb{C}^{\mu_1{\times}1}$, $\mathbf{u}_2^{(3)}{\in}\mathbb{C}^{\mu_2{\times}1}$, $\mathbf{u}_2^{(4)}{\in}\mathbb{C}^{\delta_1{\times}1}$ and $\mathbf{u}_2^{(5)}{\in}\mathbb{C}^{\delta_2{\times}1}$. The precoders, $\mathbf{V}_2^{(1)}{\in}\mathbb{C}^{M_2{\times}\tau}$, $\mathbf{V}_2^{(4)}{\in}\mathbb{C}^{M_2{\times}\delta_1}$ and $\mathbf{V}_2^{(5)}{\in}\mathbb{C}^{M_2{\times}\delta_2}$ are in the subspace of $\hat{\mathbf{H}}_{22}$, while $\mathbf{V}_2^{(2)}{\in}\mathbb{C}^{M_2{\times}\mu_1}$ and $\mathbf{V}_2^{(3)}{\in}\mathbb{C}^{M_2{\times}\mu_2}$ are ZF-precoders in the subspace of $\hat{\mathbf{H}}_{12}^\bot$. Note that all the precoders have full rank and linearly independent of each other. The received signals are expressed as
{\small\begin{IEEEeqnarray}{rcl}
\mathbf{y}_1&{=}&\underbrace{\mathbf{H}_{11}^H\mathbf{c}_1}_{P}{+}\underbrace{\mathbf{H}_{12}^H\mathbf{c}_2}_{P}{+}
\underbrace{\mathbf{H}_{12}^H\mathbf{V}_2^{(1)}\mathbf{u}_2^{(1)}}_{P}{+}
\underbrace{\boldsymbol\eta_1^\prime}_{P^{A_2^\prime{-}\alpha_1}}{+}\!\!\!\!
\underbrace{\boldsymbol\eta_1,}_{P^{(A_2{-}\alpha_1)^+}}{,}\IEEEyesnumber\IEEEyessubnumber\label{eq:y1caseII}\\
\mathbf{y}_2&{=}&\underbrace{\mathbf{H}_{21}^H\mathbf{c}_1}_{P}{+}\underbrace{\mathbf{H}_{22}^H\mathbf{c}_2}_{P}{+}
\underbrace{\mathbf{H}_{22}^H\mathbf{V}_2^{(1)}\mathbf{u}_2^{(1)}}_{P}{+}
\underbrace{\mathbf{H}_{22}^H\mathbf{V}_2^{(2)}\mathbf{u}_2^{(2)}}_{P^{A_2^\prime}}{+}\nonumber\\
&&\underbrace{\mathbf{H}_{22}^H\mathbf{V}_2^{(4)}\mathbf{u}_2^{(4)}}_{P^{A_2^\prime{-}\alpha_1}}{+}
\underbrace{\mathbf{H}_{22}^H\mathbf{V}_2^{(3)}\mathbf{u}_2^{(3)}}_{P^{A_2}}{+}
\underbrace{\mathbf{H}_{22}^H\mathbf{V}_2^{(5)}\mathbf{u}_2^{(5)}}_{P^{(A_2{-}\alpha_1)^+}},\IEEEyessubnumber\label{eq:y1caseII}
\end{IEEEeqnarray}}
where,
\begin{IEEEeqnarray}{rcl}
\boldsymbol\eta_1^\prime&{\triangleq}&\mathbf{H}_{12}^H\left(\mathbf{V}_2^{(2)}\mathbf{u}_2^{(2)}{+}
\mathbf{V}_2^{(4)}\mathbf{u}_2^{(4)}\right),\\
\boldsymbol\eta_1&{\triangleq}&\mathbf{H}_{12}^H\left(\mathbf{V}_2^{(3)}\mathbf{u}_2^{(3)}{+}
\mathbf{V}_2^{(5)}\mathbf{u}_2^{(5)}\right).
\end{IEEEeqnarray}

We can see that $\mathbf{u}_2^{(2)}$ and $\mathbf{u}_2^{(4)}$ are received by Rx1 with the power level $A_2^\prime{-}\alpha_1$. If $A_2{\geq}\alpha_1$, $\mathbf{u}_2^{(3)}$ and $\mathbf{u}_2^{(5)}$ are received by Rx1 with the power level $A_2{-}\alpha_1$, otherwise $\mathbf{u}_2^{(5)}$ is not transmitted and $\mathbf{u}_2^{(3)}$ is drowned by the noise due to ZFBF with imperfect CSIT. Then, following the general proof in Appendix A, the achievable DoF tuple lies in
\begin{IEEEeqnarray}{rcl}\label{eq:dofcaseII}
\text{\rm At Rx1:}\quad
d_{c1}&{\leq}&N_1^\prime{-}\xi(A_2^\prime{-}\alpha_1){-}
(N_1^\prime{-}\xi)(A_2{-}\alpha_1)^+,\IEEEyesnumber\IEEEyessubnumber\label{eq:dc1y1caseII}\\
d_{c2}&{\leq}&\text{\rm r.h.s. of \eqref{eq:dc1y1caseII}},\IEEEyessubnumber\label{eq:dc2y1caseII}\\
d_{c1}{+}d_{c2}&{\leq}&\text{\rm r.h.s. of \eqref{eq:dc1y1caseII}},\IEEEyessubnumber\label{eq:dcsy1caseII}\\
\text{\rm At Rx2:}\quad
d_{c1}&{\leq}&M_1{-}\mu_2A_2{-}\delta_2(A_2{-}\alpha_1)^+,\IEEEyessubnumber\label{eq:dc1y2caseII}\\
d_{c2}&{\leq}&N_2^\prime{-}\mu_2A_2{-}\delta_2(A_2{-}\alpha_1)^+{-}\nonumber\\
&&\mu_1A_2^\prime{-}\delta_1(A_2^\prime{-}\alpha_1){-}\tau,\IEEEyessubnumber\label{eq:dc2y2caseII}\\
d_{c1}{+}d_{c2}&{\leq}&N_2{-}\mu_2A_2{-}\delta_2(A_2{-}\alpha_1)^+{-}\nonumber\\
&&\mu_1A_2^\prime{-}\delta_1(A_2^\prime{-}\alpha_1){-}\tau,\IEEEyessubnumber\label{eq:dcsy2caseII}\\
d_{p2}&{=}&\mu_2A_2{+}\delta_2(A_2{-}\alpha_1)^+\!{+}\mu_1A_2^\prime{+}\nonumber\\
&&\delta_1(A_2^\prime{-}\alpha_1){+}\tau,\IEEEyessubnumber\label{eq:dp2caseII}
\end{IEEEeqnarray}
where $\xi{\triangleq}\min\{N_1^\prime{,}\mu_1{+}\delta_1\}$.

Notably, unlike the BC case and IC Case I where $M_1{\geq}N_2$, from \eqref{eq:dc1y2caseII} and \eqref{eq:dc2y2caseII}, we see that Rx2 has different common-message-decodabilities of $\mathbf{c}_1$ and $\mathbf{c}_2$. Hence, it is not suitable to perform analysis focusing on the sum DoF. Instead, in the following, using the set of constraints stated in \eqref{eq:dofcaseII}, we characterize the achievable DoF region stated in Proposition \ref{prop:IC2} by finding the maximum $d_2{\triangleq}d_{c2}{+}d_{p2}$ for a given $d_{c1}{=}\lambda$, where $\lambda{\in}[0{,}N_1^\prime]$. Specifically, the optimization problem can be formulated as
\begin{IEEEeqnarray}{rcl}\label{eq:optd2}
\!\!\!\!\!\!\!\!\max_{A_2{,}A_2^\prime{,}d_{c2}}&\,& d_{c2}{+}d_{p2}\IEEEyesnumber\IEEEyessubnumber\\
\!\!\!\!\!\!\!\!\text{s.t.}&\,&d_{c2}{\leq}N_1^\prime{-}\xi(A_2^\prime{-}\alpha_1){-}
(N_1^\prime{-}\xi)(A_2{-}\alpha_1)^+{-}\lambda,\IEEEyessubnumber\label{eq:dc2cons1}\\
\!\!\!\!\!\!\!\!&&d_{c2}{\leq}\min\{N_2^\prime{,}N_2{-}\lambda\}{-}\mu_2A_2{-}\delta_2(A_2{-}\alpha_1)^+{-}\nonumber\\
\!\!\!\!\!\!\!\!&&\mu_1A_2^\prime{-}\delta_1(A_2^\prime{-}\alpha_1){-}\tau,\IEEEyessubnumber\label{eq:dc2cons2}\\
\!\!\!\!\!\!\!\!&&\lambda{\leq}N_1^\prime{-}\xi(A_2^\prime{-}\alpha_1){-}(N_1^\prime{-}\xi)(A_2{-}\alpha_1)^+,\IEEEyessubnumber\label{eq:lambdacons1}\\
\!\!\!\!\!\!\!\!&&\lambda{\leq}M_1{-}\mu_2A_2{-}\delta_2(A_2{-}\alpha_1)^+,\IEEEyessubnumber\label{eq:lambdacons2}\\
\!\!\!\!\!\!\!\!&&0{\leq}A_2{\leq}A_2^\prime,\IEEEyessubnumber\\
\!\!\!\!\!\!\!\!&&\alpha_1{\leq}A_2^\prime{\leq}1,\IEEEyessubnumber
\end{IEEEeqnarray}
where $d_{p2}$ is given in \eqref{eq:dp2caseII}, while \eqref{eq:dc2cons2} is obtained due to \eqref{eq:dc2y2caseII} and \eqref{eq:dcsy2caseII}. To find the closed-form solution of this linear programme, we proceed the discussion by considering Case II.1, i.e., $M_2{\leq}N_2$, and Case II.2, i.e., $M_2{\geq}N_2$, because some of the constraints in \eqref{eq:optd2} become inactive in each particular case, which simplifies the derivation.

\subsubsection{Case II.1, $M_2{\leq}N_2$ and $M_1{\leq}N_2$}
In this case, using the fact that $A_2{\leq}A_2^\prime{\leq}1$, it can be verified that constraints \eqref{eq:dc2cons2} and \eqref{eq:lambdacons2} are redundant compared to \eqref{eq:dc2cons1} and \eqref{eq:lambdacons1}, respectively. Moreover, as the objective function is monotonically increasing with $d_{c2}$, we can see the optimal solution is taken when \eqref{eq:dc2cons1} is active. Hence, the optimization problem stated in \eqref{eq:optd2} becomes
\begin{IEEEeqnarray}{rcl}\label{eq:optd2_1}
\max_{A_2{,}A_2^\prime}&\quad& d_{2{,}(1)}(A_2{,}A_2^\prime{,}\lambda)\IEEEyesnumber\IEEEyessubnumber\\
\text{s.t.}&\quad&\lambda{\leq}N_1^\prime{-}\xi(A_2^\prime{-}\alpha_1){-}(N_1^\prime{-}\xi)(A_2{-}\alpha_1)^+,
\IEEEyessubnumber\label{eq:lambdacons1_1}\\
&&0{\leq}A_2{\leq}A_2^\prime,\IEEEyessubnumber\\
&&\alpha_1{\leq}A_2^\prime{\leq}1,\IEEEyessubnumber
\end{IEEEeqnarray}
where
\begin{multline}\label{eq:d2caseII1}
  d_{2{,}(1)}(A_2{,}A_2^\prime{,}\lambda){=}N_1^\prime{-}\lambda{+}(\delta_2{-}N_1^\prime{+}\xi)(A_2{-}\alpha_1)^+{+} \\
  \mu_2A_2{+}\mu_1A_2^\prime{+}(\delta_1{-}\xi)(A_2^\prime{-}\alpha_1){+}\tau,
\end{multline}
is obtained by summing \eqref{eq:dp2caseII} and \eqref{eq:dc2cons1}.

Since the objective function \eqref{eq:d2caseII1} is linearly increasing with $A_2$ and $A_2^\prime$, the optimal solution is obtained when (at least) two of the constraints \eqref{eq:lambdacons1_1}, $A_2{\leq}A_2^\prime$ and $A_2^\prime{\leq}1$ are active. Through some simple calculation, the closed-form solution, i.e., $(A_2^*{,}A_2^{\prime*})$, and the resultant maximum DoF of Rx2, i.e., $d_{2{,}(1)}(A_2^*{,}A_2^{\prime*}{,}\lambda)$ write as

\underline{\emph{For $\lambda{\in}\left[0{,}N_1^\prime\alpha_1\right]$,}}
\begin{IEEEeqnarray}{rcl}
(A_2^*{,}A_2^{\prime*})&{=}&(1{,}1),\IEEEyesnumber\IEEEyessubnumber\\
d_{2{,}(1)}(1{,}1{,}\lambda)&{=}&N_2^\prime{-}\lambda.\IEEEyessubnumber\label{eq:dp2_11}
\end{IEEEeqnarray}

\underline{\emph{For $\lambda{\in}\left[N_1^\prime\alpha_1{,}N_1^\prime\right]$,}}
\begin{IEEEeqnarray}{rcl}
\!\!\!\!(A_2^*{,}A_2^{\prime*})&{=}&\left(\frac{N_1^\prime{-}\lambda{+}N_1^\prime\alpha_1}{N_1^\prime}{,}
\frac{N_1^\prime{-}\lambda{+}N_1^\prime\alpha_1}{N_1^\prime}\right),\IEEEyesnumber\IEEEyessubnumber\\
\!\!\!\!d_{2{,}(1)}(A_2^*{,}A_2^{\prime*}{,}\lambda)&{=}&
M_2{+}(M_2{-}N_1)\alpha_1{-}\nonumber\\
&&\frac{M_2{-}N_1{+}N_1^\prime}{N_1^\prime}\lambda.\IEEEyessubnumber\label{eq:dp2_12}
\end{IEEEeqnarray}

It can be shown that the DoF pair $(\lambda{,}d_{2{,}(1)}(A_2^*{,}A_2^{\prime*}{,}\lambda))$ with $d_{2{,}(1)}(A_2^*{,}A_2^{\prime*}{,}\lambda)$ in \eqref{eq:dp2_11} and \eqref{eq:dp2_12} lie on $L_1$ and $L_2$ in Proposition \ref{prop:IC2}, respectively. When $\lambda{=}N_1^\prime\alpha_1$ and $\lambda{=}N_1^\prime$, we have the corner points $\mathcal{P}_{10}{=}(N_1^\prime{,}(M_2{-}N_1)\alpha_1)$ and $\mathcal{P}_{12}{=}(N_1^\prime\alpha_1{,}N_2^\prime{-}N_1^\prime\alpha_1)$ in Figure \ref{fig:M2leqN2}, respectively.

\subsubsection{Case II.2, $M_2{\geq}N_2$ and $M_1{\leq}N_2$}\label{sec:II2}
\begin{table*}[t]
\renewcommand{\captionfont}{\small}
\footnotesize
\captionstyle{center} \centering
\begin{tabular}{c|l|c|c|c}
& Conditions & $d_{2{,}(2)}(A_2^*{,}A_2^{\prime*}{,}\lambda)$ & $N_2{\geq}M_1{+}N_1$ & $N_2{\leq}M_1{+}N_1$\\
\hline
$\alpha_1{\leq}\frac{M_1{-}N_1^\prime}{\mu_2}$ & A: not hold & & &\\
(Not applicable for $M_1{\leq}N_1$) & B: not hold & & &\\
& C: for $\lambda{\in}\left[\frac{N_1^\prime\mu_2\alpha_1}{M_1{-}N_1^\prime}{,}N_1^\prime\right]$ & eq.\eqref{eq:dp2_C} & $L_2$ & $L_2$\\
& D: for $\lambda{\in}\left[\frac{\mu_2N_1^\prime{+}\delta_2\xi}{M_1{-}N_1^\prime{+}\xi}\alpha_1{,}
\frac{N_1^\prime\mu_2\alpha_1}{M_1{-}N_1^\prime}\right]$ & eq.\eqref{eq:dp2_D} & $L_3$ & $L_4$\\
& E: for $\lambda{\in}\left[\delta_2\alpha_1{,}\frac{\mu_2N_1^\prime{+}\delta_2\xi}{M_1{-}N_1^\prime{+}\xi}\alpha_1\right]$ & eq.\eqref{eq:dp2_E} & $L_1$ & $L_1$\\
& F: for $\lambda{\in}\left[0{,}\delta_2\alpha_1\right]$ & eq.\eqref{eq:dp2_F} & $L_1$ & $L_1$\\
\hline
$\frac{M_1{-}N_1^\prime}{\mu_2}{\leq}\alpha_1{\leq}\frac{M_1{-}N_1^\prime{+}\xi}{\mu_2{+}\xi}$ & A: for $\lambda{\in}\left[M_1{-}\mu_2\alpha_1{,}N_1^\prime\right]$ & eq.\eqref{eq:dp2_A} & $L_2$ & $L_5$\\
& B: not hold & & &\\
& C: not hold & & &\\ & D: for $\lambda{\in}\left[\frac{\mu_2N_1^\prime{+}\delta_2\xi}{M_1{-}N_1^\prime{+}\xi}\alpha_1{,}
M_1{-}\mu_2\alpha_1\right]$ & eq.\eqref{eq:dp2_D} & $L_3$ & $L_4$\\
& E: for $\lambda{\in}\left[\delta_2\alpha_1{,}\frac{\mu_2N_1^\prime{+}\delta_2\xi}{M_1{-}N_1^\prime{+}\xi}\alpha_1\right]$ & eq.\eqref{eq:dp2_E} & $L_1$ & $L_1$\\
& F: for $\lambda{\in}\left[0{,}\delta_2\alpha_1\right]$ & eq.\eqref{eq:dp2_F} & $L_1$ & $L_1$\\
\hline
$\alpha_1{\geq}\frac{M_1{-}N_1^\prime{+}\xi}{\mu_2{+}\xi}$ & A: for $\lambda{\in}\left[N_1^\prime{-}\xi(1{-}\alpha_1){,}N_1^\prime\right]$ & eq.\eqref{eq:dp2_A} & $L_2$ & $L_5$\\
& B: for $\lambda{\in}\left[M_1{-}\mu_2\alpha_1{,}N_1^\prime{-}\xi(1{-}\alpha_1)\right]$ & eq.\eqref{eq:A2A2pstarB} & $L_1$ & $L_1$\\
& C: not hold & & & \\
& D: not hold & & & \\
& E: for $\lambda{\in}\left[\delta_2\alpha_1{,}M_1{-}\mu_2\alpha_1\right]$ & eq.\eqref{eq:dp2_E} & $L_1$ & $L_1$\\
& F: for for $\lambda{\in}\left[0{,}\delta_2\alpha_1\right]$ & eq.\eqref{eq:dp2_F} & $L_1$ & $L_1$\\
\end{tabular}
\caption{Achievability of the weighted-sum constraints in Case II.2.a and II.2.b}\label{tab:dofpairs}
\end{table*}
In this case, we perform the same derivation as in Case II.1 by taking $d_{c2}$ equal to the minimum of r.h.s. of \eqref{eq:dc2cons1} and \eqref{eq:dc2cons2}, because the objective function in \eqref{eq:optd2} is monotonically increasing with $d_{c2}$. Then, the optimization problem can be reformulated as
\begin{IEEEeqnarray}{rcl}\label{eq:optd2_2}
\max_{A_2{,}A_2^\prime}&\quad& d_{2{,}(2)}(A_2{,}A_2^\prime{,}\lambda)\IEEEyesnumber\IEEEyessubnumber\\
\text{s.t.}&\quad&
\lambda{\leq}N_1^\prime{-}\xi(A_2^\prime{-}\alpha_1){-}(N_1^\prime{-}\xi)(A_2{-}\alpha_1)^+,\IEEEyessubnumber\label{eq:lambdacons1_2}\\
&&\lambda{\leq}M_1{-}\mu_2A_2{-}\delta_2(A_2{-}\alpha_1)^+,\IEEEyessubnumber\label{eq:lambdacons2_2}\\
&&0{\leq}A_2{\leq}A_2^\prime,\IEEEyessubnumber\\
&&\alpha_1{\leq}A_2^\prime{\leq}1,\IEEEyessubnumber
\end{IEEEeqnarray}
where
\begin{multline}\label{eq:d2caseII2}
d_{2{,}(2)}(A_2{,}A_2^\prime{,}\lambda){=}
\min\left\{N_2{-}\lambda{,}N_1^\prime{-}\lambda{+}\mu_2A_2{+}\mu_1A_2^\prime{+}\right.\\
\left.(\delta_2{-}N_1^\prime{+}\xi)(A_2{-}\alpha_1)^+{+}(\delta_1{-}\xi)(A_2^\prime{-}\alpha_1){+}\tau\right\},
\end{multline}
is obtained by summing \eqref{eq:dp2caseII} and the minimum of \eqref{eq:dc2cons1} and \eqref{eq:dc2cons2}.

Following the derivations in Appendix B, the closed-form solution, i.e., $(A_2^*{,}A_2^{\prime*})$, and the resultant maximum DoF of Rx2, i.e., $d_{2{,}(2)}(A_2^*{,}A_2^{\prime*}{,}\lambda)$, write in the following six conditions:
\begin{enumerate}
  \item[A)] For $\lambda{\in}\left[\max\{M_1{-}\mu_2\alpha_1{,}N_1^\prime{-}\xi(1{-}\alpha_1)\}{,}N_1^\prime\right]$,
      \begin{IEEEeqnarray}{rcl}
      (A_2^*{,}A_2^{\prime*})&{=}&\left(\frac{M_1{-}\lambda}{\mu_2}{,}\alpha_1{+}\frac{N_1^\prime{-}\lambda}{\xi}\right), \IEEEyesnumber\IEEEyessubnumber\label{eq:A2A2pstarA}\\
      d_{2{,}(2)}(A_2^*{,}A_2^{\prime*}{,}\lambda)&{=}&\max\{N_2{,}M_1{+}N_1\}{+}\nonumber\\&&\mu_1\alpha_1{-}
      \left(1{+}\frac{\mu_1}{\xi}\right)\lambda;\IEEEyessubnumber
      \label{eq:dp2_A}
      \end{IEEEeqnarray}
  \item[B)] For $\lambda{\in}\left[M_1{-}\mu_2\alpha_1{,}N_1^\prime{-}\xi(1{-}\alpha_1)\right]$,
      \begin{IEEEeqnarray}{rcl}
      (A_2^*{,}A_2^{\prime*})&{=}&\left(\frac{M_1{-}\lambda}{\mu_2}{,}1\right), \IEEEyesnumber\IEEEyessubnumber\label{eq:A2A2pstarB}\\
      d_{2{,}(2)}(A_2^*{,}A_2^{\prime*}{,}\lambda)&{=}&N_2{-}\lambda;\IEEEyessubnumber\label{eq:dp2_B}
      \end{IEEEeqnarray}
  \item[C)] For $\lambda{\in}\left[\frac{N_1^\prime\mu_2\alpha_1}{M_1{-}N_1^\prime}{,}\min\{M_1{-}\mu_2\alpha_1{,}N_1^\prime\}\right]$,
      \begin{IEEEeqnarray}{rcl}
      \!\!\!\!\!\!\!\!(A_2^*{,}A_2^{\prime*})&{=}&\left(\frac{N_1^\prime{-}\lambda{+}N_1^\prime\alpha_1}{N_1^\prime}{,}\right.\nonumber\\
      &&\left.\frac{N_1^\prime{-}\lambda{+}N_1^\prime\alpha_1}{N_1^\prime}\right),\IEEEyesnumber\IEEEyessubnumber\label{eq:A2A2pstarC}\\
      \!\!\!\!\!\!\!\!d_{2{,}(2)}(A_2^*{,}A_2^{\prime*}{,}\lambda)&{=}&
      N_2{+}(\mu_1{+}\mu_2)\alpha_1{-}\nonumber\\
      &&\frac{N_2{-}N_1{+}N_1^\prime}{N_1^\prime}\lambda;\IEEEyessubnumber\label{eq:dp2_C}
      \end{IEEEeqnarray}
  \item[D)] For $\lambda{\in}\left[\frac{\mu_2N_1^\prime{+}\delta_2\xi}{M_1{-}N_1^\prime{+}\xi}\alpha_1{,} \min\{M_1{-}\mu_2\alpha_1{,}\frac{N_1^\prime\mu_2\alpha_1}{M_1{-}N_1^\prime}{,}N_1^\prime\}\right]$,
      \begin{IEEEeqnarray}{rcl}
      (A_2^*{,}A_2^{\prime*})&{=}&\left(\frac{M_1{-}\lambda{+}\delta_2\alpha_1}{M_1}{,}
      1{-}\frac{\left(M_1{-}N_1^\prime{+}\xi\right)}{M_1\xi}\lambda{+}\right.\nonumber\\
      &&\left.\frac{\left(\mu_2N_1^\prime{+}\delta_2\xi\right)}{M_1\xi}\alpha_1\right),\IEEEyesnumber\IEEEyessubnumber\label{eq:A2A2pstarD}\\
      d_{2{,}(2)}(A_2^*{,}A_2^{\prime*}{,}\lambda)&{=}&
      N_2{+}\left[1{+}\frac{\mu_2\left(N_1{-}\xi\right)}{M_1\xi}\right]\mu_1\alpha_1{-}\nonumber\\
      &&\left[1{+}\frac{M_1{-}N_1^\prime{+}\xi}{M_1\xi}\mu_1\right]\lambda;\IEEEyessubnumber\label{eq:dp2_D}
      \end{IEEEeqnarray}
  \item[E)] For $\lambda{\in}\left[\delta_2\alpha_1{,} \min\{\frac{\mu_2N_1^\prime{+}\delta_2\xi}{M_1{-}N_1^\prime{+}\xi}\alpha_1{,}M_1{-}\mu_2\alpha_1{,}N_1^\prime\}\right]$,
      \begin{IEEEeqnarray}{rcl}
      (A_2^*{,}A_2^{\prime*})&{=}&\left(\frac{M_1{-}\lambda{+}\delta_2\alpha_1}{M_1}{,}1\right), \IEEEyesnumber\IEEEyessubnumber\label{eq:A2A2pstarE}\\
      d_{2{,}(2)}(A_2^*{,}A_2^{\prime*}{,}\lambda)&{=}&N_2{-}\lambda;\IEEEyessubnumber\label{eq:dp2_E}
      \end{IEEEeqnarray}
  \item[F)] For $\lambda{\in}\left[0{,}\delta_2\alpha_1\right]$,
     \begin{IEEEeqnarray}{rcl}
     (A_2^*{,}A_2^{\prime*})&{=}&(1{,}1), \IEEEyesnumber\IEEEyessubnumber\label{eq:A2A2pstarF}\\
     d_{2{,}(2)}(A_2{,}A_2^\prime{,}\lambda)&{=}&N_2{-}\lambda.\IEEEyessubnumber\label{eq:dp2_F}
     \end{IEEEeqnarray}
\end{enumerate}

To be complete, Table \ref{tab:dofpairs} summarizes the validation of these six conditions for different values of $\alpha_1$, and also present the resultant weighted-sum constraints where the corresponding DoF pair $(\lambda{,}d_{p2}(A_2^*{,}A_2^{\prime*}{,}\lambda))$ lies on for Case II.2.a $N_2{\geq}M_1{+}N_1$ (i.e., $N_1^\prime{\leq}\mu_1$) and Case II.2.b $N_2{\leq}M_1{+}N_1$ (i.e., $N_1^\prime{\geq}\mu_1$). 

\section{Conclusion}\label{sec:conclusion}
In this paper, for the first time in the literature, we characterize achievable DoF regions of a general two receiver $(M{,}N_1{,}N_2)$ MIMO BC and $(M_1{,}M_2{,}N_1{,}N_2)$ MIMO IC with imperfect CSIT, whose error decays with the SNR. Without loss of generality, we consider $N_1{\leq}N_2$. We propose Rate-Splitting schemes suitable for the asymmetric antenna deployment. In BC, compared to the RS scheme designed for the symmetric case, the new ingredients of the scheme lie in 1) delivering additional non-ZF-precoded private symbols to Rx2, and 2) a Space-Time implementation. In IC, the scheme proposed for BC is modified according to a row transformation to the channel matrices. Such an operation allows us to identify the signal space where the transmitted signals interfere with each other and derive a proper power allocation policy to achieve a satisfactory DoF region.

We also derive an outer-bound for the DoF region of MIMO BC and IC using the aligned image set and the sliding window lemma. Using this outer-bound and the optimal DoF region when there is a mixture of the imperfect current CSIT and perfect delayed CSIT, we show that our proposed achievable DoF region is optimal under some antenna configurations and CSIT qualities. Remarkably, the maximal sum DoF is achievable in the case $\Phi_{BC}{\leq}0$ and $\Phi_{IC}{\leq}0$. This implies that Rx$1$ (i.e., the user with the smaller number of antennas) needs a greater CSIT quality than Rx$2$ (i.e., the user with the greater number of antennas). This fact contrasts with the symmetric case where the maximal sum DoF is achieved with equal CSIT qualities. On the other hand, if the Rx$1$ does not have a good enough CSIT quality, sending more streams of private messages to Rx$2$ (greater than the dimension of the null space) with the power higher than the CSIT quality is beneficial to the sum DoF performance. This contrasts with the symmetric case where unequal power allocation does not provide sum DoF gain.

Finally, it is noted that studying the DoF of MIMO networks with imperfect CSIT has attracted research attentions. While the paper is under review, another work was posted on arXiv on 3rd April, 2016 by Yuan and Jafar \cite{elevated_mux}. The authors investigated the same problem, but focused on two-receiver MIMO IC only and no outer-bound is provided. Compared to their scheme, so called elevated multiplexing, our RS approach has DoF gain in the case $N_1{\leq}N_2{\leq}\min\{M_1{,}M_2\}$ especially with the space-time transmission, while suffers from DoF loss in the case $N_1{<}M_1{\leq}N_2{<}M_2$. The advantage of our scheme lies in the unified framework, where the precoders and the number of private symbols and the power allocation policy are dynamically determined by the antenna configuration and CSIT qualities. Besides, by assuming the common message only carries information intended for Rx1 or Rx2, we obtain two DoF pairs, which is convenient to find a DoF region. One interesting work in the future would consist in studying how to harmonize both approaches to further tighten the achievability and outer bounds.

\section*{Appendices}
\subsection{Achievability DoF region of the related MAC}
We aim to show the achievable DoF tuples specified in \eqref{eq:dofnew}, \eqref{eq:dofcaseI} and \eqref{eq:dofcaseII} following the proof in \cite{xinping_mimo}. Without loss of generality, let us write the received signal at Rx$k$ as
\begin{IEEEeqnarray}{rcl}
\text{\rm BC:}\,\mathbf{y}_k&{=}&\mathbf{H}_k^H\mathbf{c}{+}\mathbf{H}_k^H\mathbf{x}_k{+}\boldsymbol\eta_{k{,}BC},\IEEEyesnumber\IEEEyessubnumber\\
\text{\rm IC:}\,\mathbf{y}_k&{=}&\mathbf{H}_{kk}^H\mathbf{c}_k{+}\mathbf{H}_{kj}^H\mathbf{c}_j{+}\mathbf{H}_{kk}^H\mathbf{x}_k{+}\boldsymbol\eta_{k{,}IC},
\IEEEyessubnumber\label{eq:ygeneral}
\end{IEEEeqnarray}
where $\mathbf{x}_k$ refers to the precoded private messages transmitted by Tx in BC and Tx$k$ in IC intended for Rx$k$, while $\boldsymbol\eta_{k{,}BC}{\triangleq}\mathbf{H}_k^H\mathbf{x}_j{+}\mathbf{n}_k{,}k{\neq}j$ and $\boldsymbol\eta_{k{,}IC}{\triangleq}\mathbf{H}_{kj}^H\mathbf{x}_j{+}\mathbf{n}_k{,}k{\neq}j$ represent the interference plus noise in BC and IC, respectively. In the following, let us only focus on \eqref{eq:ygeneral} as the derivation for the BC case follows similarly by simply taking $\mathbf{H}_{kk}{=}\mathbf{H}_{kj}$. For convenience, let us use $\boldsymbol\eta_k$ instead of $\boldsymbol\eta_{k{,}IC}$.

\begin{figure*}
\begin{IEEEeqnarray}{rcl} \label{eq:rate_conditions}
\!\!\!\!\!\!\!\!\!R_{ck}&{\leq}&I(\mathbf{c}_k;\mathbf{y}_k|\mathbf{c}_j{,}\mathbf{x}_k{,}\mathcal{H}_k){=}
h(\mathbf{y}_k|\mathbf{c}_j{,}\mathbf{x}_k{,}\mathcal{H}_k){-}h(\mathbf{y}_k|\mathbf{c}_k{,}\mathbf{c}_j{,}\mathbf{x}_k{,}\mathcal{H}_k),
\IEEEyesnumber\IEEEyessubnumber\label{eq:Rck}\\
\!\!\!\!\!\!\!\!\!R_{cj}&{\leq}&I(\mathbf{c}_j;\mathbf{y}_k|\mathbf{c}_k{,}\mathbf{x}_k{,}\mathcal{H}_k){=}
h(\mathbf{y}_k|\mathbf{c}_k{,}\mathbf{x}_k{,}\mathcal{H}_k){-}h(\mathbf{y}_k|\mathbf{c}_k{,}\mathbf{c}_j{,}\mathbf{x}_k{,}\mathcal{H}_k),
\IEEEyessubnumber\label{eq:Rcj}\\
\!\!\!\!\!\!\!\!\!R_{pk}&{\leq}&I(\mathbf{x}_k;\mathbf{y}_k|\mathbf{c}_k{,}\mathbf{c}_j{,}\mathcal{H}_k){=}
h(\mathbf{y}_k|\mathbf{c}_k{,}\mathbf{c}_j{,}\mathcal{H}_k){-}h(\mathbf{y}_k|\mathbf{c}_k{,}\mathbf{c}_j{,}\mathbf{x}_k{,}\mathcal{H}_k),
\IEEEyessubnumber\label{eq:Rpk}\\
\!\!\!\!\!\!\!\!\!R_{pk}{+}R_{ck}&{\leq}&I(\mathbf{x}_k{,}\mathbf{c}_k;\mathbf{y}_k|\mathbf{c}_j{,}\mathcal{H}_k){=}
h(\mathbf{y}_k|\mathbf{c}_j{,}\mathcal{H}_k){-}h(\mathbf{y}_k|\mathbf{c}_k{,}\mathbf{c}_j{,}\mathbf{x}_k{,}\mathcal{H}_k),
\IEEEyessubnumber\label{eq:Rckpk}\\
\!\!\!\!\!\!\!\!\!R_{pk}{+}R_{cj}&{\leq}&I(\mathbf{x}_k{,}\mathbf{c}_j;\mathbf{y}_k|\mathbf{c}_k{,}\mathcal{H}_k)\nonumber{=}
h(\mathbf{y}_k|\mathbf{c}_k{,}\mathcal{H}_k){-}h(\mathbf{y}_k|\mathbf{c}_k{,}\mathbf{c}_j{,}\mathbf{x}_k{,}\mathcal{H}_k),
\IEEEyessubnumber\label{eq:Rcjpk}\\
\!\!\!\!\!\!\!\!\!R_{ck}{+}R_{cj}{+}R_{pk}&{\leq}&I(\mathbf{c}_k{,}\mathbf{c}_j{,}\mathbf{x}_k;\mathbf{y}_k|\mathcal{H}_k)\nonumber{=}
h(\mathbf{y}_k|\mathcal{H}_k){-}h(\mathbf{y}_k|\mathbf{c}_k{,}\mathbf{c}_j{,}\mathbf{x}_k{,}\mathcal{H}_k),
\IEEEyessubnumber\label{eq:Rckcjpk}
\end{IEEEeqnarray}
\hrulefill
\end{figure*}
As pointed out in \cite{xinping_mimo}, the MIMO system in \eqref{eq:ygeneral} is a MAC as Rx$k$ aims to decode $\mathbf{c}_k$, $\mathbf{c}_j$ and $\mathbf{x}_k$. Then, according to \cite{network_info}, a rate tuple $(R_{c1}{,}R_{c2}{,}R_{pk})$ is achievable if \eqref{eq:rate_conditions} hold for any input distribution $p_{\mathbf{x}_k{,}\mathbf{c}_k{,}\mathbf{c}_j}{=}p_{\mathbf{x}_k}p_{\mathbf{c}_k}p_{\mathbf{c}_j}$, and $\mathcal{H}_k{\triangleq}\{\mathbf{H}_{kk}{,}\mathbf{H}_{kj}\}$ is the set of the channel state. By setting $R_{pk}$ equal to the r.h.s. of \eqref{eq:Rpk} and plugging it into \eqref{eq:Rckpk}, \eqref{eq:Rcjpk} and \eqref{eq:Rckcjpk}, we have
\begin{IEEEeqnarray}{rcl}
R_{ck}&{\leq}&h(\mathbf{y}_k|\mathbf{c}_j{,}\mathcal{H}_k){-}h(\mathbf{y}_k|\mathbf{c}_k{,}\mathbf{c}_j{,}\mathcal{H}_k),
\IEEEyesnumber\IEEEyessubnumber\label{eq:Rckactive}\\
R_{cj}&{\leq}&h(\mathbf{y}_k|\mathbf{c}_k{,}\mathcal{H}_k){-}h(\mathbf{y}_k|\mathbf{c}_k{,}\mathbf{c}_j{,}\mathcal{H}_k),
\IEEEyessubnumber\label{eq:Rcjactive}\\
R_{ck}{+}R_{cj}&{\leq}&h(\mathbf{y}_k|\mathcal{H}_k){-}h(\mathbf{y}_k|\mathbf{c}_k{,}\mathbf{c}_j{,}\mathcal{H}_k).
\IEEEyessubnumber\label{eq:Rckcjactive}
\end{IEEEeqnarray}
Note that the r.h.s. of \eqref{eq:Rckactive} can be interpreted as $I(\mathbf{c}_k;\mathbf{y}_k|\mathbf{c}_j{,}\mathcal{H}_k)$, which is equal to $I(\mathbf{c}_k;\mathbf{y}_k^\prime|\mathcal{H}_k)$ with $\mathbf{y}_k^\prime{\triangleq}\mathbf{H}_{kk}^H\mathbf{c}_k{+}\mathbf{H}_{kk}^H\mathbf{x}_k{+}\boldsymbol\eta_k$. Similarly, the r.h.s. of \eqref{eq:Rck} can be expressed as $I(\mathbf{c}_k;\mathbf{y}_k^{\prime\prime}|\mathcal{H}_k)$ with $\mathbf{y}_k^{\prime\prime}{\triangleq}\mathbf{H}_{kk}^H\mathbf{c}_k{+}\boldsymbol\eta_k$. Since $\mathbf{c}_k{\to}\mathbf{y}_k^{\prime\prime}{\to}\mathbf{y}_k^{\prime}$ forms a Markov Chain, we have $I(\mathbf{c}_k;\mathbf{y}_k^\prime|\mathcal{H}_k){\leq}I(\mathbf{c}_k;\mathbf{y}_k^{\prime\prime}|\mathcal{H}_k)$ due to the data processing inequality \cite{Inf_Theo}. Therefore, the inequalities in \eqref{eq:Rck} and \eqref{eq:Rcj} are inactive and the achievable rate of the common messages is specified by \eqref{eq:Rckactive}, \eqref{eq:Rcjactive} and \eqref{eq:Rckcjactive}.

Let the input be $\mathbf{c}_k\stackrel{d}{\sim}\mathcal{CN}(0{,}P\mathbf{I}_{M_k})$, $\mathbf{c}_j\stackrel{d}{\sim}\mathcal{CN}(0{,}P\mathbf{I}_{M_j})$ and $\mathbf{x}_k\stackrel{d}{\sim}\mathcal{CN}(0{,}\mathbf{B}_k)$. The achievable rate constraints \eqref{eq:Rckactive}, \eqref{eq:Rcjactive}, \eqref{eq:Rckcjactive} and \eqref{eq:Rpk} can be further expressed by
\begin{IEEEeqnarray}{rcl}\label{eq:Rckcjpkcompute}
R_{ck}&{\leq}&{\log}_2\det(\mathbf{Q}_{ck}{+}\mathbf{Q}_k{+}\mathbf{Q}_{\boldsymbol\eta_k}){-}\nonumber\\
&&{\log}_2\det(\mathbf{Q}_{k}{+}\mathbf{Q}_{\boldsymbol\eta_k}),\IEEEyesnumber\IEEEyessubnumber\label{eq:Rckcompute}\\
R_{cj}&{\leq}&{\log}_2\det(\mathbf{Q}_{cj}{+}\mathbf{Q}_k{+}\mathbf{Q}_{\boldsymbol\eta_k}){-}\nonumber\\
&&{\log}_2\det(\mathbf{Q}_{k}{+}\mathbf{Q}_{\boldsymbol\eta_k}),\IEEEyessubnumber\label{eq:Rcjcompute}\\
R_{ck}{+}R_{cj}&{\leq}&{\log}_2\det(\mathbf{Q}_{ck}{+}\mathbf{Q}_{cj}{+}\mathbf{Q}_k{+}\mathbf{Q}_{\boldsymbol\eta_k}){-}\nonumber\\
&&{\log}_2\det(\mathbf{Q}_{k}{+}\mathbf{Q}_{\boldsymbol\eta_k}),\IEEEyessubnumber\label{eq:Rckcjcompute}\\
R_{pk}&{=}&{\log}_2\det(\mathbf{Q}_k{+}\mathbf{Q}_{\boldsymbol\eta_k}){-}
{\log}_2\det(\mathbf{Q}_{\boldsymbol\eta_k}),\IEEEyessubnumber\label{eq:Rpkcompute}
\end{IEEEeqnarray}
where $\mathbf{Q}_{ck}{=}P\mathbf{H}_{kk}^H\mathbf{H}_{kk}$, $\mathbf{Q}_{cj}{=}P\mathbf{H}_{kj}^H\mathbf{H}_{kj}$, $\mathbf{Q}_k{=}\mathbf{H}_{kk}^H\mathbf{B}_k\mathbf{H}_{kk}$ and $\mathbf{Q}_{\boldsymbol\eta_k}{=}\mathbf{H}_{kj}^H\mathbf{B}_j\mathbf{H}_{kj}{+}\mathbf{I}_{N_k}$ denote the covariance matrices of $\mathbf{H}_{kk}^H\mathbf{c}_k$, $\mathbf{H}_{kj}^H\mathbf{c}_j$, $\mathbf{H}_{kk}^H\mathbf{x}_k$ and $\boldsymbol\eta_k$ in \eqref{eq:ygeneral}, respectively.

Next, let us identify the related covariance matrix in the MIMO IC when $M_1{\geq}N_2$ and MIMO IC when $M_1{\leq}N_2$. The derivation of the covariance in MIMO BC follows similarly to the MIMO IC when $M_1{\geq}N_2$ by setting $M_1{=}M_2$ and $N_2^\prime{=}N_2$.

\subsubsection{Case I: $M_1{\geq}N_2$ and $M_2{\geq}N_1$}
In this case, we have $\mathbf{B}_1{=}P^{A_1}\mathbf{S}_{\hat{\mathbf{H}}_{21}^\bot}$ and $\mathbf{B}_2{=}P^{A_2}\mathbf{S}_{\hat{\mathbf{H}}_{12}^\bot}{+}P^{(A_2{-}\alpha_1)^+}\mathbf{S}_{\hat{\mathbf{H}}_{22}}$, where $\mathbf{S}_{\hat{\mathbf{H}}_{21}^\bot}{\triangleq}\mathbf{V}_1\mathbf{V}_1^H$, $\mathbf{S}_{\hat{\mathbf{H}}_{12}^\bot}{\triangleq}\mathbf{V}_2^{(1)}\mathbf{V}_2^{(1)H}$ and $\mathbf{S}_{\hat{\mathbf{H}}_{22}}{\triangleq}\mathbf{V}_2^{(2)}\mathbf{V}_2^{(2)H}$.

At Rx1, as covariance matrices $\mathbf{Q}_{c1}$ and $\mathbf{Q}_{c2}$ are rank $N_1$ (since $M_2{\geq}N_1$ and $M_1{\geq}N_2{\geq}N_1$), it readily shows that ${\log}_2\det(\mathbf{Q}_{c1}{+}\mathbf{Q}_1{+}\mathbf{Q}_{\boldsymbol\eta_1})$, ${\log}_2\det(\mathbf{Q}_{c2}{+}\mathbf{Q}_1{+}\mathbf{Q}_{\boldsymbol\eta_1})$ and ${\log}_2\det(\mathbf{Q}_{c1}{+}\mathbf{Q}_{c2}{+}\mathbf{Q}_1{+}\mathbf{Q}_{\boldsymbol\eta_1})$ are equal to $N_1{\log}_2P{+}o({\log}_2P)$ as $\mathbf{Q}_{c1}$ and $\mathbf{Q}_{c2}$ are dominating compared to $\mathbf{Q}_1$ and $\mathbf{Q}_{\boldsymbol\eta_1}$. Moreover, let us write the eigenvalue decomposition of $\mathbf{Q}_1$ and $\mathbf{Q}_{\boldsymbol\eta_1}$ as $\mathbf{U}_1\mathbf{D}_1\mathbf{U}_1^H$ and $\mathbf{U}_{\boldsymbol\eta_1}\mathbf{D}_{\boldsymbol\eta_1}\mathbf{U}_{\boldsymbol\eta_1}^H$, respectively, where $\mathbf{D}_1{\sim}{\rm diag}(P^{A_1}\mathbf{I}_{M_1{-}N_2}{,}\mathbf{0}_{N_1{+}N_2{-}M_1})$ and $\mathbf{D}_{\boldsymbol\eta_1}{\sim}P^{(A_2{-}\alpha_1)^+}\mathbf{I}_{N_1}$. Then, it follows that
\begin{IEEEeqnarray}{rcl}
{\log}_2\det(\mathbf{Q}_1{+}\mathbf{Q}_{\boldsymbol\eta_1})&{=}&(M_1{-}N_2)\max\{A_2{-}\alpha_1{,}A_1\}{\log}_2P{+}\nonumber\\
&&(N_1{+}N_2{-}M_1)(A_2{-}\alpha_1^+){\log}_2P{+}\nonumber\\
&&o({\log}_2P),\nonumber\\
{\log}_2\det(\mathbf{Q}_{\boldsymbol\eta_1})&{=}&N_1(A_2{-}\alpha_1^+){\log}_2P{+}o({\log}_2P).\nonumber
\end{IEEEeqnarray}
Plugging the corresponding values into \eqref{eq:Rckcjpkcompute} leads to \eqref{eq:dc1y1caseI}, \eqref{eq:dc2y1caseI}, \eqref{eq:dp1caseI} and \eqref{eq:dcsy1caseI}.

At Rx2, let us write the eigenvalue decomposition of the $\mathbf{Q}_2$ and $\mathbf{Q}_{\boldsymbol\eta_2}$ as $\mathbf{U}_2\mathbf{D}_2\mathbf{U}_2^H$ and $\mathbf{U}_{\boldsymbol\eta_2}\mathbf{D}_{\boldsymbol\eta_2}\mathbf{U}_{\boldsymbol\eta_2}^H$, respectively, where $\mathbf{D}_2{\sim}{\rm diag}(P^{A_2}\mathbf{I}_{M_2{-}N_1}{,}P^{(A_2{-}\alpha_1)^+}\mathbf{I}_{N_1{+}N_2{-}M_1})$ and $\mathbf{D}_{\boldsymbol\eta_2}{\sim}P^{0}\mathbf{I}_{N_1}$ since $A_1{\leq}\alpha_2$. Besides, the covariance matrices $\mathbf{Q}_{c1}$ is rank $N_2$ and $\mathbf{Q}_{c2}$ is rank $N_2^\prime$. Then it can be readily shown that ${\log}_2\det(\mathbf{Q}_{c2}{+}\mathbf{Q}_2{+}\mathbf{Q}_{\boldsymbol\eta_2}){=}N_2^\prime{\log}_2P{+}o({\log}_2P)$, while ${\log}_2\det(\mathbf{Q}_{c1}{+}\mathbf{Q}_2{+}\mathbf{Q}_{\boldsymbol\eta_2}){=}{\log}_2\det(\mathbf{Q}_{c1}{+} \mathbf{Q}_{c2}{+}\\ \mathbf{Q}_2{+}\mathbf{Q}_{\boldsymbol\eta_2}){=}N_2{\log}_2P{+}o({\log}_2P)$, and
\begin{IEEEeqnarray}{rcl}
{\log}_2\det(\mathbf{Q}_2{+}\mathbf{Q}_{\boldsymbol\eta_2})&{=}&
(N_1{+}N_2{-}M_2)(A_2{-}\alpha_1^+){\log}_2P{+}\nonumber\\
&&(M_2{-}N_1)A_2{\log}_2P{+}o({\log}_2P),\nonumber\\
{\log}_2\det(\mathbf{Q}_{\boldsymbol\eta_1})&{=}&o({\log}_2P)\nonumber
\end{IEEEeqnarray}
Plugging the corresponding values into \eqref{eq:Rckcjpkcompute} leads to \eqref{eq:dc1y2caseI}, \eqref{eq:dc2y2caseI}, \eqref{eq:dp2caseI} and \eqref{eq:dcsy2caseI}.

\subsubsection{Case II: $M_1{\leq}N_2$ and $M_2{\geq}N_1$}
In this case, $\mathbf{B}_1{=}0$ as there is no private messages intended for Rx1, while
\begin{IEEEeqnarray}{rcl}
\mathbf{B}_2&{=}&P\mathbf{S}_{\hat{\mathbf{H}}_{22}}^{(1)}{+}P^{A_2^\prime}\mathbf{S}_{\hat{\mathbf{H}}_{12}^\bot}^{(2)}{+} P^{A_2}\mathbf{S}_{\hat{\mathbf{H}}_{12}^\bot}^{(3)}{+}\nonumber\\
&&P^{A_2^\prime{-}\alpha_1}\mathbf{S}_{\hat{\mathbf{H}}_{22}}^{(4)}{+} P^{(A_2{-}\alpha_1)^+}\mathbf{S}_{\hat{\mathbf{H}}_{22}}^{(5)}\nonumber
\end{IEEEeqnarray}
where $\mathbf{S}_{\hat{\mathbf{H}}_{22}}^{(m)}{\triangleq}\mathbf{V}_2^{(m)}\mathbf{V}_2^{(m)H}{,}m{=}1{,}4{,}5$ and $\mathbf{S}_{\hat{\mathbf{H}}_{12}^\bot}^{(m)}{\triangleq}\mathbf{V}_2^{(m)}\mathbf{V}_2^{(m)H}{,}m{=}2{,}3$. Let us focus on the received signal $\tilde{\mathbf{y}}_k{=}\mathbf{T}_k\mathbf{y}_k{,}k{=}1{,}2$, as linear transformation does not change the mutual information.

At Rx1, covariance matrices $\mathbf{Q}_{c1}$ and $\mathbf{Q}_{c2}$ rewrites as $P\mathbf{T}_1\mathbf{H}_{11}^H\mathbf{H}_{11}\mathbf{T}_1^H$ and $P\bar{\mathbf{H}}_{12}^H\bar{\mathbf{H}}_{12}$, where $\mathbf{T}_1\mathbf{H}_{11}^H$ and $\bar{\mathbf{H}}_{12}^H$ are given by \eqref{eq:T1}. Besides, the eigenvalue decomposition of the covariance matrix $\mathbf{Q}_{\boldsymbol\eta_1}{=}\bar{\mathbf{H}}_{12}^H\mathbf{B}_2\bar{\mathbf{H}}_{12}$ can be expressed as $\mathbf{U}_{\boldsymbol\eta_1}\mathbf{D}_{\boldsymbol\eta_1}\mathbf{U}_{\boldsymbol\eta_1}^H$, where
\begin{IEEEeqnarray}{rcl}
\mathbf{D}_{\boldsymbol\eta_1}&{\sim}&{\rm diag}(P\mathbf{I}_{\tau}{,}P^{A_2^\prime{-}\alpha_1}\mathbf{I}_{\min\{N_1^\prime{,}\mu_1{+}\delta_1\}}{,}\nonumber\\
&&P^{(A_2{-}\alpha_1)^+}\mathbf{I}_{N_1^\prime{-}\min\{N_1^\prime{,}\mu_1{+}\delta_1\}}).
\end{IEEEeqnarray}
This is due to the following reasons: 1) according to $\tilde{\mathbf{y}}_2$, $\tau$ private messages are received with power $P$, $\mu_1{+}\delta_1$ private messages are received with the power level $A_2^\prime{-}\alpha_1$ and $\mu_2{+}\delta_2$ private messages are received with the power level $(A_2{-}\alpha_1)^+$ because of ZFBF with imperfect CSIT, and 2) as $A_2{\leq}A_2^\prime$, the $\mu_2{+}\delta_2$ private messages with power level $(A_2{-}\alpha_1)^+$ are drowned by the other $\tau{+}\mu_1{+}\delta_1$ private messages. Note that if $\mu_1{+}\delta_1{\geq}N_1^\prime$, i.e., $\mu_1{+}\delta_1{+}\tau{\geq}N_1$, the $\mu_2{+}\delta_2$ private messages with power level $(A_2{-}\alpha_1)^+$ do not impact $\mathbf{D}_{\boldsymbol\eta_1}$; otherwise, there are $N_1^\prime{-}\mu_1{-}\delta_1$ eigenvalues with power level $(A_2{-}\alpha_1)^+$. In this way, it can be readily shown that ${\log}_2\det(\mathbf{Q}_{c1}{+}\mathbf{Q}_{c2}{+}\mathbf{Q}_1{+}\mathbf{Q}_{\boldsymbol\eta_1})$, ${\log}_2\det(\mathbf{Q}_{c2}{+}\mathbf{Q}_1{+}\mathbf{Q}_{\boldsymbol\eta_1})$ and ${\log}_2\det(\mathbf{Q}_{c1}{+}\mathbf{Q}_1{+}\mathbf{Q}_{\boldsymbol\eta_1})$ are equal to $N_1{\log}_2P{+}o({\log}_2P)$, while
\begin{IEEEeqnarray}{rcl}
{\log}_2\det(\mathbf{Q}_{\boldsymbol\eta_1})&{=}&{\log}_2\det(\mathbf{Q}_1{+}\mathbf{Q}_{\boldsymbol\eta_1})\nonumber\\
&{=}&\tau{\log}_2P{+}\min\{N_1^\prime{,}\mu_1{+}\delta_1\}(A_2^\prime{-}\alpha_1){\log}_2P{+}\nonumber\\&&
(N_1^\prime{-}\min\{N_1^\prime{,}\mu_1{+}\delta_1\})(A_2{-}\alpha_1)^+{\log}_2P{+}\nonumber\\
&&o({\log}_2P).\nonumber
\end{IEEEeqnarray}
Plugging the corresponding values into \eqref{eq:Rckcjpkcompute} leads to \eqref{eq:dc1y1caseII}, \eqref{eq:dc2y1caseII} and \eqref{eq:dcsy1caseII}.

At Rx2, after the linear transformation, it can be shown that covariance matrices $\mathbf{Q}_{c1}$ and $\mathbf{Q}_{c2}$ rewrites as $P\mathbf{T}_2\mathbf{H}_{21}^H\mathbf{H}_{21}\mathbf{T}_2^H$ and $P\mathbf{T}_2\mathbf{H}_{22}^H\mathbf{H}_{22}\mathbf{T}_2^H$, where $\mathbf{T}_2\mathbf{H}_{21}^H$ and $\mathbf{T}_2\mathbf{H}_{22}^H$ are given by \eqref{eq:T2}. Besides, the covariance matrix $\mathbf{Q}_2{=}\mathbf{T}_2\mathbf{H}_{21}^H\mathbf{B}_2\mathbf{H}_{22}\mathbf{T}_2^H$ can be expressed as $\mathbf{U}_2\mathbf{D}_2\mathbf{U}_2^H$, where
\begin{IEEEeqnarray}{rcl}
\mathbf{D}_2&{\sim}&{\rm diag}(\mathbf{0}_{N_2{-}N_2^\prime}{,}P^{(A_2{-}\alpha_1)^+}\mathbf{I}_{\delta_2}{,}P^{A_2}\mathbf{I}_{\mu_2}{,}\nonumber\\ &&P^{A_2^\prime{-}\alpha_1}\mathbf{I}_{\delta_1}{,}P^{A_2^\prime}\mathbf{I}_{\mu_1}{,}P\mathbf{I}_{\tau})\nonumber
\end{IEEEeqnarray}
As there is no private messages sent to Rx1, $\boldsymbol\eta_2$ only consists of noise so that $\mathbf{Q}_{\boldsymbol\eta_2}{=}P^0\mathbf{I}_{N_2}$. Accordingly, it is clear that ${\log}_2\det(\mathbf{Q}_{c2}{+}\mathbf{Q}_2{+}\mathbf{Q}_{\boldsymbol\eta_2}){=}N_2^\prime{\log}_2P{+}o({\log}_2P)$ and ${\log}_2\det(\mathbf{Q}_{c1}{+}\mathbf{Q}_{c2}{+}\mathbf{Q}_2{+}\mathbf{Q}_{\boldsymbol\eta_2}){=}N_2{\log}_2P{+}\\ o({\log}_2P)$. Moreover, we can see that the last $\tau{+}\mu_1{+}\delta_1{=}N_2{-}M_1$ columns in $\mathbf{U}_2$ do not overlap with the column space of $\bar{\mathbf{H}}_{21}^H$, then, it can be readily shown that
\begin{IEEEeqnarray}{rcl}
{\log}_2\det(\mathbf{Q}_{c1}{+}\mathbf{Q}_2{+}\mathbf{Q}_{\boldsymbol\eta_2})&{=}&
M_1{\log}_2P{+}\tau{\log}_2P{+}\mu_1A_2^\prime{\log}_2P{+}\nonumber\\
&&\delta_1(A_2^\prime{-}\alpha_1){\log}_2P{+}o({\log}_2P),\nonumber\\
{\log}_2\det(\mathbf{Q}_2{+}\mathbf{Q}_{\boldsymbol\eta_2})&{=}&
\tau{\log}_2P{+}\mu_1A_2^\prime{\log}_2P{+}\nonumber\\
&&\delta_1(A_2^\prime{-}\alpha_1){\log}_2P{+}\mu_2A_2{\log}_2P{+}\nonumber\\
&&\delta_2(A_2{-}\alpha_1)^+{\log}_2P{+}o({\log}_2P).\nonumber
\end{IEEEeqnarray}
Plugging the corresponding values into \eqref{eq:Rckcjpkcompute} leads to \eqref{eq:dc1y2caseII}, \eqref{eq:dc2y2caseII}, \eqref{eq:dp2caseII} and \eqref{eq:dcsy1caseII}.

\subsection{Solving the optimization problem in \eqref{eq:optd2_2}}

We firstly transform the problem into two sub-problems by considering $A_2{\leq}\alpha_1$ and $A_2{\geq}\alpha_1$, whose closed-form solutions are convenient to calculate. Then, we obtain the closed-form solution to \eqref{eq:optd2_2} by comparing these two closed-form solutions.

\subsubsection{$A_2{\leq}\alpha_1$} In this case, the optimization problem in \eqref{eq:optd2_2} rewrites as
\begin{IEEEeqnarray}{rcl}\label{eq:optd2_2_1}
\max_{A_2{,}A_2^\prime}&\quad& d_{2{,}(2)}(A_2{,}A_2^\prime{,}\lambda)\IEEEyesnumber\IEEEyessubnumber\\
\text{s.t.}&\quad&
\lambda{\leq}N_1^\prime{-}\xi(A_2^\prime{-}\alpha_1),\IEEEyessubnumber\label{eq:lambdacons1_2_1}\\
&&\lambda{\leq}M_1{-}\mu_2A_2,\IEEEyessubnumber\label{eq:lambdacons2_2_1}\\
&&0{\leq}A_2{\leq}\alpha_1,\IEEEyessubnumber\\
&&\alpha_1{\leq}A_2^\prime{\leq}1.\IEEEyessubnumber
\end{IEEEeqnarray}
As $d_{2{,}(2)}(A_2{,}A_2^\prime{,}\lambda)$ given in \eqref{eq:d2caseII2} is increasing with $A_2$ and $A_2^\prime$, it is straightforward to obtain the optimal solution to \eqref{eq:optd2_2_1}, which writes as
\begin{IEEEeqnarray}{rcl}
A_2{=}\min\left\{\alpha_1{,}\frac{M_1{-}\lambda}{\mu_2}\right\}&{,}\,&
A_2^\prime{=}\min\left\{1{,}\alpha_1{+}\frac{N_1^\prime{-}\lambda}{\xi}\right\}.\label{eq:solution1}
\end{IEEEeqnarray}

\subsubsection{$A_2{\geq}\alpha_1$} In this case, the optimization problem in \eqref{eq:optd2_2} rewrites as
\begin{IEEEeqnarray}{rcl}\label{eq:optd2_2_2}
\max_{A_2{,}A_2^\prime}&\quad& d_{2{,}(2)}(A_2{,}A_2^\prime{,}\lambda)\IEEEyesnumber\IEEEyessubnumber\\
\text{s.t.}&\quad&
\lambda{\leq}N_1^\prime{-}\xi A_2^\prime{+}N_1^\prime\alpha_1){-}(N_1^\prime{-}\xi)A_2,\IEEEyessubnumber\label{eq:lambdacons1_2_2}\\
&&\lambda{\leq}M_1{-}M_1A_2{+}\delta_2\alpha_1,\IEEEyessubnumber\label{eq:lambdacons2_2_2}\\
&&\alpha_1{\leq}A_2{\leq}A_2^\prime,\IEEEyessubnumber\label{eq:equal}\\
&&\alpha_1{\leq}A_2^\prime{\leq}1,\IEEEyessubnumber\label{eq:A2pleq1}
\end{IEEEeqnarray}
where we have used the fact that $\mu_2{+}\delta{=}M_1$ given the condition $M_2{\geq}N_2$. As $d_{2{,}(2)}(A_2{,}A_2^\prime{,}\lambda)$ given in \eqref{eq:d2caseII2} is increasing with $A_2$ and $A_2^\prime$, we learn that the optimal solution is obtained when (at least) two of the constraints \eqref{eq:lambdacons1_2_2}, $\eqref{eq:lambdacons2_2_2}$, $A_2{\leq}A_2^\prime$ and $A_2^\prime{\leq}1$ are active. Notably, from \eqref{eq:lambdacons2_2_2} and $A_2{\geq}\alpha_1$, we see that $\lambda$ should be smaller than or equal to $M_1{-}\mu_2\alpha_1$, otherwise, there is no solution to \eqref{eq:optd2_2_2}. Therefore, the discussion goes into following four cases.

\underline{If $A_2{=}A_2^\prime{=}1$:} In this case, $\lambda$ is such that $\lambda{\leq}\min\{\delta_2\alpha_1{,}N_1^\prime\alpha_1\}{=}\delta_2\alpha_1$ according to \eqref{eq:lambdacons1_2_2} and \eqref{eq:lambdacons2_2_2}.

\underline{If $A_2{<}A_2^\prime{=}1$:} In this case, plugging $A_2^\prime{=}1$ into \eqref{eq:lambdacons1_2_2} and \eqref{eq:lambdacons2_2_2} yields
\begin{IEEEeqnarray}{rcl}
A_2&{=}&\min\left\{1{-}\frac{\lambda{-}\delta_2\alpha_1}{M_1}{,}1{-}\frac{\lambda{-}N_1^\prime\alpha_1}{N_1^\prime{-}\xi}\right\}.\label{eq:case2}
\end{IEEEeqnarray}

\underline{If $A_2{=}A_2^\prime{<}1$:} In this case, plugging $A_2{=}A_2^\prime$ into \eqref{eq:lambdacons1_2_2} and \eqref{eq:lambdacons2_2_2} yields
\begin{IEEEeqnarray}{rcl}
A_2^\prime{=}A^\prime&{=}&\min\left\{1{-}\frac{\lambda{-}\delta_2\alpha_1}{M_1}
{,}1{-}\frac{\lambda{-}N_1^\prime\alpha_1}{N_1^\prime}\right\}.\label{eq:case3}
\end{IEEEeqnarray}

\underline{If $A_2{<}A_2^\prime{<}1$:} In this case, using \eqref{eq:lambdacons1_2_2} and \eqref{eq:lambdacons2_2_2}, we have
\begin{IEEEeqnarray}{rcl}
A_2&{=}&1{-}\frac{\lambda{-}\delta_2\alpha_1}{M_1},\IEEEyesnumber\IEEEyessubnumber\label{eq:case41}\\
A_2^\prime&{=}&1{-}\frac{\left(M_1{-}N_1^\prime{+}\xi\right)}{M_1\xi}\lambda{+}
\frac{\left(\mu_2N_1^\prime{+}\delta_2\xi\right)}{M_1\xi}\alpha_1,\IEEEyessubnumber\label{eq:case42}
\end{IEEEeqnarray}
while $\lambda$ is such that $\frac{\mu_2N_1^\prime{+}\delta_2\xi}{M_1{-}N_1^\prime{+}\xi}\alpha_1{\leq}\lambda{\leq}\frac{\mu_2N_1^\prime}{M_1{-}N_1^\prime}\alpha_1$ according constraints \eqref{eq:equal} and \eqref{eq:A2pleq1}, .

From these four cases, we conclude that the closed-form solution to the optimization problem \eqref{eq:optd2_2_2} is given by \eqref{eq:solution2} at the top of the next page.
\begin{figure*}
\begin{IEEEeqnarray}{rll} \label{eq:solution2}
&\text{\rm For }\lambda{\in}\left[0{,}\delta_2\alpha_1\right],& A_2{=}A_2^\prime{=}1,\IEEEyesnumber\IEEEyessubnumber\\
&\text{\rm For }\lambda{\in}\left[\delta_2\alpha_1{,} \min\left\{\frac{\mu_2N_1^\prime{+}\delta_2\xi}{M_1{-}N_1^\prime{+}\xi}\alpha_1{,}M_1{-}\mu_2\alpha_1\right\}\right],
& A_2{=}\text{\rm Eq.}\eqref{eq:case2}{=}1{-}\frac{\lambda{-}\delta_2\alpha_1}{M_1},A_2^\prime{=}1,\IEEEyessubnumber\\
&\text{\rm For }\lambda{\in}\left[\frac{\mu_2N_1^\prime{+}\delta_2\xi}{M_1{-}N_1^\prime{+}\xi}\alpha_1{,} \min\left\{M_1{-}\mu_2\alpha_1{,}\frac{N_1^\prime\mu_2\alpha_1}{M_1{-}N_1^\prime}\right\}\right],& A_2{=}\text{\rm Eq.}\eqref{eq:case41},A_2^\prime{=}\text{\rm Eq.}\eqref{eq:case42}\IEEEyessubnumber\\
&\text{\rm For }\lambda{\in}\left[\frac{N_1^\prime\mu_2\alpha_1}{M_1{-}N_1^\prime}{,}M_1{-}\mu_2\alpha_1\right],& A_2{=}A_2^\prime{=}\text{\rm Eq.}\eqref{eq:case3}{=}1{-}\frac{\lambda{-}N_1^\prime\alpha_1}{N_1^\prime}.\IEEEyessubnumber
\end{IEEEeqnarray}
\hrulefill
\end{figure*}

\subsubsection{Obtain the solution to \eqref{eq:optd2_2}}

The remaining task is to compare the solution to \eqref{eq:optd2_2_1} and \eqref{eq:optd2_2_2} in order to obtain the solution to \eqref{eq:optd2_2}. As mentioned in the above derivation, the closed form solution to \eqref{eq:optd2_2_2}, namely \eqref{eq:solution2}, is valid when $\lambda{\leq}M_1{-}\mu_2\alpha_1$. In this case, \eqref{eq:solution1} writes as $A_2{=}\alpha_1$ and $A_2^\prime{=}\min\{1{,}\alpha_1{+}\frac{N_1^\prime{-}\lambda}{\xi}\}$. By plugging \eqref{eq:solution1} and \eqref{eq:solution2} into $d_{2{,}(2)}(A_2{,}A_2^\prime{,}\lambda)$, it can be shown that \eqref{eq:solution2} leads to a greater value of $d_{2{,}(2)}(A_2{,}A_2^\prime{,}\lambda)$. Therefore, when $\lambda{\leq}M_2{-}\mu_2\alpha_1$, the closed-form solution to \eqref{eq:optd2_2} is given by \eqref{eq:solution2}, which leads to Condition C, D, E and F shown in Section \ref{sec:II2}. When $\lambda{\geq}M_2{-}\mu_2\alpha_1$, closed-form solution to \eqref{eq:optd2_2} is given by \eqref{eq:solution1}, namely $A_2{=}\frac{M_1{-}\lambda}{\mu_2}$ and $A_2^\prime{=}\min\{1{,}\alpha_1{+}\frac{N_1^\prime{-}\lambda}{\xi}\}$, which leads to Condition A and B shown in Section \ref{sec:II2}.

\subsection{Proof of Proposition \ref{prop:BC_outer}}
In this section, we present the proof focusing on real domain. The extension to the complex signal follows similarly as in \cite{Davoodi14}. In the following, the proof is carried out in three antenna configurations, i.e., $M{=}N_1{+}N_2$, $N_2{\leq}M{<}N_1{+}N_2$ and $N_1{\leq}M{<}N_2$. The proof for $M{>}N_1{+}N_2$ is similar to $M{=}N_1{+}N_2$. The proof for $M{<}N_1$ is omitted as the DoF region is consistent with the case of no CSIT, and is also consistent with the case of perfect CSIT.
\subsubsection{$M{=}N_1{+}N_2$}

The derivation follows the footsteps in \cite{Davoodi14}. There are three main steps. The first step is to obtain a canonical form of the MIMO system, the second step is to define the functional dependence and the aligned image set, while the last step is to bound the probability that two realizations of one user's observation provide the same image in the other user's observation.

\textbf{Step 1:} Let us write the received signals of the MIMO BC as
\begin{equation}
  \left[\begin{array}{c}
          \mathbf{y}_1 \\
          \mathbf{y}_2
        \end{array}\right]{=}\left[\begin{array}{cc}
                                     \mathbf{H}_{11}^H & \mathbf{H}_{12}^H \\
                                     \mathbf{H}_{21}^H & \mathbf{H}_{22}^H
                                   \end{array}\right]\left[\begin{array}{c}
                                                             \mathbf{s}_1 \\
                                                             \mathbf{s}_2
                                                           \end{array}\right]{+}\left[\begin{array}{c}
                                                             \mathbf{n}_1 \\
                                                             \mathbf{n}_2
                                                           \end{array}\right],
\end{equation}
where $\mathbf{s}_1$ is the transmitted signal of the first $N_1$ antennas, while $\mathbf{s}_2$ is the transmitted signal of the last $N_2$ antennas. Note that in this section, $\mathbf{H}_{k1}$ denotes the channel matrices from the first $N_1$ transmit antennas to user $k$, while $\mathbf{H}_{k2}$ denotes the channel matrices from the last $N_2$ transmit antennas to user $k$. Assuming there is perfect CSIT for user $1$, the canonical form writes as
\begin{equation}
  \left[\begin{array}{c}
          \tilde{\mathbf{y}}_1 \\
          \tilde{\mathbf{y}}_2
        \end{array}\right]{=}\left[\begin{array}{cc}
                                     \mathbf{I}_{N_1} & \mathbf{0}_{N_1{\times}N_2} \\
                                     \mathbf{G}_2 & \mathbf{I}_{N_2}
                                   \end{array}\right]\left[\begin{array}{c}
                                                             \tilde{\mathbf{s}}_1 \\
                                                             \tilde{\mathbf{s}}_2
                                                           \end{array}\right]{+}\left[\begin{array}{c}
                                                             \mathbf{n}_1 \\
                                                             \mathbf{n}_2
                                                           \end{array}\right]\label{eq:canonical_user1},
\end{equation}
where $\mathbf{G}_2{=}\mathbf{H}_{21}^H\mathbf{H}_{11}^{-H}$, $\tilde{\mathbf{s}}_1{=}\mathbf{H}_{11}^H\mathbf{s}_1{+}\mathbf{H}_{12}^H\mathbf{s}_2$ and $\tilde{\mathbf{s}}_2{=}\left(\mathbf{H}_{22}^H{-}\mathbf{G}_2\mathbf{H}_{12}^H\right)\mathbf{s}_2$.

Then, denoting $\bar{\mathbf{s}}_1{\triangleq}\tilde{\mathbf{s}}_1$ and $\bar{\mathbf{s}}_2{\triangleq}\tilde{\mathbf{s}}_2$ as the discretization type of the transmitted signal, and $\bar{\mathbf{y}}_1$ and $\bar{\mathbf{y}}_2$ as the discretization type of the received signal to capture the effect of noise, we have
\begin{IEEEeqnarray}{rcl}
\bar{\mathbf{y}}_1{=}\bar{\mathbf{s}}_1{,}&\quad&
\bar{\mathbf{y}}_2{=}\lfloor\mathbf{G}_2\bar{\mathbf{s}}_1\rfloor{+}\bar{\mathbf{s}}_2{.}
\end{IEEEeqnarray}

Then, enhancing user $1$ with the message of user $2$, we have
\begin{IEEEeqnarray}{rcl}
nR_1&{\leq}&I(W_1{;}\bar{\mathbf{y}}_1^n{|}W_2{,}\mathbf{G}_2)\nonumber\\
&{=}&H(\bar{\mathbf{y}}_1^n{|}W_2{,}\mathbf{G}_2){+}o(\log P){,}\\
nR_2&{\leq}&I(W_2{;}\bar{\mathbf{y}}_2^n{|}\mathbf{G}_2)\nonumber\\
&{\leq}&N_2\log P{-}H(\bar{\mathbf{y}}_2^n{|}W_2{,}\mathbf{G}_2){,}\label{eq:nR2_sum}\\
nR_1{+}nR_2&{\leq}&nN_2\log P{+}H(\bar{\mathbf{y}}_1^n{|}W_2{,}\mathbf{G}_2){-}\nonumber\\
&&H(\bar{\mathbf{y}}_2^n{|}W_2{,}\mathbf{G}_2)\\
&{\leq}&nN_2\log P{+}H(\bar{\mathbf{y}}_1^n{,}\bar{\mathbf{y}}_2^n{|}W_2{,}\mathbf{G}_2){-}\nonumber\\
&&H(\bar{\mathbf{y}}_2^n{|}W_2{,}\mathbf{G}_2)\\
&{=}&nN_2\log P{+}H(\bar{\mathbf{y}}_1^n{|}\bar{\mathbf{y}}_2^n{,}W_2{,}\mathbf{G}_2)\\
&{\leq}&nN_2\log P{+}\sum_{i{=}1}^{N_1}H(\bar{s}_{1{,}i}^n{|}\bar{y}_{2{,}i}^n{,}\mathbf{G}_2)\label{eq:setcardi}
\end{IEEEeqnarray}

\textbf{Step 2: } Functional dependence and aligned image set.\footnote{The code block length $n$ is omitted in step 2 and 3 for convenience.}

For a given channel realization, there are multiple vectors $[\bar{s}_{1{,}1}{,}\cdots{,}\bar{s}_{1{,}N_1}{,}\bar{s}_{2{,}i}]$ that cast the same image in $\bar{y}_{2{,}i}$. Thus, the mapping from $\bar{s}_{1{,}i}$ to $[\bar{s}_{1{,}1}{,}\cdots{,}\bar{s}_{1{,}i{-}1}{,}\bar{s}_{1{,}i{+}1}{,}\cdots{,}\bar{s}_{1{,}N_1}{,}\bar{s}_{2{,}i}]$ are random. We fix the minimum mapping that leads to the smallest number of images in the following discussion.

Consequently, the observation $\bar{y}_{2{,}i}$ can be expressed as a function of $\bar{s}_{1{,}i}$, i.e., $\bar{y}_{2{,}i}(\bar{s}_{1{,}i}{,}\mathbf{G}_2)$. With this notation, let us define the aligned image set as the set of all $s_{1{,}i}$ that have the same image in $\bar{y}_{2{,}i}$, i.e.,
\begin{equation}
  \mathcal{S}_v(\mathbf{G}_2){\triangleq}\left\{x{\in}\{s_{1{,}i}\}{:} \bar{y}_{2{,}i}(x{,}\mathbf{G}_2){=}\bar{y}_{2{,}i}(v{,}\mathbf{G}_2)\right\}{.}
\end{equation}
Then, following the derivation in \cite{Davoodi14}, \eqref{eq:setcardi} is bounded by
\begin{equation}
  H(\bar{s}_{1{,}i}{|}\bar{y}_{2{,}i}{,}\mathbf{G}_2){\leq}\log \mathbb{E}\left[|\mathcal{S}_{\bar{s}_{1{,}i}}(\mathbf{G}_2)|\right],
\end{equation}
where $|\mathcal{S}_{\bar{s}_{1{,}i}}(\mathbf{G}_2)|$ is the cardinality of $\mathcal{S}_{\bar{s}_{1{,}i}}(\mathbf{G}_2)$.

\textbf{Step 3:} Bounding the probability that two realizations of $\bar{s}_{1{,}i}$ provide the same image in $\bar{y}_{2{,}i}$.

Let us consider two realization of $\bar{s}_{1{,}i}$, e.g., $x$ and $x^\prime$, which map to $\left\{v_{1{,}j{:}{\forall}j{\neq}i}{,}u\right\}$ and $\left\{v_{1{,}j{:}{\forall}j{\neq}i}^\prime{,}u^\prime\right\}$, respectively. Then, if they produce the same image in $\bar{y}_{2{,}i}$, we have \eqref{eq:image} at the top of the next page.
\begin{figure*}
\begin{IEEEeqnarray}{rcl} \label{eq:image}
  \lfloor g_ix\rfloor{+}\sum_{j{=}1{,}j{\neq}i}^{N_1}\lfloor g_jv_{1{,}j}\rfloor{+}u&{=}& \lfloor g_ix^\prime\rfloor{+}\sum_{j{=}1{,}j{\neq}i}^{N_1}\lfloor g_jv_{1{,}j}^\prime\rfloor{+}u^\prime\\
  \Rightarrow g_i(x{-}x^\prime)&{\in}&\sum_{j{=}1{,}j{\neq}i}^{N_1}\lfloor g_jv_{1{,}j}^\prime\rfloor{-}\lfloor g_iv_{1{,}j}\rfloor{+}u^\prime{-}u{+}(-1{,}1).\\
  \text{\rm and }\Rightarrow g_l(v_l{-}v_l^\prime)&{\in}&\sum_{j{=}1{,}j{\neq}i{,}j{\neq}l}^{N_1}\lfloor g_jv_{1{,}j}^\prime\rfloor{-}\lfloor g_jv_{1{,}j}\rfloor{+}\lfloor g_ix^\prime\rfloor{-}\lfloor g_ix\rfloor {+} u^\prime{-}u{+}(-1{,}1){,}{\forall}l{\neq}i.
\end{IEEEeqnarray}
\hrulefill
\end{figure*}
Next, let us define
\begin{equation}
  L{\triangleq}\max_{{\forall}l{\neq}i}\{|v_l{-}v_l^\prime|{,}|x{-}x^\prime|\}{.}
\end{equation}
Hence, the value of $g_j$, $j{=}1{,}\cdots{,}N_1$, must lie within the interval of length no more than $\frac{2}{L}$. Therefore, the probability that the images due to $x$ and $x^\prime$ align at $\bar{y}_{2{,}i}$ is bounded as follows
\begin{equation}
  \mathbb{P}\left(x{\in}\mathcal{S}_{x^\prime}(G)\right){\leq}f_{\max{,}2}^n\prod_{t{=}1}^{n}\frac{2}{L(t)}{,}
\end{equation}
where $L$ is a time-varying parameter, and the time index $t$ is omitted in the above derivations for simplicity. Moreover, $f_{\max{,}2}{=}O(P^{\alpha_2})$ is a function of the CSIT quality defined in Section \ref{sec:SM}. Consequently, $H(\bar{s}_{1{,}i}^n{|}\bar{y}_{2{,}i}^n{,}\mathbf{G}_2)$ is bounded by $n\alpha_2\log P$. This leads to the sum DoF constraint $d_1{+}d_2{\leq}N_2{+}N_1\alpha_2$.

For the weighted-sum inequality, the derivation only differs by the first step. Specifically, let us write a canonical form by switching the role of user $1$ and user $2$ as
\begin{equation}
  \left[\begin{array}{c}
          \tilde{\mathbf{y}}_2 \\
          \tilde{\mathbf{y}}_1
        \end{array}\right]{=}\left[\begin{array}{cc}
                                     \mathbf{I}_{N_2} & \mathbf{0}_{N_2{\times}N_1} \\
                                     \mathbf{G}_1 & \mathbf{I}_{N_1}
                                   \end{array}\right]\left[\begin{array}{c}
                                                             \tilde{\mathbf{s}}_2 \\
                                                             \tilde{\mathbf{s}}_1
                                                           \end{array}\right]{+}\left[\begin{array}{c}
                                                             \mathbf{n}_1 \\
                                                             \mathbf{n}_2
                                                           \end{array}\right],
\end{equation}
where $\mathbf{G}_1{=}\mathbf{H}_{12}^H\mathbf{H}_{22}^{-H}$, $\tilde{\mathbf{s}}_2{=}\mathbf{H}_{22}^H\mathbf{s}_2{+}\mathbf{H}_{21}^H\mathbf{s}_1$ and $\tilde{\mathbf{s}}_1{=}\left(\mathbf{H}_{11}^H{-}\mathbf{G}_1\mathbf{H}_{21}^H\right)\mathbf{s}_1$.

Then, denoting $\bar{\mathbf{y}}_1$ and $\bar{\mathbf{y}}_2$ as the discretization type of the received signal to capture the effect of noise, and denoting $\bar{\mathbf{s}}_1$ and $\bar{\mathbf{s}}_2$ as discretization type of the transmitted signal, we have
\begin{IEEEeqnarray}{rcl}
\bar{\mathbf{y}}_2{=}\bar{\mathbf{s}}_2{,}&\quad&
\bar{\mathbf{y}}_1{=}\lfloor\mathbf{G}_1\bar{\mathbf{s}}_2\rfloor{+}\bar{\mathbf{s}}_1{.}
\end{IEEEeqnarray}

Then, enhancing user $2$ with the message of user $1$, we have
\begin{IEEEeqnarray}{rcl}
nR_1&{\leq}&I(W_1{;}\bar{\mathbf{y}}_1^n{|}\mathbf{G}_1)\nonumber\\
&{=}&nN_1\log P{-}H(\bar{\mathbf{y}}_1^n{|}W_1{,}\mathbf{G}_1){,}\\
nR_2&{\leq}&I(W_2{;}\bar{\mathbf{y}}_2^n{|}W_1{,}\mathbf{G}_2)\nonumber\\
&{=}&H(\bar{\mathbf{y}}_2^n{|}W_1{,}\mathbf{G}_1){-}o(\log P){,}\\
n(&N_2&R_1{+}N_1R_2)\nonumber\\
&{\leq}&nN_1N_2\log P{+}N_1H(\bar{\mathbf{y}}_2^n{|}W_1{,}\mathbf{G}_1){-}\nonumber\\
&&N_2H(\bar{\mathbf{y}}_1^n{|}W_1{,}\mathbf{G}_1)\\
&{\leq}&nN_1N_2\log P{+} \sum_{j{=}1}^{N_2}\left(H(\bar{\mathbf{y}}_{2{,}j{:}j{+}N_1{-}1}^n{|}W_1{,}\mathbf{G}_1){-}\right.\nonumber\\
&&\left.H(\bar{\mathbf{y}}_1^n{|}W_1{,}\mathbf{G}_1)\right) \label{eq:sliding}\\
&{\leq}&nN_1N_2\log P{+} \sum_{j{=}1}^{N_2}\left(H(\bar{\mathbf{y}}_{2{,}j{:}j{+}N_1{-}1}^n{,}\bar{\mathbf{y}}_1^n{|}W_1{,}\mathbf{G}_1){-}\right.\nonumber\\ 
&& \left. H(\bar{\mathbf{y}}_1^n{|}W_1{,}\mathbf{G}_1)\right)\\
&{=}&nN_1N_2\log P{+}\sum_{j{=}1}^{N_2}H(\bar{\mathbf{y}}_{2{,}j{:}j{+}N_1{-}1}^n{|}\bar{\mathbf{y}}_1^n{,}W_1{,}\mathbf{G}_1)\\
&{\leq}&nN_1N_2\log P{+}\sum_{j{=}1}^{N_2}\sum_{i{=}1}^{N_1}H(\bar{s}_{2{,}j{+}i{-}1}^n{|}\bar{y}_{1{,}i}^n{,}\mathbf{G}_1).\label{eq:setcardi_ws}
\end{IEEEeqnarray}
Inequality \eqref{eq:sliding} follows from the sliding window lemma introduced in \cite[Lemma 1]{Borzoo_Kuser}. The notation $\bar{\mathbf{y}}_{2{,}j{:}j{+}N_1{-}1}$ stand for the $j$th through to the $(j{+}N_1{-}1)$th entries of $\bar{\mathbf{y}}_2$, and the calculation $j{+}N_1{-}1$ is based on modulo $N_2$.

Following step 2 and step 3, one can show that $H(\bar{s}_{2{,}j{+}i{-}1}^n{|}\bar{y}_{1{,}i}^n{,}\mathbf{G}_1){\leq}n\alpha_1\log P$, which leads to the weighted sum DoF $\frac{d_1}{N_1}{+}\frac{d_2}{N_2}{\leq}1{+}\alpha_1$.

\subsubsection{$N_2{\leq}M{\leq}N_1{+}N_2$}
In this case, the linear space spanned by the channel matrices of the two users overlap with each other, and the dimension of the overlapping part is $N_1{+}N_2{-}M$. Hence, we perform a linear transformation to the received signals as follows
\begin{IEEEeqnarray}{rcl}
&\mathbf{F}_1\mathbf{y}_1{=}\hat{\mathbf{y}}_1{=}\left[\begin{array}{c}
                                   \hat{\mathbf{y}}_{1a} \\
                                   \hat{\mathbf{y}}_{1b}
                                 \end{array}\right]{=}\left[\begin{array}{c}
                                                              \mathbf{F}_{1a}\mathbf{H}_1^H\mathbf{s}{+} \mathbf{F}_{1a}\mathbf{n}_1 \\
                                                              \mathbf{F}_{1b}\mathbf{H}_1^H\mathbf{s}{+} \mathbf{F}_{1b}\mathbf{n}_1
                                                            \end{array}\right]{,}&\IEEEyessubnumber\\
&\mathbf{F}_2\mathbf{y}_2{=}\hat{\mathbf{y}}_2{=}\left[\begin{array}{c}
                                   \hat{\mathbf{y}}_{2a} \\
                                   \hat{\mathbf{y}}_{2b}
                                 \end{array}\right]{=}\left[\begin{array}{c}
                                                              \mathbf{F}_{2a}\mathbf{H}_2^H\mathbf{s}{+} \mathbf{F}_{2a}\mathbf{n}_2 \\
                                                              \mathbf{F}_{2b}\mathbf{H}_2^H\mathbf{s}{+} \mathbf{F}_{2b}\mathbf{n}_2
                                                            \end{array}\right]{,}&\IEEEyessubnumber
\end{IEEEeqnarray}
where $\mathbf{F}_k$ is a $N_k{\times}N_k$ full rank matrix, $\mathbf{F}_{1a}$ and $\mathbf{F}_{2a}$ are the first $M{-}N_2$ rows of $\mathbf{F}_1$ and the first $M{-}N_1$ rows of $\mathbf{F}_2$, respectively, while $\mathbf{F}_{1b}$ and $\mathbf{F}_{2b}$ are the remaining $N_1{+}N_2{-}M$ rows of $\mathbf{F}_1$ and $\mathbf{F}_2$, respectively. $\mathbf{F}_{1b}$ and $\mathbf{F}_{2b}$ are such that $\mathbf{F}_{1b}\mathbf{H}_1^H{=}\mathbf{F}_{2b}\mathbf{H}_2^H$. This means that $\hat{\mathbf{y}}_{1b}$ can be obtained using $\hat{\mathbf{y}}_{2b}$ within noise error.

Consequently, one can obtain a canonical form using $\hat{\mathbf{y}}_1$ and $\hat{\mathbf{y}}_2$ as
\begin{equation}
  \left[\begin{array}{c}
          \tilde{\mathbf{y}}_1 \\
          \tilde{\mathbf{y}}_2
        \end{array}\right]{=}\left[\begin{array}{cc}
                                     \mathbf{I}_{N_1} & \mathbf{0}_{N_1{\times}(M{-}N_1)} \\
                                     \mathbf{G}_2 & \mathbf{Z}_2
                                   \end{array}\right]\left[\begin{array}{c}
                                                             \tilde{\mathbf{s}}_1 \\
                                                             \tilde{\mathbf{s}}_2
                                                           \end{array}\right]{+}\left[\begin{array}{c}
                                                             \mathbf{n}_1 \\
                                                             \mathbf{n}_2
                                                           \end{array}\right]\label{eq:canonical_user1_case2},
\end{equation}
where $\mathbf{Z}_2{\triangleq}\left[\begin{array}{c}\mathbf{I}_{M{-}N_1}\\ \mathbf{0}_{(N_2{+}N_1{-}M){\times}(M{-}N_1)}\end{array}\right]$, $\mathbf{G}_2{=}\text{Bdiag}\{\mathbf{G}_{2a}{,} \mathbf{I}_{N_1{+}N_2{-}M}\}$ with a $(M{-}N_1){\times}(M{-}N_2)$ matrix $\mathbf{G}_{2a}{=}\mathbf{F}_{2a}\mathbf{H}_{21}^H\mathbf{H}_{11}\mathbf{F}_{1a}^H\cdot \left(\mathbf{F}_{1a}\mathbf{H}_{11}^H\mathbf{H}_{11}\mathbf{F}_{1a}^H\right)^{-1}$, $\tilde{\mathbf{s}}_1{=}\hat{\mathbf{y}}_1$ and $\tilde{\mathbf{s}}_2{=}(\mathbf{F}_{2a}\mathbf{H}_{22}^H{-}\mathbf{G}_{2a}\mathbf{F}_{1a}\mathbf{H}_{12}^H)\mathbf{s}_2$. Note that here $\mathbf{H}_{11}$ and $\mathbf{H}_{21}$ refer to the channel matrices between the first $N_1$ transmit antennas to user $1$ and user $2$, respectively, while $\mathbf{H}_{12}$ and $\mathbf{H}_{22}$ refer to the channel matrices between the remaining $M{-}N_1$ transmit antennas to user $1$ and user $2$, respectively. $\mathbf{s}_1$ and $\mathbf{s}_2$ are the signals transmitted from the first $N_1$ transmit antennas and the remaining $M{-}N_1$ transmit antennas, respectively.

Then, following the footsteps in the case $M{=}N_1{+}N_2$, we bound the sum rate by the summation of $N_1$ conditional entropies as in \eqref{eq:setcardi}. According to the above analysis, since the last $N_1{+}N_2{-}M$ observations of $\tilde{\mathbf{y}}_1$ can be constructed using the $N_1{+}N_2{-}M$ observations of $\tilde{\mathbf{y}}_2$, the last $N_1{+}N_2{-}M$ entropies are equal to $o(\log P)$. This leads to the sum DoF constraint $d_1{+}d_2{\leq}N_2{+}(M{-}N_2)\alpha_2$.

Similarly, for the weighted sum entropy, we switch the role of the two users and write a canonical form as
\begin{equation}
  \left[\begin{array}{c}
          \tilde{\mathbf{y}}_2 \\
          \tilde{\mathbf{y}}_1
        \end{array}\right]{=}\left[\begin{array}{cc}
                                     \mathbf{I}_{N_2} & \mathbf{0}_{N_2{\times}(M{-}N_2)} \\
                                     \tilde{\mathbf{G}}_1 & \mathbf{Z}_1
                                   \end{array}\right]\left[\begin{array}{c}
                                                             \tilde{\mathbf{s}}_2 \\
                                                             \tilde{\mathbf{s}}_1
                                                           \end{array}\right]{+}\left[\begin{array}{c}
                                                             \mathbf{n}_2 \\
                                                             \mathbf{n}_1
                                                           \end{array}\right]\label{eq:canonical_user2_case2},
\end{equation}
where $\mathbf{Z}_1{\triangleq}\left[\begin{array}{c}\mathbf{I}_{M{-}N_2}\\ \mathbf{0}_{(N_2{+}N_1{-}M){\times}(M{-}N_2)}\end{array}\right]$,  $\mathbf{G}_1{=}\text{Bdiag}\{\mathbf{G}_{1a}{,} \mathbf{I}_{N_1{+}N_2{-}M}\}$ with a $(M{-}N_2){\times}(M{-}N_1)$ matrix $\mathbf{G}_{1a}{=}\mathbf{F}_{1a}\mathbf{H}_{12}^H\mathbf{H}_{22}\mathbf{F}_{2a}^H\cdot \left(\mathbf{F}_{2a}\mathbf{H}_{22}^H\mathbf{H}_{22}\mathbf{F}_{2a}^H\right)^{-1}$, $\tilde{\mathbf{s}}_2{=}\hat{\mathbf{y}}_2$ and $\tilde{\mathbf{s}}_1{=}(\mathbf{F}_{1a}\mathbf{H}_{11}^H{-}\mathbf{G}_{1a}\mathbf{F}_{2a}\mathbf{H}_{21}^H)\mathbf{s}_1$. Note that here $\mathbf{H}_{21}$ and $\mathbf{H}_{11}$ refer to the channel matrices between the first $M{-}N_2$ transmit antennas to user $1$ and user $2$, respectively, while $\mathbf{H}_{22}$ and $\mathbf{H}_{12}$ refer to the channel matrices between the remaining $N_2$ transmit antennas to user $1$ and user $2$, respectively. $\mathbf{s}_1$ and $\mathbf{s}_2$ are the signals transmitted from the first $M{-}N_2$ antennas and the remaining $N_2$ antennas, respectively.

Then, following the footsteps in the case $M{=}N_1{+}N_2$, we bound the weighted sum of the rate as
\begin{IEEEeqnarray}{rcl}
n(&N_2&R_1{+}N_1R_2)\nonumber\\
&{\leq}&nN_1N_2\log P{+}\sum_{j{=}1}^{N_2}H(\bar{\mathbf{y}}_{2{,}j{:}j{+}N_1{-}1}^n{|}\bar{\mathbf{y}}_1^n{,}W_2{,}\mathbf{G}_1)\\
&{\leq}&nN_1N_2\log P{+}\sum_{j{=}1}^{N_2}\sum_{i{=}1}^{N_1}H(\bar{s}_{2{,}j{+}i{-}1}^n{|}\bar{\mathbf{y}}_1^n{,}\mathbf{G}_1)\\
&{=}&nN_1N_2\log P{+}N_1\sum_{j{=}1}^{N_2}H(\bar{s}_{2{,}j}^n{|}\bar{\mathbf{y}}_1^n{,}\mathbf{G}_1){,}\label{eq:conditional_ws_case2}
\end{IEEEeqnarray}
where the last equality is because every observation of $\bar{\mathbf{s}}_2$ is counted $N_1$ times due to the sliding window. According to the above analysis, since the last $N_1{+}N_2{-}M$ observations of $\tilde{\mathbf{y}}_2$ (i.e., $\bar{\mathbf{s}}_2$) can be constructed using the $N_1{+}N_2{-}M$ observations of $\tilde{\mathbf{y}}_1$, the last $N_1{+}N_2{-}M$ entropies are equal to $o(\log P)$. This upper-bounds \eqref{eq:conditional_ws_case2} by $nN_1N_2\log P{+}N_1(M{-}N_1)\alpha_1\log P$, which leads to the weighted sum DoF constraint \eqref{eq:outer_wsum}.

\subsubsection{$N_1{\leq}M{\leq}N_2$}
In this case, the derivation follows the footsteps of the case $M{<}N_1{+}N_2$. Specifically, since $M{<}N_2$, \eqref{eq:nR2_sum} rewrites as $nR_2{\leq}nM{\log}P{-}H(\bar{\mathbf{y}}_2^n{|}W_1{,}\mathbf{G}_2)$. Besides, since $M{<}N_2$, the dimension of the overlapping part between the received signals at the two users is $N_1$. This implies that the $N_1$ observations at user $1$ can be constructed using the user 2's received signal within noise error. Hence, the $N_1$ conditional entropies in \eqref{eq:setcardi} equal to $o(\log P)$, leading to the sum DoF $d_1{+}d_2{\leq}M$.

For the weighted-sum inequality, one has
\begin{IEEEeqnarray}{rcl}
n(MR_1&{+}&N_1R_2)\nonumber\\
&{\leq}&nN_1M\log P{+}N_1H(\bar{\mathbf{y}}_2^n{|}W_1{,}\mathbf{G}_1){-}\nonumber\\
&&MH(\bar{\mathbf{y}}_1^n{|}W_1{,}\mathbf{G}_1)\\
&{=}&nN_1M\log P{+}N_1H(\bar{\mathbf{y}}_{2{,}1{:}M}^n{|}W_1{,}\mathbf{G}_1){+} \nonumber\\ &&\underbrace{N_1H(\bar{\mathbf{y}}_{2{,}M{+}1{:}N_2}^n{|}\bar{\mathbf{y}}_{2{,}1{:}M}^n{,}W_1{,}\mathbf{G}_1)}_{nN_1o(\log P)}{-}\nonumber\\ &&MH(\bar{\mathbf{y}}_1^n{|}W_1{,}\mathbf{G}_1).\label{eq:difference_entropy_case3}
\end{IEEEeqnarray}

Then, since $N_1$ observations of $\bar{\mathbf{y}}_{2{,}1{:}M}$ can be constructed using $\bar{\mathbf{y}}_1$, the difference between the entropies in \eqref{eq:difference_entropy_case3} is bounded by $nN_1(M{-}N_1)\alpha_1{\log}P$, which completes the proof. 

\bibliographystyle{IEEEtran}

\bibliography{MIMOBCjnl}

\end{document}